\documentclass[11pt]{article}
\usepackage[latin1]{inputenc}   
\usepackage{amsmath}   
\usepackage{amsfonts}  
\usepackage{amssymb}  
\usepackage{latexsym}  
\usepackage{graphicx} 
\usepackage[a4paper]{geometry}
\usepackage{color} 
\usepackage{subfigure}
\usepackage{url}
\usepackage{minitoc}
\usepackage{longtable}

\newtheorem{theorem}{Theorem}

\newtheorem{lemma}{Lemma}
\newtheorem{claim}{Claim}
\newtheorem{definition}{Definition}

\newenvironment{proof}{\noindent \emph{Proof.}\ }{\hfill
    $\Box$\vspace{1em}}

\def\claim{$$\vcenter\bgroup\advance\hsize by -8em\noindent
\refstepcounter{claim}\ignorespaces\it}	    
\makeatletter
\def\endclaim{\rm\egroup\leqno(\theclaim)$$\global\@ignoretrue}
\makeatother

\newenvironment{proofclaim}[1][]%
	{\noindent {}{#1}{}}{ This proves claim~(\arabic{claim}).\hfill
    $\Diamond$\vspace{1em}}

  \newcommand{\md}{\mathcal D} \newcommand{\mv}{\mathcal V}
  \newcommand{\ms}{\mathcal S_\mv} \newcommand{\mr}{\mathbb R}

\newcommand{\mz}{\mathbb Z}
   
  \title{Toroidal maps : Schnyder woods, orthogonal surfaces and
    straight-line representations\thanks{This work was partially
      supported by the ANR grant GRATOS ANR-09-JCJC-0041 and ANR grant
      EGOS ANR-12-JCJC-xxxx.}}

  \author{Daniel Gon\c{c}alves, Benjamin L\'ev\^eque\\ \ \\
    LIRMM, CNRS, Universit\'e Montpellier 2\\
    161 rue Ada - 34095 Montpellier Cedex 5 France}

\begin{document}
\maketitle

\begin{abstract}
  A Schnyder wood is an orientation and coloring of the edges of a
  planar map satisfying a simple local property.  We propose a
  generalization of Schnyder woods to graphs embedded on the torus
  with application to graph drawing.  We prove several properties on
  this new object. Among all we prove that a graph embedded on the
  torus admits such a Schnyder wood if and only if it is an
  essentially 3-connected toroidal map. We show that these Schnyder
  woods can be used to embed the universal cover of an essentially
  3-connected toroidal map on an infinite and periodic orthogonal
  surface. Finally we use this embedding to obtain a straight-line
  flat torus representation of any toroidal map in a polynomial size
  grid.
\end{abstract}

\section{Introduction}

A closed curve on a surface is \emph{contractible} if it can be
continuously transformed into a single point.  Given a graph embedded
on the torus, a \emph{contractible loop} is an edge forming a
contractible cycle. Two \emph{homotopic multiple edges} are two edges
with the same extremities such that their union forms a contractible
cycle.  In this paper, we will almost always consider graphs embedded
on the torus with no contractible loop and no homotopic multiple
edges. We call these graphs \emph{toroidal graphs} for short and keep
the distinction with \emph{graph embedded on the torus} that may have
contractible loops or homotopic multiple edges.  A \emph{map} on a
surface is a graph embedded on this surface where every face is
homeomorphic to an open disk. A \emph{map embedded on the torus} is a
graph embedded on the torus that is a map (it may contains
contractible loops or homotopic multiple edges). A \emph{toroidal map}
is a toroidal graph that is a map (it has no contractible loop and no
homotopic multiple edges).  A \emph{toroidal triangulation} is a
toroidal map where every face has size three.  A general graph
(i.e. not embedded on a surface) is \emph{simple} if it contains no
loop and no multiple edges. Since some loops and multiple edges are
allowed in toroidal graphs, the class of toroidal graphs is larger
than the class of simple toroidal graphs.

The torus is represented by a parallelogram in the plane whose
opposite sides are pairwise identified.  This representation is called
the \emph{flat torus}. The \emph{universal cover} $G^\infty$ of a
graph $G$ embedded on the torus is the infinite planar graph obtained
by replicating a flat torus representation of $G$ to tile the plane
(the tiling is obtained by translating the flat torus along two
vectors corresponding to the sides of the parallelogram).  Note that a
graph $G$ embedded on the torus has no contractible loop and no
homotopic multiple edges if and only if $G^\infty$ is simple.

Given a general graph $G$, let $n$ be the number of vertices and $m$
the number of edges. Given a graph embedded on a surface, let $f$ be the
number of faces.  Euler's formula says that any  map on
 a surface of genus $g$ satisfies $n-m+f=2-2g$. Where the
plane is the surface of genus $0$, and  the torus the surface of
genus $1$.

Schnyder woods where originally defined for planar triangulations by
Schnyder~\cite{Sch89}:

\begin{definition}[Schnyder wood, Schnyder property]
\label{def:schnyder}
Given a planar triangulation $G$, a \emph{Schnyder wood} is an
orientation and coloring of the edges of $G$ with the colors $0$, $1$,
$2$ where each inner vertex $v$ satisfies the  \emph{Schnyder
  property}, (see Figure~\ref{fig:LSP} where each color is represented
by a different type of arrow):

\begin{itemize}
\item Vertex $v$ has out-degree one in each color.
\item The edges $e_0(v)$, $e_1(v)$, $e_2(v)$ leaving $v$ in colors
  $0$, $1$, $2$, respectively, occur in counterclockwise order.
\item Each edge entering $v$ in color $i$ enters $v$ in the
  counterclockwise sector from $e_{i+1}(v)$ to $e_{i-1}(v)$ (where
  $i+1$ and $i-1$ are understood modulo $3$).
\end{itemize}
\end{definition}

\begin{figure}[!h]
\center
\includegraphics[scale=0.5]{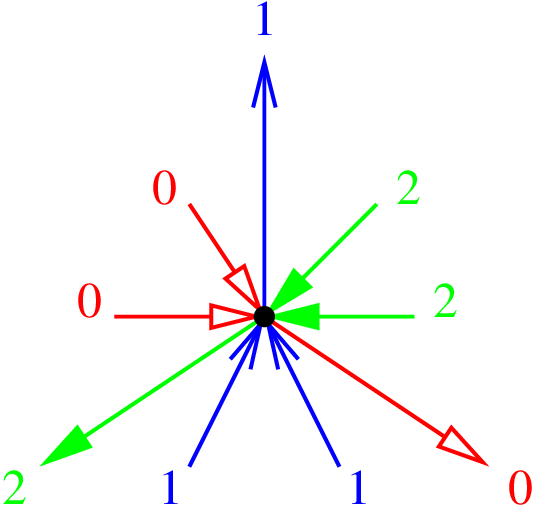}
\caption{Schnyder property}
\label{fig:LSP}
\end{figure}

For higher genus triangulated surfaces, a generalization of Schnyder
wood has been proposed by Castelli Aleardi et al.~\cite{CFL09}, with
applications to encoding.  Unfortunately, in this definition, the
simplicity and the symmetry of the original Schnyder wood are
lost. Here we propose an alternative generalization of Schnyder woods
for toroidal graphs, with application to graph drawings.

By Euler's formula, a planar triangulation satisfies $m=3n-6$. Thus
there is not enough edges in the graph for all vertices to be of
out-degree three. This explain why just some vertices (inner ones) are
required to verify the Schnyder property.  For a toroidal
triangulation, Euler's formula gives exactly $m=3n$ so there is hope
for a nice object satisfying the Schnyder property for every
vertex. This paper shows that such an object exists. Here we do not
restrict ourselves to triangulations and we directly define Schnyder
woods in a more general framework.

Felsner~\cite{Fel01, Fel03} (see also \cite{Mil02}) has generalized
Schnyder woods to 3-connected planar maps by allowing edges to be
oriented in one direction or in two opposite directions.  We also
allow edges to be oriented in two directions in our definition:

\begin{definition}[Toroidal Schnyder wood]
  Given a toroidal graph $G$, a \emph{(toroidal) Schnyder wood} of $G$
  is an orientation and coloring of the edges of $G$ with the colors
  $0$, $1$, $2$, where every edge $e$ is oriented in one direction or
  in two opposite directions (each direction having a distinct color),
  satisfying the following (see example of
  Figure~\ref{fig:example-schnyder}):

\begin{itemize}
\item[(T1)] Every vertex $v$ satisfies the Schnyder property (see Definition~\ref{def:schnyder})
\item[(T2)] Every monochromatic cycle of color $i$ intersects at least one
  monochromatic cycle of color $i-1$ and at least one monochromatic
  cycle of color $i+1$.
\end{itemize}
\end{definition}

\begin{figure}[h!]
\center
\includegraphics[scale=0.5]{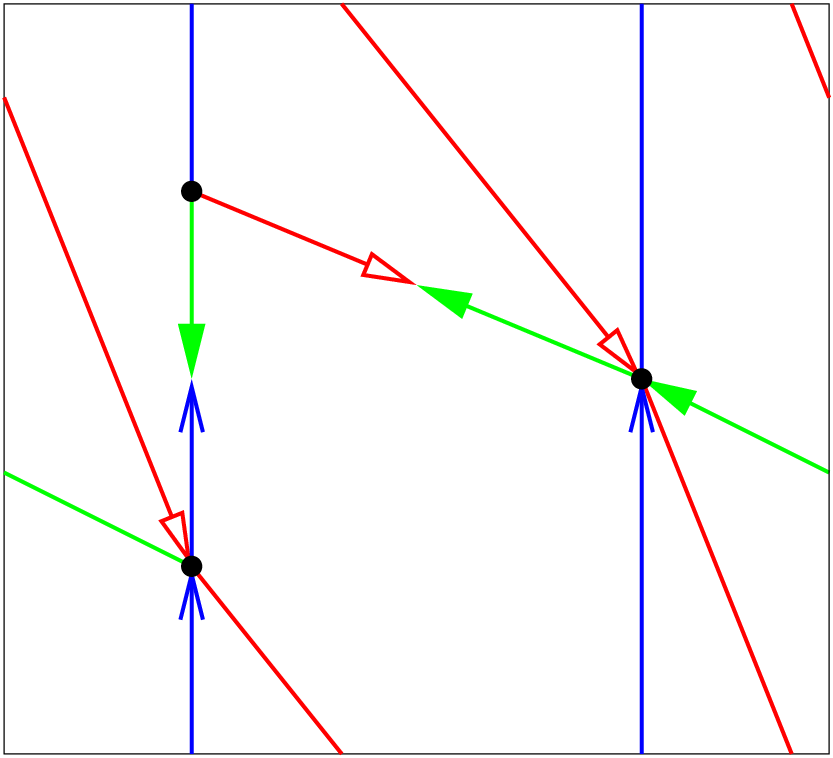}
\caption{Example of a Schnyder wood of a toroidal graph.}
\label{fig:example-schnyder}
\end{figure}

In the case of toroidal triangulations, $m=3n$ implies that there are
too many edges to have bi-oriented edges. Thus, we can use this
general definition of Schnyder wood for toroidal graphs and keep in mind
that when restricted to toroidal triangulations all edges are oriented
in one direction only.

Extending the notion of essentially 2-connectedness \cite{MR98}, we
say that a toroidal graph $G$ is \emph{essentially k-connected} if its
universal cover is k-connected. Note that an essentially 1-connected
toroidal graph is a toroidal map.  We prove that essentially
3-connected toroidal maps are characterized by existence of Schnyder
woods.

\begin{theorem}
\label{th:existence}
A toroidal graph admits a Schnyder wood if and only if it is an
essentially 3-connected toroidal map.
\end{theorem}

In our definition of Schnyder woods, two properties are required : a
local one (T1) and a global one (T2). This second property is important
to use Schnyder woods to embed toroidal graphs on orthogonal surfaces
like it has been done in the plane by Miller~\cite{Mil02} (see also
\cite{Fel03}).

\begin{theorem}
\label{th:triortho}
The universal cover of an essentially 3-connected toroidal map admits
a geodesic embedding on an infinite and periodic orthogonal surface.
\end{theorem}

A \emph{straight-line flat torus representation} of a toroidal map
$G$ is the restriction to a flat torus of a periodic straight-line
representation of $G^{\infty}$. 
The problem to find a straight
line flat torus representation of a toroidal map was previously solved
on exponential size grids~\cite{Moh96}. There are several works to
represent a toroidal map inside a parallelogram in a polynomial size
grid~\cite{CEG11,DGK11}, but in these representations the opposite
sides of the parallelogram do not perfectly match.  In the embeddings
obtained by Theorem~\ref{th:triortho}, vertices are not coplanar but
we prove that for toroidal triangulations one can project the vertices
on a plane to obtain a periodic straight-line representation of
$G^{\infty}$. This gives the first straight-line flat torus
representation of any toroidal map in a polynomial size grid.  


\begin{theorem}
\label{cor:straightline}
A toroidal graph admits a straight-line flat torus representation in a
polynomial size grid.
\end{theorem}

In Section~\ref{sec:schnyderwoods}, we explain how our definition of
Schnyder woods in the torus generalize the planar case. In
Section~\ref{sec:properties}, we show that our Schnyder woods are of
two fundamentally different types.  In Section~\ref{sec:universal}, we
study the behavior of Schnyder woods in the universal cover, we define
the notion of regions and show that the existence of Schnyder woods
for a toroidal graph implies that the graph is an essentially
3-connected toroidal map.
In Section~\ref{sec:dual}, we define the dual of a Schnyder wood.  In
Section~\ref{sec:relax}, we show how the definition of Schnyder woods
can be relax for one of the two types of Schnyder wood. This
relaxation is used in the next sections for proving existence of
Schnyder wood.  In Section~\ref{sec:existence}, we use a result of
Fijavz~\cite{Fij} on existence of non homotopic cycles in simple
toroidal triangulations to obtain a short proof of existence of
Schnyder wood for simple triangulations.  In Section~\ref{sec:lemma},
we prove a technical lemma showing how a Schnyder wood of a graph $G$
can be derived from a Schnyder wood of the graph $G'$, where $G'$ is
obtained from $G$ by contracting an edge. This lemma is then used in
Section~\ref{sec:existence3connected} to prove the existence of
Schnyder woods for any essentially 3-connected toroidal maps. In
Section~\ref{sec:ortho}, we use Schnyder woods to embed the universal
cover of essentially 3-connected toroidal maps on periodic and
infinite orthogonal surfaces by generalizing the region vector method
defined in the plane. In Section~\ref{sec:dualortho}, we show that the
dual map can also be embedded on this orthogonal surface. In
Section~\ref{sec:straight}, we show that, in the case of
toroidal triangulations, this orthogonal surface can be projected on a plane to
obtain a straight-line flat torus representation.

\section{Generalization of the planar case}
\label{sec:schnyderwoods}

Felsner~\cite{Fel01, Fel03} has generalized planar Schnyder woods by allowing
edges to be oriented in one direction or in two opposite
directions. The formal definition is the following: 

\begin{definition}[Planar Schnyder wood]
  Given a planar map $G$. Let $x_0$, $x_1$, $x_2$ be three distinct
  vertices occurring in counterclockwise order on the outer face of
  $G$. The \emph{suspension} $G^\sigma$ is obtained by attaching a
  half-edge that reaches into the outer face to each of these special
  vertices.  A \emph{(planar) Schnyder wood} rooted at $x_0$, $x_1$, $x_2$ is
  an orientation and coloring of the edges of $G^\sigma$ with the
  colors $0$, $1$, $2$, where every edge $e$ is oriented in one
  direction or in two opposite directions (each direction having a
  distinct color), satisfying the following (see example of
  Figure~\ref{fig:example-planar}):

\begin{itemize}
\item[(P1)] Every vertex $v$ satisfies the Schnyder property and the
  half-edge at $x_i$ is directed outwards and colored $i$
\item[(P2)] There is no monochromatic cycle.
\end{itemize}
\end{definition}

\begin{figure}[!h]
\center
\includegraphics[scale=0.5]{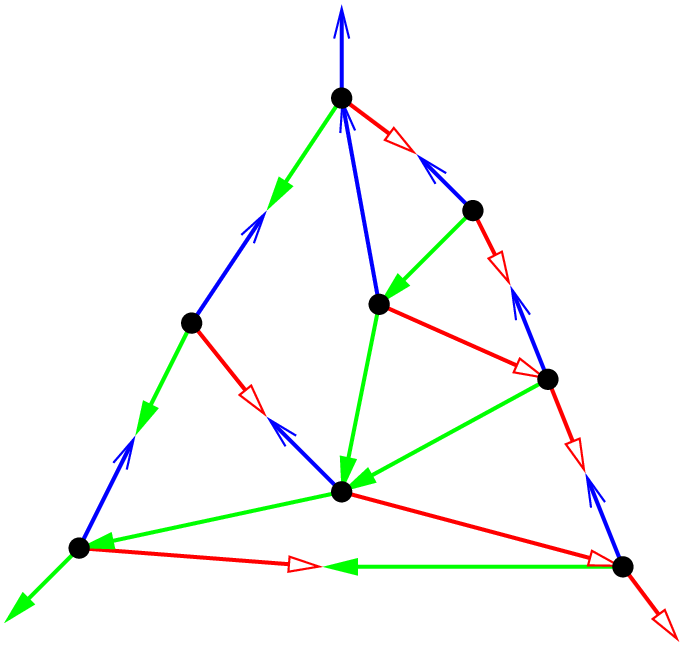}
\caption{Example of a Schnyder wood of a planar map.}
\label{fig:example-planar}
\end{figure}

In the definition given  by Felsner~\cite{Fel03}, property (P2) is
in fact replaced by ``There is no interior face the boundary of which is a
monochromatic cycle'', but the two are equivalent by results
of~\cite{Fel01, Fel03}.

With our definition of Schnyder wood for toroidal graph, the goal is to
generalize the definition of Felsner. In the torus, property (P1) can
be simplified as every vertex plays the same role: there is no special
outer vertices with a half-edge reaching into the outer face. This
explain property (T1) in our definition. Then if one asks that every
vertex satisfies the Schnyder property, there is necessarily
monochromatic cycles and (P2) is not satisfied. This explain why (P2)
has been replaced by (T2) in our generalization to the torus.

 It would have been possible to replace (P2) by ``there is no contractible monochromatic
 cycles'' but this is no enough to suit our needs.  Our goal is to use
 Schnyder woods to embed universal cover of toroidal graphs on
 orthogonal surfaces like it has been done in the plane by
 Miller~\cite{Mil02} (see also \cite{Fel03}). The difference being
  that our surface is infinite and periodic. In such a
 representation the three colors $0$, $1$, $2$ corresponds to the
 three directions of the space. Thus the monochromatic cycles with
 different colors have to intersect each other in a particular
 way. This explains why property (T2) is required.
 Figure~\ref{fig:nocross} gives an example of an orientation and
 coloring of the edges of a toroidal triangulation satisfying (T1) but
 not (T2) as there is no pair of intersecting monochromatic cycles.

\begin{figure}[!h]
\center
\includegraphics[scale=0.5]{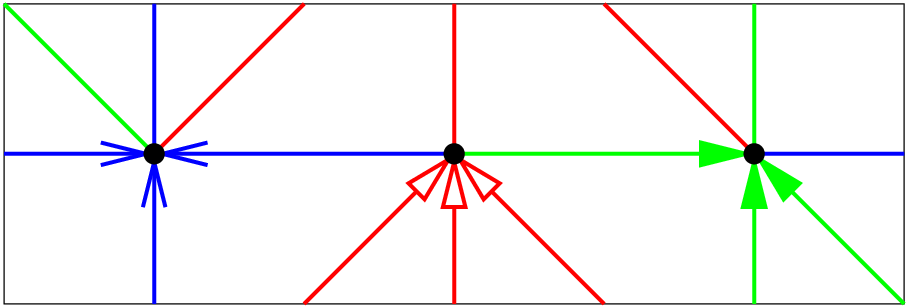}
\caption{An orientation and coloring of the edges of a toroidal
  triangulation satisfying (T1) but not (T2) as there is no pair of
  intersecting monochromatic cycles.}
\label{fig:nocross}
\end{figure}

Let $G$ be a toroidal graph given with a Schnyder wood.  Let $G_i$ be
the directed graph induced by the edges of color $i$. This definition
includes edges that are half-colored $i$, and in this case, the edges
gets only the direction corresponding to color $i$.  Each graph $G_i$
has exactly $n$ edges, so it does not induce a rooted tree (contrarily
to planar Schnyder woods) and the term ``wood'' has to be handle with
care here.  Note also that $G_i$ is not necessarily connected (for
example in the graph of Figure~\ref{fig:notconnected}, every Schnyder
wood has one color which corresponding subgraph is not connected). But
each components of $G_i$ has exactly one outgoing arc for each of its
vertices. Thus each connected component of $G_i$ has exactly one
directed cycle that is a \emph{monochromatic cycle} of color $i$, or
\emph{$i$-cycle} for short. Note that monochromatic cycles can contain
edges oriented in two directions with different colors, but the
\emph{orientation} of a $i$-cycle is the orientation given by the
(half-)edges of color $i$.  The graph $G_i^{-1}$ is the graph obtained
from $G_i$ by reversing all its edges.  The graph $G_i\cup
G_{i-1}^{-1}\cup G_{i+1}^{-1}$ is obtained from the graph $G$ by
orienting edges in one or two direction depending on whether this
orientation is present in $G_i$, $G_{i-1}^{-1}$ or $G_{i+1}^{-1}$.
The following Lemma shows that our property (T2) in fact implies that
there is no contractible monochromatic cycles.

 \begin{figure}[!h]
 \center
 \includegraphics[scale=0.5]{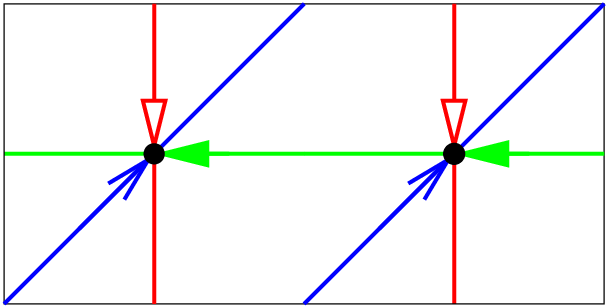}
 \caption{A toroidal graph where every Schnyder wood has one color which
   corresponding subgraph is not connected.}
 \label{fig:notconnected}
 \end{figure}

\begin{lemma}
\label{lem:nocontractiblecycle}
The graph $G_i\cup G_{i-1}^{-1}\cup G_{i+1}^{-1}$ contains no
contractible directed cycle.
\end{lemma}

\begin{proof}
  Suppose there is a contractible directed cycle in $G_i\cup
  G_{i-1}^{-1}\cup G_{i+1}^{-1}$. Let $C$ be such a cycle containing
  the minimum number of faces in the closed disk $D$ bounded by
  $C$. Suppose by symmetry that $C$ turns around $D$ clockwisely Then,
  by (T1), there is no edge of color $i-1$ leaving the closed disk
  $D$. So there is a $(i-1)$-cycle in $D$ and this cycle is $C$ by
  minimality of $C$.  Then, by (T1), there is no edge of color $i$
  leaving $D$. So, again by minimality of $C$, the cycle $C$ is a
  $i$-cycle.  Thus all the edges of $C$ are oriented in color $i$
  clockwisely and in color $i-1$ counterclockwisely.  Then, by (T1),
  all the edges of color $i+1$ incident to $C$ have to leave $D$. Thus
  there is no $(i+1)$-cycle intersecting $C$, a contradiction to
  property (T2).
\end{proof}

Let $G$ be a planar map and $x_0$, $x_1$, $x_2$ be three
distinct vertices occurring in counterclockwise order on the outer
face of $G$. One can transform $G^\sigma$ into the following toroidal
map $G^+$ (see Figure~\ref{fig:planar-tore}): Add a vertex $v$ in
the outer face of $G$. Add three non-parallel and non-contractible
loops on $v$. Connect the three half edges leaving $x_i$ to $v$ such
that there is no two such edge entering $v$ consecutively.  Then we
have the following.

\begin{theorem}
\label{th:planar}
  The Schnyder woods of a planar map $G$ rooted at $x_0$, $x_1$, $x_2$
  are in bijection with the Schnyder woods of the toroidal map $G^+$.
\end{theorem}

\begin{proof}
  ($\Longrightarrow$) Given a Schnyder wood of the planar graph $G$,
  rooted at $x_0$, $x_1$, $x_2$. Orient and color the graph $G^+$ as in
  the example of Figure~\ref{fig:planar-tore}, i.e. the edges of the
  original graph $G$ have the same color and orientation as in
  $G^\sigma$, the edge from $x_i$ to $v$ is colored $i$ and leaving
  $x_i$, the three loops around $v$ are colored and oriented
  appropriately so that $v$ satisfies the Schnyder property. Then it is
  clear that all the vertices of $G^+$ satisfy (T1). By (P2), we know
  that $G^\sigma$ has no monochromatic cycles. All the edges between
  $G$ and $v$ are leaving $G$, so there is no monochromatic cycle of
  $G^+$ involving vertices of $G$. Thus the only monochromatic cycles
  of $G^+$ are the three loops around $v$ and they satisfy (T2).

  ($\Longleftarrow$) Given a Schnyder wood of $G^+$, the restriction
  of the orientation and coloring to $G$ and the three edges leaving
  $v$ gives a Schnyder wood of $G^\sigma$. The three loops around $v$
  are three monochromatic cycles corresponding to three edges leaving
  $v$, thus they have different colors by (T1). Thus the three edges
  between $G$ and $v$ are entering $v$ with three different
  colors. The three loops around $v$ have to leave $v$ in
  counterclockwise order $0,1,2$ and we can assume by maybe permuting
  the colors that the edge leaving $x_i$ is colored $i$. Clearly all
  the vertices of $G^\sigma$ satisfies (P1). By
  Lemma~\ref{lem:nocontractiblecycle}, there is no contractible
  monochromatic cycles in $G^+$, so $G^\sigma$ satisfies (P2).
\end{proof}

\begin{figure}[!h]
\center
\includegraphics[scale=0.3]{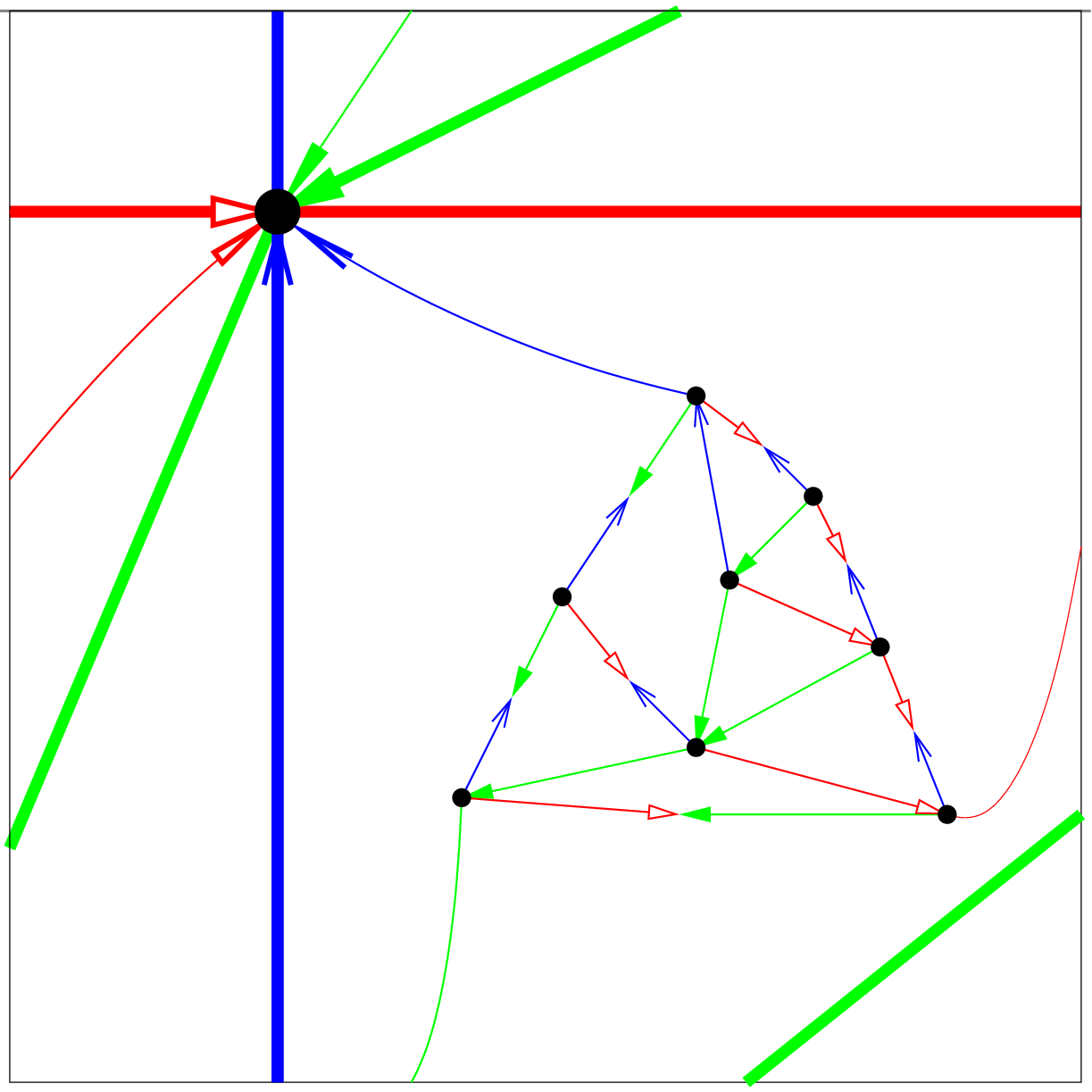}\\
\caption{The toroidal Schnyder wood corresponding to the planar
  Schnyder wood of Figure~\ref{fig:example-planar}.}
\label{fig:planar-tore}
\end{figure}

A planar map $G$ is \emph{internally 3-connected} if there exists
three vertices on the outer face such that the graph obtained from $G$
by adding a vertex adjacent to the three vertices is
3-connected. Miller~\cite{Mil02} (see also~\cite{Fel01}) proved that a
planar map admits a Schnyder wood if and only if it is internally
3-connected.  The following results show that the notion of
essentially 3-connected is the natural generalization of internally
3-connected to the torus. 

\begin{theorem}
  A planar map $G$ is internally 3-connected if and only if there
  exists three vertices on the outer face of $G$ such that $G^+$ is an
  essentially 3-connected toroidal map.
\end{theorem}

\begin{proof}
  ($\Longrightarrow$) Let $G$ be an internally 3-connected planar
  map. By definition, there exists three vertices $x_0$, $x_1$, $x_2$
  on the outer face such that the graph $G'$ obtained from $G$ by
  adding a vertex adjacent to these three vertices is 3-connected.
  Let $G''$ be the graph obtained from $G$ by adding three vertices
  $y_0$, $y_1$, $y_2$ that form a triangle and by adding the three
  edges $x_iy_i$.  It is not difficult to check that $G''$ is
  3-connected.  Since $G^{\infty}$ can be obtained from the (infinite)
  triangular grid, which is 3-connected, by gluing copies of $G''$
  along triangles, $G^{\infty}$ is clearly 3-connected. Thus $G^+$ is
  an essentially 3-connected toroidal map.

  ($\Longleftarrow$) Suppose there exists three vertices on the outer
  face of $G$ such that $G^+$ is an essentially 3-connected toroidal
  map, i.e. $G^{\infty}$ is 3-connected.  A copy of $G$ is contained
  in a triangle $y_0y_1y_2$ of $G^{\infty}$.  Let $G''$ be the
  subgraph of $G^{\infty}$ induced by this copy plus the triangle, and
  let $x_i$ be the unique neighbor of $y_i$ in the copy of $G$. Since
  $G''$ is connected to the rest of $G^{\infty}$ by a triangle, $G''$
  is also 3-connected. Let us now prove that this implies that $G$ is
  internally 3-connected for $x_0, x_1$ and $x_2$. This is equivalent
  to say that the graph $G'$, obtained by adding a vertex $z$
  connected to $x_0, x_1$ and $x_2$, is 3-connected.  If $G'$ had a
  separator $\{a,b\}$ or $\{a,z\}$, with $a,b\in V(G')\setminus
  \{z\}$, then $\{a,b\}$ or $\{a,y_i\}$, for some $i\in [0,2]$, would
  be a separator of $G''$. This would contradict the 3-connectedness
  of $G''$. So $G$ is internally 3-connected.
\end{proof}

\section{Two different types of Schnyder woods}
\label{sec:properties}

Two non contractible closed curves are \emph{homotopic} if one can be
continuously transformed into the other. The following are general
useful lemmas on the torus.

\begin{lemma}
\label{lem:intersect2}
  Let $C_1,C_2$ be two non contractible closed curve  on the Torus. If $C_1,C_2$
  are not homotopic, then their intersection is non empty.
\end{lemma}

\begin{lemma}
\label{lem:homotopic3}
Let $C_1,C_2,C_3$ be three non contractible closed curve on the
Torus. If $C_1,C_2$ are homotopic and $C_2,C_3$ are homotopic, then
$C_1,C_3$ are homotopic.
\end{lemma}

\begin{lemma}
\label{lem:intersect3}
Let $C_1,C_2,C_3$ be three non contractible closed curve  on the Torus. If
$C_1,C_2$ are homotopic and $C_1,C_3$ are not homotopic. Then $C_2,C_3$ are
not homotopic and thus their intersection is non empty.
\end{lemma}

Two non contractible oriented closed curves on the torus are
\emph{fully-homotopic} if one can be continuously transformed into the
other by preserving the orientation.  We say that two monochromatic
directed cycles $C_i,C_j$ of different colors are \emph{reversal} if
one is obtained from the other by reversing all the edges
($C_i=C_j^{-1}$). We say that two monochromatic cycle are
\emph{crossing} if they intersects but are not reversal. We define the
\emph{right side} of a $i$-cycle $C_i$, as the right side while
``walking'' along the directed cycle by following the orientation
given by the edges colored $i$.

Let $G$ be a toroidal graph given with a Schnyder wood.

\begin{lemma}
\label{lem:allhomotopic}
All $i$-cycles are non contractible, non intersecting and
fully-homotopic.
\end{lemma}

\begin{proof}
  By Lemma~\ref{lem:nocontractiblecycle}, all $i$-cycles are non
  contractible.  If there exists two such distinct $i$-cycles that are
  intersecting.  Then there is a vertex that has two outgoing edge of
  color $i$, a contradiction to (T1). So the $i$-cycles are non
  intersecting.  Then, by Lemma~\ref{lem:intersect2}, they are
  homotopic.

  Suppose that there exists two $i$-cycles $C_i,C_i'$ that are not
  fully-homotopic.  By the first part of the proof, cycles $C_i,C_i'$
  are non contractible, non intersecting and homotopic.  Let $R$ be
  the region between $C_i$ and $C_i'$ situated on the right of $C_i$.
  Suppose by symmetry that $C_i^{-1}$ is not a $(i+1)$-cycle. By (T2),
  there exists a cycle $C_{i+1}$ intersecting $C_i$ and thus $C_{i+1}$
  is crossing $C_i$. By Property (T1), $C_{i+1}$ is entering $C_i$
  from its right side and so it is leaving the region $R$ when it
  crosses $C_i$. To enter the region $R$, the cycle $C_{i+1}$ has to
  enter $C_i$ or $C_i'$ from their left side, a contradiction to
  property (T1).
\end{proof}

\begin{lemma}
\label{lem:twonothomotopic}
If two monochromatic cycles are crossing then they are of different
colors and they are not homotopic.
\end{lemma}

\begin{proof}
  By Lemma~\ref{lem:allhomotopic}, two crossing monochromatic cycles
  are not of the same color. Suppose that there exists two
  monochromatic cycles $C_{i-1}$ and $C_{i+1}$, of color $i-1$ and
  $i+1$, that are crossing and homotopic.  By
  Lemma~\ref{lem:nocontractiblecycle}, the cycles $C_{i-1}$ and
  $C_{i+1}$ are not contractible.  Since $C_{i-1}\neq C_{i+1}^{-1}$
  and $C_{i-1}\cap C_{i+1}\neq \emptyset$, the cycle $C_{i+1}$ is
  leaving $C_{i-1}$. It is leaving $C_{i-1}$ on its right side by
  (T1).  Since $C_{i-1}$ and $C_{i+1}$ are homotopic, the cycle
  $C_{i+1}$ is entering $C_{i-1}$ at least once from its right
  side. This is in contradiction with (T1).
\end{proof}

Let $\mathcal C_i$ be the set of $i$-cycles of $G$. Let $(\mathcal
C_i)^{-1}$ denote the set of cycles obtained by reversing all the
cycles of $\mathcal C_i$.  By Lemma~\ref{lem:allhomotopic}, the cycles
of $\mathcal C_i$ are non contractible, non intersecting and
fully-homotopic. So we can order them as follow $\mathcal
C_i=\{C_i^0,\ldots,C_i^{k_i-1}\}$, $k_i\geq 1$, such that, for $0\leq
j\leq k_i-1$, there is no $i$-cycle in the region $R(C_i^j,C_i^{j+1})$
between $C_i^j$ and $C_i^{j+1}$ containing the right side of $C_i^j$
(superscript understood modulo $k_i$).

We show that Schnyder woods are of two different types (see
figure~\ref{fig:case}):

\begin{figure}[!h]
\center
\begin{tabular}{ccc}
\includegraphics[scale=0.4]{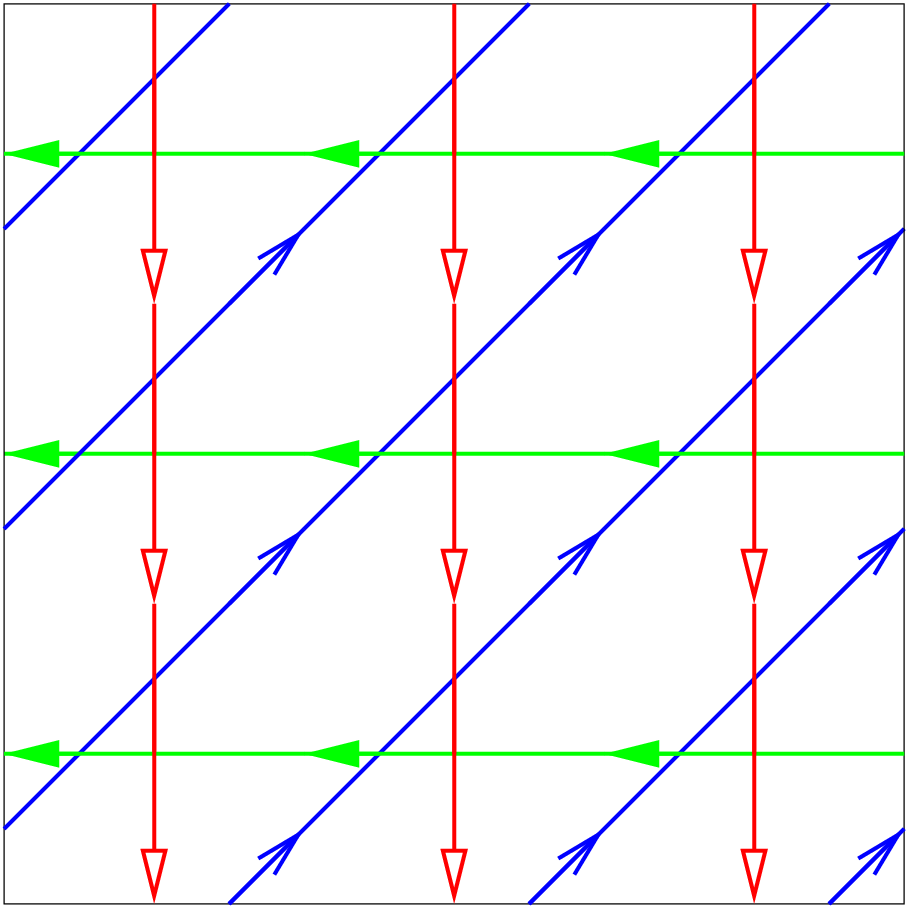}
& \hspace{2em} & 
\includegraphics[scale=0.4]{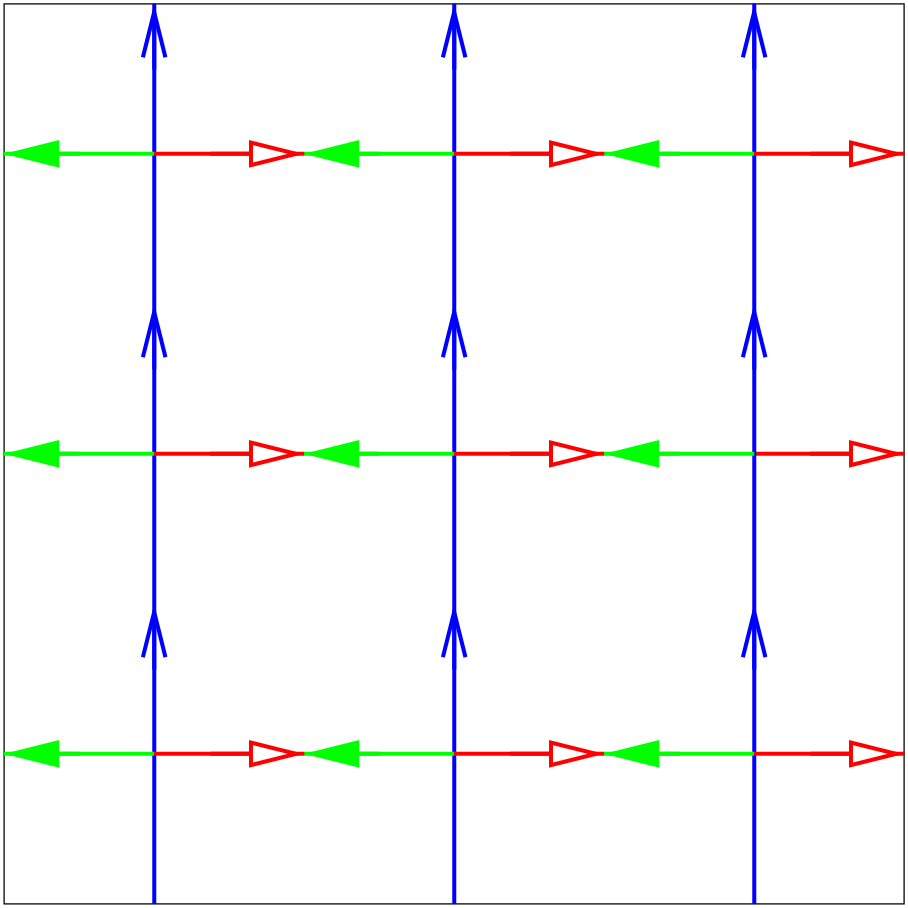} \\
Type 1& & Type 2  
\end{tabular}
\caption{The two types of Schnyder woods on toroidal graphs}
\label{fig:case}
\end{figure}

\begin{theorem}
  \label{lem:type}
Let $G$ be a toroidal graph given with a Schnyder wood, then all $i$-cycles are non contractible, non intersecting and
fully-homotopic and either:
\begin{itemize}
\item For every pair of two
  monochromatic cycles $C_i,C_j$ of different colors $i,j$, the two
  cycles $C_i$ and $C_j$ are not homotopic and thus intersect (We say
  the Schnyder wood is of \emph{Type 1}).
\end{itemize}
or
\begin{itemize}
\item There exists a color $i$ such that $\mathcal C_{i-1}=(\mathcal
  C_{i+1})^{-1}$ and for any pair of monochromatic cycles $C_i,C_j$ of
  colors $i,j$, with $j\neq i$, the two cycles $C_i$ and $C_j$ are not
  homotopic and thus intersect (We say the Schnyder wood is of \emph{Type 2},
  or \emph{Type 2.i} if we want to specify the color $i$).
\end{itemize}

Moreover, if $G$ is a toroidal triangulation, then there is no edges
oriented in two directions and the Schnyder wood is of Type 1.
\end{theorem}

\begin{proof}
  By Lemma~\ref{lem:allhomotopic}, all $i$-cycles are non
  contractible, non intersecting and fully-homotopic.  Suppose that
  there exists a $(i-1)$-cycle $C_{i-1}$ and a $(i+1)$-cycle $C_{i+1}$
  that are homotopic, i.e. $C_{i-1}=C_{i+1}^{-1}$. We prove that the
  Schnyder wood is of Type 2.i.  We first prove that $\mathcal
  C_{i-1}=(\mathcal C_{i+1})^{-1}$.  Let $C'_{i-1}$ be any
  $(i-1)$-cycle.  By (T2), $C'_{i-1}$ intersects a $(i+1)$-cycle
  $C'_{i+1}$.  By Lemma~\ref{lem:allhomotopic}, $C'_{i-1}$
  (resp. $C'_{i+1}$) is homotopic to $C_{i-1}$ (resp. $C_{i+1}$). So,
  by Lemma~\ref{lem:homotopic3}, $C'_{i-1}$ and $C'_{i+1}$ are
  homotopic. By Lemma~\ref{lem:twonothomotopic}, $C'_{i-1}$ and
  $C'_{i+1}$ are reversal. Thus $\mathcal C_{i-1}\subseteq(\mathcal
  C_{i+1})^{-1}$ and so by symmetry $\mathcal C_{i-1}=(\mathcal
  C_{i+1})^{-1}$.  Now we prove that for any pair of monochromatic
  cycles $C'_i,C'_j$ of colors $i,j$, with $j\neq i$, the two cycles
  $C'_i$ and $C'_j$ are not homotopic.  By (T2), $C'_{j}$ intersects a
  $i$-cycle $C_{i}$. Since $\mathcal C_{i-1}=(\mathcal C_{i+1})^{-1}$,
  cycle $C'_{j}$ is bi-oriented in color $i-1$ and $i+1$, thus we
  cannot have $C'_{j}=C_{i}^{-1}$. So $C'_{j}$ and $C_{i}$ are
  crossing and by Lemma~\ref{lem:twonothomotopic}, they are not
  homotopic.  By Lemma~\ref{lem:allhomotopic}, $C'_{i}$ and $C_i$ are
  homotopic. Thus, by Lemma~\ref{lem:intersect3}, $C'_{j}$ and
  $C'_{i}$ are not homotopic. Thus the Schnyder wood is of Type 2.i.

  If there is no two monochromatic cycles of different colors that are
  homotopic, then the Schnyder wood is of Type 1.

  For toroidal triangulation, $m=3n$ by Euler's formula, so there is no edges oriented in
  two directions, so only Type 1 is possible.
\end{proof}

Note that in a Schnyder wood of Type 1, we may have edges that are in
two monochromatic cycles of different colors (see
Figure~\ref{fig:example-schnyder}).

We do not know if the set of Schnyder woods of a given toroidal graph
has a kind of lattice structure like in the planar case~\cite{Fel04}.
De Fraysseix et al. \cite{FO01} proved that Schnyder woods of a planar
triangulation are in one-to-one correspondence with orientation of the
edges of the graph where each inner vertex has out-degree three. It is
possible to retrieve the coloring of the edges of a Schnyder wood from
the orientation. The situation is different for toroidal
triangulations. There exists orientations of toroidal triangulations
where each vertex has out-degree three but there is no corresponding
Schnyder wood.  For example, if one consider a toroidal triangulation
with just one vertex, the orientations of edges that satisfies (T1)
are the orientations where there is no three consecutive edges leaving
the vertex (see Figure~\ref{fig:outdegree}).

\begin{figure}[!h]
\center
\includegraphics[scale=0.5]{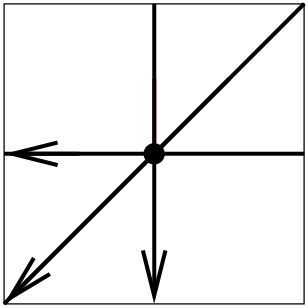}
\hspace{4em}
\includegraphics[scale=0.5]{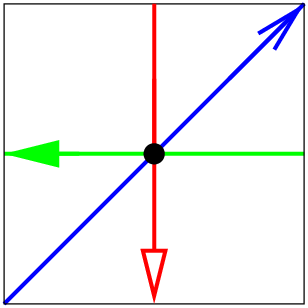}
\caption{Two different orientations of a toroidal triangulation. Only
  the second one corresponds to a Schnyder wood.}
\label{fig:outdegree}
\end{figure}

\section{Schnyder woods in the universal cover}
\label{sec:universal}

Let $G$ be a toroidal graph given with a Schnyder wood. Consider the
orientation and coloring of the edges of $G^\infty$ that corresponds
to the Schnyder wood of $G$.

\begin{lemma}
  The orientation and coloring of the edges of $G^\infty$ satisfies
  the following:

\begin{itemize}
\item[(U1)]   Every vertex of $G^\infty$ verifies the Schnyder property 
\item[(U2)]   There is no monochromatic cycle in $G^\infty$.
\end{itemize}
\end{lemma}

\begin{proof}
  Clearly, (U1) is satisfied.  Now we prove (U2). Suppose by
  contradiction that there is a monochromatic cycle $U$ of color $i$
  in $G^\infty$.  Let $C$ be the closed curve of $G$ corresponding to
  edges of $U$.  If $C$ self intersects, then there is a vertex of $G$
  with two edges leaving $v$ in color $i$, a contradiction to (T1). So
  $C$ is a monochromatic cycle of $G$.  Since $C$ corresponds to a
  cycle of $G^\infty$, it is a contractible cycle of $G$, a
  contradiction to Lemma~\ref{lem:nocontractiblecycle}.
\end{proof}

One can remark that properties (U1) and (U2) are the same as in the
definition of Schnyder wood for 3-connected planar graphs (properties
(P1) and (P2)). Note that if the orientation and coloring of the edges
of $G^\infty$, corresponding to an orientation and coloring of the
edges of $G$, satisfies properties (U1) and (U2), we do not
necessarily have a Schnyder wood of $G$. For example the graph
$G^\infty$ obtained by replicating the graph $G$ of
Figure~\ref{fig:nocross} satisfies (U1) and (U2) whereas the
orientation and coloring of $G$ is not a Schnyder wood as (T2) is not
satisfied.

Recall that the notation $\mathcal C_i=\{C_i^0,\ldots,C_i^{k_i-1}\}$
denotes the set of $i$-cycles of $G$ such that there is no $i$-cycle
in the region $R(C_i^j,C_i^{j+1})$. As monochromatic cycles are not
contractible by Lemma~\ref{lem:nocontractiblecycle}, a directed
monochromatic cycles $C_i^j$ corresponds to a family of infinite
directed monochromatic paths of $G^\infty$ (infinite in both
directions of the path). This family is denoted $\mathcal L_i^j$. Each
element of $\mathcal L_i^j$ is called a \emph{monochromatic line} of
color $i$, or \emph{$i$-line} for short. By
Lemma~\ref{lem:allhomotopic}, all $i$-lines are non intersecting and
oriented in the same direction. Given any two $i$-lines $L$, $L'$, the
unbounded region between $L$ and $L'$ is noted $R(L,L')$.  We say that
two $i$-lines $L,L'$ are \emph{consecutive} if there is no $i$-lines
contained in $R(L,L')$.

Let $v$ be a vertex of $G^\infty$. For each color $i$, vertex $v$ is
the starting vertex of a unique infinite directed monochromatic path
of color $i$, denoted $P_i(v)$. Indeed this is a path since there is
no monochromatic cycle in $G^\infty$ by Property (U2), and it is
infinite (in one direction of the path only) because every reached
vertex of $G^\infty$ has exactly one edge leaving in color $i$ by
Property (U1).  As $P_i(v)$ is infinite, it necessarily contains two
vertices $u,u'$ of $G^\infty$ that are copies of the same vertex of
$G$. The subpath of $P_i(v)$ between $u$ and $u'$ corresponds to a
$i$-cycle of $G$ and thus is part of a $i$-line of $G^\infty$.  Let
$L_i(v)$ be the $i$-line intersecting $P_i(v)$.

\begin{lemma}
  \label{lem:nocontractiblecycleuniversal}
  The graph $G^\infty_i\cup (G_{i-1}^{^\infty})^{-1}\cup
  (G_{i+1}^{^\infty})^{-1}$ contains no directed cycle.
\end{lemma}

\begin{proof}
  Suppose there is a contractible directed cycle $C$ in
  $G^\infty_i\cup (G_{i-1}^{^\infty})^{-1}\cup
  (G_{i+1}^{^\infty})^{-1}$. Let $D$ be the closed disk bounded by
  $C$.  Suppose by symmetry that $C$ turns around $D$
  clockwisely. Then, by (U1), there is no edge of color $i-1$ leaving
  the closed disk $D$. So there is a $(i-1)$-cycle in $D$, a
  contradiction to (U2).
\end{proof}

\begin{lemma}
  \label{lem:nocommon}
  For every vertex $v$ and color $i$, the two paths
  $P_{i-1}(v)$ and $P_{i+1}(v)$ have $v$ as only common vertex.
\end{lemma}

\begin{proof}
  If $P_{i-1}(v)$ and $P_{i+1}(v)$  intersect on two vertices,
  then $G^\infty_{i-1}\cup (G_{i+1}^{^\infty})^{-1}$  contains a cycle, contradicting
  Lemma~\ref{lem:nocontractiblecycleuniversal}.
\end{proof}

By Lemma~\ref{lem:nocommon}, for every vertex $v$, the three paths
$P_0(v)$, $P_1(v)$, $P_2(v)$ divide $G^\infty$ into three unbounded
regions $R_0(v)$, $R_1(v)$ and $R_2(v)$, where $R_i(v)$ denotes the
region delimited by the two paths $P_{i-1}(v)$ and $P_{i+1}(v)$. Let
$R_i^{\circ}(v)=R_i(v)\setminus (P_{i-1}(v)\cup P_{i+1}(v))$.

\begin{figure}[!h]
\center
\input{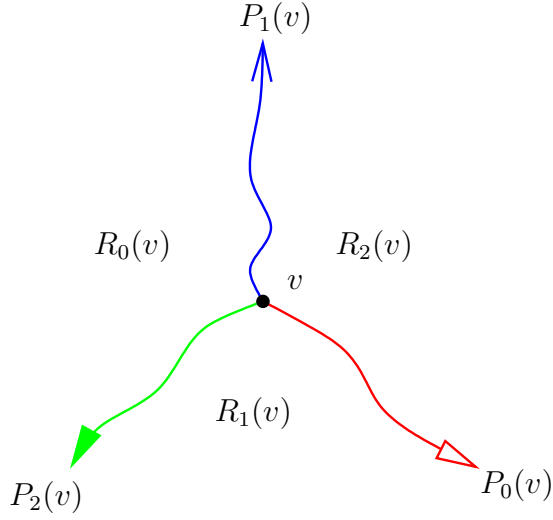}
\caption{Regions corresponding to a vertex}
\label{fig:regions}
\end{figure}

\begin{lemma}
  \label{lem:regionss} For all distinct vertices $u$, $v$, we have:
\begin{description}
\item[(i)] If $u\in R_i (v)$, then $R_i(u)\subseteq R_i(v)$.
\item[(ii)] If $u\in R_i^\circ (v)$, then $R_i(u)\subsetneq R_i(v)$.
\item[(iii)] There exists $i$ and $j$ with $R_i(u) \subsetneq R_i(v)$
  and $R_j(v)\subsetneq R_j (u)$.
\end{description}
\end{lemma}

\begin{proof}
  (i) Suppose by symmetry that the Schnyder wood is not of Type
  2.(i+1).  Then in $G$, $i$-cycles are not homotopic to
  $(i-1)$-cycles. Thus in $G^\infty$, every $i$-line crosses every
  $(i-1)$-line. Moreover a $i$-line crosses a $(i-1)$-line exactly
  once and from its right side to its left side by (U1).  Vertex $v$
  is between two consecutive monochromatic $(i-1)$-lines
  $L_{i-1},L'_{i-1}$, with $L'_{i-1}$ situated on the right of
  $L_{i-1}$.  Let $R$ be the region situated on the right of
  $L_{i-1}$, so $v\in R$.

  \begin{claim}
    \label{cl:leaves}
    For any vertex $w$ of $R$, the path $P_{i}(w)$ leaves the region $R$.
  \end{claim}

  \begin{proofclaim}
    The $i$-line $L_i(w)$ has to cross $L_{i-1}$ exactly once and from
    right to left, thus $P_{i}(w)$ leaves the region $R$.
  \end{proofclaim}

  The path $P_{i+1}(v)$ cannot leave the region $R$ as this would
  contradict (U1).  Thus by Claim~\ref{cl:leaves} for $w=v$, we have
  $R_i(v)\subseteq R$ and so $u\in R$.  Moreover the paths
  $P_{i-1}(u)$ and $P_{i+1}(u)$ cannot leave region $R_i(v)$ as this
  would contradict (U1).  Thus by Claim~\ref{cl:leaves} for $w=u$, the
  path $P_{i}(u)$ leaves the region $R_i(v)$ and so $R_i(u)\subseteq
  R_i(v)$.

  (ii) By (i), $R_i(u)\subseteq R_i(v)$, so the paths $P_{i-1}(u)$ and
  $P_{i+1}(u)$ are contained in $R_i(v)$. Then none of them can
  contain $v$ as this would contradict (U1).  So all the faces of
  $R_i(v)$ incident to $v$ are not in $R_i(u)$ (and there is at least
  one such face).

  (iii) By symmetry, we prove that there exists $i$ with $R_i(u)
  \subsetneq R_i(v)$.  If $u\in R_i^\circ(v)$ for some color $i$, then
  $R_i(u) \subsetneq R_i(v)$ by (ii).  Suppose now that $u\in P_i(v)$
  for some $i$.  By Lemma~\ref{lem:nocommon}, at least one of the two
  paths $P_{i-1}(u)$ and $P_{i+1}(u)$ does no contain $v$. Suppose by
  symmetry that $P_{i-1}(u)$ does not contain $v$. As $u\in
  P_i(v)\subseteq R_{i+1}(v)$, we have $R_{i+1}(u)\subseteq
  R_{i+1}(v)$ by (i), and as none of $P_{i-1}(u)$ and $P_{i}(u)$
  contains $v$, we have $R_{i+1}(u) \subsetneq R_{i+1}(v)$.
\end{proof}

\begin{lemma}
\label{lem:essentially}
If a toroidal graph $G$ admits a Schnyder wood, then $G$ is
essentially 3-connected.
\end{lemma}

\begin{proof}
  Let $u,v,x,y$ be any four distinct vertices of $G^\infty$. Let us
  prove that there exists a path between $u$ and $v$ in
  $G^\infty\setminus\{x,y\}$. Suppose by symmetry, that the Schnyder
  wood is of Type 1 or Type 2.1. Then the monochromatic lines of color
  $0$ and $2$ form a kind of grid, i.e. the $0$-lines intersect all
  the $2$-lines.  Let $L_0, L_0'$ be $0$-lines and $L_2, L_2'$ be
  $2$-lines, such that $u,v,x,y$ are all in the interior of the
  bounded region $R(L_0, L_0')\cap R(L_2,L'_2)$.

  By Lemma~\ref{lem:nocommon}, the three paths $P_i(v)$, for $0\leq
  i\leq 2$, are disjoint except on $v$. Thus there exists $i$, such
  that $P_i(v)\cap\{x,y\}=\emptyset$ .  Similarly there exists $j$,
  such that $P_j(u)\cap\{x,y\}=\emptyset$. The two paths $P_i(v)$ and
  $P_j(u)$ are infinite, so they intersect the boundary of $R(L_0,
  L_0')\cap R(L_2,L'_2)$.  Thus $P_i(v)\cup P_j(u)\cup L_0\cup
  L_0'\cup L_2\cup L'_2$ contains a path from $u$ to $v$ in
  $G^\infty\setminus\{x,y\}$.
\end{proof}

By Lemma~\ref{lem:essentially}, if $G$ admits a Schnyder wood, then it
is essentially 3-connected, so it is a map and each face is a disk. 

Note that if (T2) is not required in the definition of Schnyder wood,
then Lemma~\ref{lem:essentially} is false.
Figure~\ref{fig:not3connected} gives an example of an orientation and
coloring of the edges of a toroidal graph satisfying (T1), such that
there is no contractible monochromatic cycles but where (T2) is not
satisfied as there is a $0$-cycle not intersecting any $2$-cycle. This
graph is not essentially 3-connected, indeed $G^\infty$ is not
connected.

\begin{figure}[!h]
\center
\includegraphics[scale=0.5]{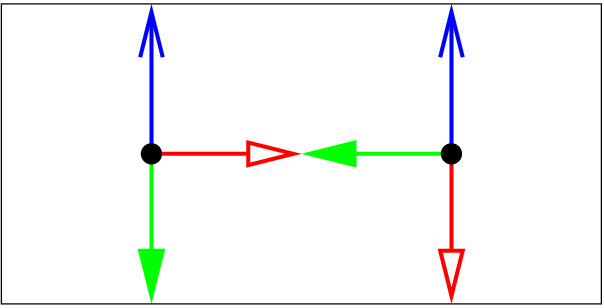}
\caption{An orientation and coloring of the edges of a toroidal graph
  satisfying (T1) but that is not essentially 3-connected.}
\label{fig:not3connected}
\end{figure}

\section{Duality of Schnyder woods}
\label{sec:dual}

Given a planar map $G$, and $x_0$, $x_1$, $x_2$ three distinct
vertices occurring in counterclockwise order on the outer face of $G$.
A \emph{Schnyder angle labeling}~\cite{Fel01} of $G$ with respect to
$x_0$, $x_1$, $x_2$ is a labeling of the angles of $G^\sigma$
satisfying the following:

\begin{itemize}
\item[(L1)] The label of the angles at each vertex form, in
  counterclockwise order, nonempty intervals of $0$'s, $1$'s and $2$'s.
  The two angles at the half-edge at $x_i$ have labels $i+1$ and $i-1$
\item[(L2)] The label of the angles at each inner face form, in
  counterclockwise order, nonempty intervals of $0$'s, $1$'s and
  $2$'s.  At the outer-face the same is true in clockwise order.
\end{itemize}

Felsner~\cite{Fel03} proved that, for planar maps, Schnyder woods are
in bijection with Schnyder angle labellings. In the toroidal case, we
do not see a simple definition of Schnyder angle labeling that would
be equivalent to our definition of Schnyder woods. This is due to the
fact that contrarily to (P2) that is local and can be verified just by
considering faces, (T2) is global. Nevertheless we have one
implication.

  The \emph{angle labeling corresponding to} a Schnyder wood of a
  toroidal map $G$ is a labeling of the angles of $G$ such that the
  angles at a vertex $v$ in the counterclockwise sector between
  $e_{i+1}(v)$ and $e_{i-1}(v)]$ are labeled $i$ (see
  Figure~\ref{fig:angle}).

\begin{figure}[!h]
\center
\includegraphics[scale=0.5]{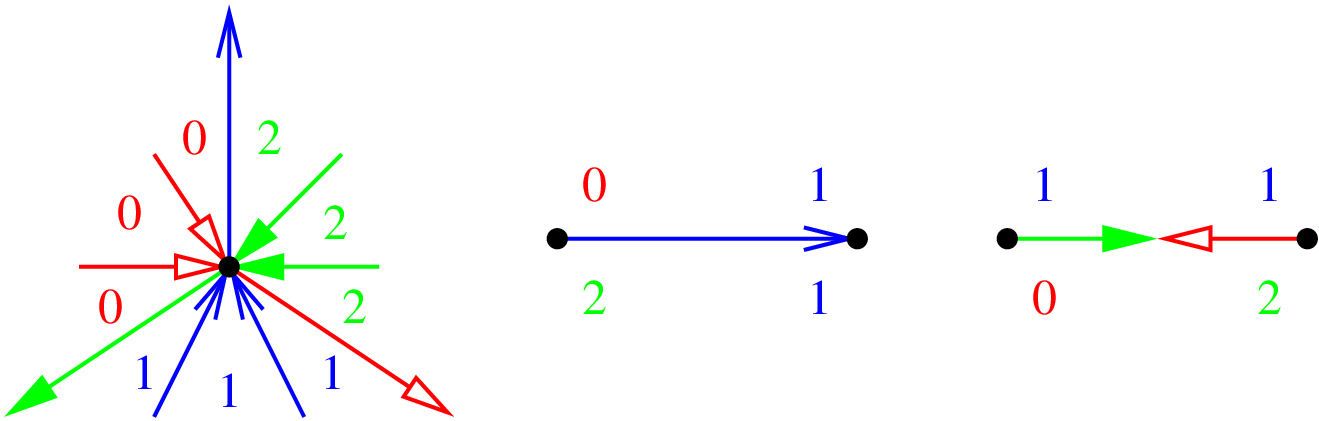}
\caption{Angle labeling around  vertices and edges}
\label{fig:angle}
\end{figure}

\begin{lemma}
\label{lem:angle}
The angle labeling corresponding to a Schnyder wood of a toroidal map
satisfies the following: the angles at each vertex and at each face
form, in counterclockwise order, nonempty intervals of $0$'s, $1$'s
and $2$'s.
\end{lemma}

\begin{proof}
  Clearly the property is true at each vertex by (T1). To prove that
  the property is true at each faces we count the number of color
  changes around vertices, faces and edges. This number of changes is
  denoted $d$. For a vertex $v$ there are exactly three changes, so
  $d(v)=3$ (see Figure~\ref{fig:angle}).  For an edge $e$, that can be
  either oriented in one or two direction, there are also exactly
  three changes, so $d(e)=3$ (see Figure~\ref{fig:angle}).  Now
  consider a face $F$.  Suppose we cycle counterclockwise around $F$,
  then an angle colored $i$ is always followed by an angle colored $i$
  or $i + 1$. Consequently, $d(F)$ must be a multiple of three.
  Suppose that $d(F)=0$, then all its angles are colored with one
  color $i$. In that case the cycle around face $F$ would be
  completely oriented in counterclockwise order in color $i+1$ (and in
  clockwise order in color $i-1$). This cycle being contractible this
  would contradict Lemma~\ref{lem:nocontractiblecycle}. So $d(F)\geq
  3$.

  The sum of the changes around edges must be equal to the sum of the
  changes around faces and vertices. Thus $3m=\sum_e d(e)=\sum_v d(v)
  + \sum_F d(F)=3n + \sum_F d(F)$. Euler's formula gives $m=n+f$, so
  $\sum_F d(F)=3f$ and this is possible only if $d(F)=3$ for every
  face $F$.
\end{proof}

There is no converse to Lemma~\ref{lem:angle}.
Figure~\ref{fig:nocross} gives an example of a coloring and
orientation of the edges of a toroidal triangulation not satisfying
(T2) but where the angles at each vertex and at each face form, in
counterclockwise order, nonempty intervals of $0$'s, $1$'s and $2$'s.

Let $G$ be a toroidal graph given with a Schnyder wood. By
Lemma~\ref{lem:essentially}, $G$ is an essentially 3-connected
toroidal map, thus the dual $G^*$ of $G$ has no contractible loop and
no homotopic multiple edges.  Let $\widetilde{G}$ be a simultaneous
drawing of $G$ and $G^*$ such that only dual edges intersect.

The \emph{dual} of the Schnyder wood is the orientation and coloring
of the edges of $G^*$ obtained by the following method (see
Figure~\ref{fig:rule-dual} and~\ref{fig:example-dual-tore}): Let $e$
be an edge of $G$ and $e^*$ the dual edge of $e$.  If $e$ is oriented
in one direction only and colored $i$, then $e^*$ is oriented in two
direction, entering $e$ from the right side in color $i-1$ and from
the left side in color $i+1$ (the right side of $e$ is the right side
while following the orientation of $e$).  Symmetrically, if $e$ is
oriented in two direction in color $i+1$ and $i-1$, then $e^*$ is
oriented in one direction only and colored $i$ such that $e$ is
entering $e^*$ from its right side in color $i-1$.

\begin{figure}[!h]
\center
\includegraphics[scale=0.5]{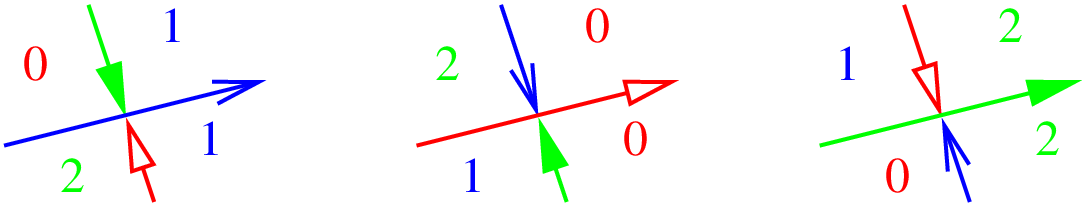}
\caption{Rules for the dual Schnyder wood and angle labeling.}
\label{fig:rule-dual}
\end{figure}

\begin{figure}[!h]
\center
\includegraphics[scale=0.5]{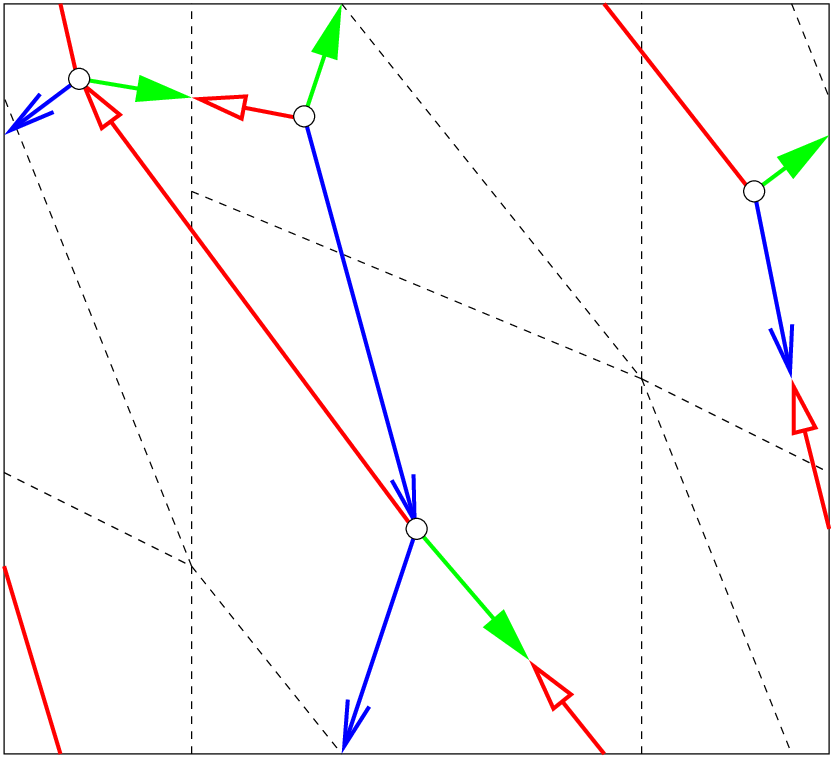}
\caption{Dual Schnyder wood of the Schnyder wood of
  Figure~\ref{fig:example-schnyder}.}
\label{fig:example-dual-tore}
\end{figure}

  \begin{lemma}
\label{lem:dual}
    Let $G$ be a toroidal map. 
    The dual of a Schnyder wood of a toroidal map $G$ is a Schnyder
    wood of the dual $G^*$. Moreover we have :

    \begin{itemize}
    \item[(i)] On the simultaneous drawing $\widetilde{G}$ of $G$ and
      $G^*$, the $i$-cycles of the dual Schnyder wood are homotopic to
      the $i$-cycles of the primal Schnyder wood and oriented in
      opposite direction.

    \item[(ii)] The dual of a Schnyder wood is of Type 2.i if and only if
      the primal Schnyder wood is of Type 2.i.
    \end{itemize}
  \end{lemma}

  \begin{proof}
    In every face of $\widetilde{G}$, there is exactly one angle of
    $G$ and one angle of $G^*$. Thus a Schnyder angle labeling of $G$
    corresponds to an angle labeling of $G^*$.  The dual of the
    Schnyder wood is defined such that an edge $e$ is leaving $F$ in
    color $i$ if and only if the angle at $F$ on the left of $e$ is
    labeled $i-1$ and the angle at $F$ on the right of $e$ is
    labeled $i+1$, and such that an edge $e$ is entering $F$ in color
    $i$ if and only if at least one of the angles at $F$ incident to
    $e$ is labeled $i$ (see Figure~\ref{fig:rule-dual}).  By
    Lemma~\ref{lem:angle}, the angles at a face form, in
    counterclockwise order, nonempty intervals of $0$'s, $1$'s and
    $2$'s. Thus the edges around a vertex of $G^*$ satisfy property
    (T1).

    Consider $\widetilde{G}$ with the orientation and coloring of
    primal and dual edges.

    Let $C$ be a $i$-cycle of $G^*$. Suppose, by contradiction, that
    $C$ is contractible. Let $D$ be the disk delimited by $C$. Suppose
    by symmetry that $C$ is going anticlockwise around $D$.  Then all
    the edges of $G$ that are dual to edges of $C$ are entering
    $D$ in color $i-1$. Thus $D$ contains a $i-1$-cycle of $G$, a
    contradiction to Lemma~\ref{lem:nocontractiblecycle}. Thus every
    monochromatic cycle of $G^*$ is non contractible.

    The dual of the Schnyder wood is defined in such a way that an
    edge of $G$ and an edge of $G^*$ of the same color never
    intersect in $\widetilde{G}$. Thus the $i$-cycles of $G^*$ are
    homotopic to $i$-cycles of $G$.  Consider a $i$-cycle $C_i$
    (resp. $C_i^*$) of $G$ (resp. $G^*$). The two cycles $C_i$ and
    $C_i^*$ are homotopic.  By symmetry, we assume that the primal
    Schnyder wood is not of Type 2.(i-1). Let $C_{i+1}$ be a
    $(i+1)$-cycle of $G$. The two cycles $C_i$ and $C_{i+1}$ are not
    homotopic and $C_i$ is entering $C_{i+1}$ on its left side. Thus,
    by Lemma~\ref{lem:intersect3}, the two cycles $C_i^*$ and
    $C_{i+1}$ are not homotopic and by the dual rules $C_i^*$ is
    entering $C_{i+1}$ on its right side. So $C_i$ and $C_i^*$ are
    homotopic and going in opposite directions.

    Suppose the Schnyder wood of $G$ is of Type 1.  Then two
    monochromatic cycles of $G$ of different colors are not
    homotopic. Thus the same is true for monochromatic cycles of the
    dual. So (T2) is satisfied and the dual of the Schnyder wood is a
    Schnyder wood of Type 1.

    Suppose now that the Schnyder wood of $G$ is of Type 2.  Assume by
    symmetry that it is of Type 2.i.  Then all monochromatic cycles of
    color $i$ and $j$, with $j\in \{i-1,i+1\}$, intersect. Now
    suppose, by contradiction, that there is a $j$-cycle $C^*$, with
    $j\in \{i-1,i+1\}$, that is not equal to a monochromatic cycle of
    color in $\{i-1,i+1\}\setminus \{j\}$. By symmetry we can assume
    that $C^*$ is of color $i-1$.  Let $C$ be the $(i-1)$-cycle of the
    primal that is the first on the right side of $C^*$ in
    $\widetilde{G}$. By definition of Type 2.i, $C^{-1}$ is a
    $(i+1)$-cycle of $G$. Let $R$ be the region delimited by $C^*$ and
    $C$ situated on the right side of $C^*$.  Cycle $C^*$ is not a
    $(i+1)$-cycle so there is at least one edge of color $i+1$ leaving
    a vertex of $C^*$. By (T1) in the dual, this edge is entering the
    interior of the region $R$. An edge of $G^*$ of color $i+1$ cannot
    intersect $C$ and cannot enter $C^*$ from its right side. So in
    the interior of the region $R$ there is at least one $(i+1)$-cycle
    $C^{*}_{i+1}$ of $G^*$.  Cycle $C^{*}_{i+1}$ is homotopic to
    $C^{*}$ and going in opposite direction (i.e. $C^{*}_{i+1}$ and
    $C^{*}$ are not fully-homotopic). If $C^{*}_{i+1}$ is not a
    $(i+1)$-cycle, then we can define $R'\subsetneq R$ the region
    delimited by $C^{*}_{i+1}$ and $C$ situated on the left side of
    $C^{*}_{i+1}$ and as before we can prove that there is a
    $(i-1)$-cycle of $G^*$ in the interior of $R'$. So in any case,
    there is a $(i-1)$-cycle $C^*_{i-1}$ of $G^*$ in the interior of
    $R$ and $C^{*}_{i-1}$ is fully-homotopic to $C^{*}$. Let
    $R''\subsetneq R$ be the region delimited by $C^*$ and $C^*_{i-1}$
    situated on the right side of $C^*$. Clearly $R''$ does not
    contain $C$. Thus by definition of $C$, the region $R''$ does not
    contain any $(i-1)$-cycle of $G$. But $R''$ is non empty and
    contains at least one vertex $v$ of $G$. The path $P_{i-1}(v)$
    cannot leave $R''$, a contradiction. So (T2) is satisfied and the
    dual Schnyder wood is of Type 2.i.
  \end{proof}

  By Lemma~\ref{lem:dual}, we have the following:

  \begin{theorem}
\label{th:bijdual}
    There is a bijection between Schnyder woods of a toroidal map and
    Schnyder woods of its dual.
  \end{theorem}

\section{Relaxing the definition}
\label{sec:relax}

In the plane, the proof of existence of Schnyder woods can be done without too
much difficulty as the properties  to be satisfied are 
only local. 
In the toroidal case, things are much more complicated as property
(T2) is global.  The following lemma shows that property (T2)
can be relaxed a bit in the case of Schnyder woods of Type 1.

\begin{lemma}
\label{lem:tri-relax}
Let $G$ be a toroidal graph given with an orientation and coloring of
the edges of $G$ with the colors $0$, $1$, $2$, where every edge $e$
is oriented in one direction or in two opposite directions.  The
orientation and coloring is a Schnyder wood of Type 1 if and only if
it satisfies the following:

\begin{itemize}
\item[(T1')] Every vertex $v$ satisfies the Schnyder property.

\item[(T2')] For each pair $i,j$ of different colors, there exists a
  $i$-cycle intersecting a $j$-cycle.

\item[(T3')] There is no monochromatic cycles $C_i,C_j$ of different
  colors $i,j$ such that $C_i=C_j^{-1}$.
\end{itemize}

\end{lemma}

\begin{proof} 
  ($\Longrightarrow$) If we have a Schnyder wood of Type 1, then
  Property (T1') is satisfied as it is equal to Property (T1).
  Property (T1) implies that there always exist a monochromatic cycles
  of each color, thus Property (T2') is a relaxation of (T2). Property
  (T3') is implied by definition of Type 1 (see
  Theorem~\ref{lem:type}).

  ($\Longleftarrow$) Conversely, suppose we have an orientation and
  coloring satisfying (T1'), (T2'), (T3').  We prove several
  properties.

  \begin{claim}
    \label{cl:claim1} All $i$-cycles are non contractible, non
    intersecting and homotopic.
  \end{claim}

  \begin{proofclaim}
    Suppose there is a contractible monochromatic cycle. Let $C$ be
    such a cycle containing the minimum number of faces in the closed
    disk $D$ bounded by $C$. Suppose by symmetry that $C$ turns around
    $D$ clockwisely. Let $i$ be the color of $C$. Then, by (T1'), there
    is no edge of color $i-1$ leaving the closed disk $D$. So there is
    a $(i-1)$-cycle in $D$ and this cycle is $C$ by minimality of $C$,
    a contradiction to (T3').

    If there exists two distinct $i$-cycles that are intersecting.
    Then there is vertex that has two outgoing edge of color $i$, a
    contradiction to (T1'). So the $i$-cycles are non intersecting.
    Then, by Lemma~\ref{lem:intersect2}, they are homotopic.
  \end{proofclaim}

  \begin{claim}
    \label{cl:claim2}
    If two monochromatic cycles are intersecting then they are not
    homotopic.
  \end{claim}

  \begin{proofclaim}
    Suppose by contradiction that there exists $C,C'$ two distinct
    directed monochromatic cycles that are homotopic and
    intersecting. By Claim~(\ref{cl:claim1}), they are not
    contractible and of different color.  Suppose $C$ is a
    $(i-1)$-cycle and $C'$ a $(i+1)$-cycle.  By (T1'), $C'$ is
    leaving $C$ on its right side. Since $C,C'$ are homotopic, the
    cycle $C'$ is entering $C$ at least once from its right side, a
    contradiction with (T1').
  \end{proofclaim}

  We are now able to prove that (T2) is satisfied. Let $C_i$ be any
  $i$-cycle of color $i$. We have to prove that $C_i$ intersects at
  least one $(i-1)$-cycle and at least one $(i+1)$-cycle.  Let $j$ be
  either $i-1$ or $i+1$.  By (T2'), there exists a $i$-cycle $C'_i$
  intersecting a $j$-cycle $C'_j$ of color $j$. The two cycles $C'_i,
  C'_j$ are not reversal by (T3'), thus they are crossing.  By
  Claim~(\ref{cl:claim2}), $C'_i$ and $C'_j$ are not homotopic.  By
  Claim~(\ref{cl:claim1}), $C_i$ and $C'_i$ are homotopic. Thus by
  Lemma~\ref{lem:intersect3}, $C_i$ and $C'_j$ are intersecting.

Thus (T1) and (T2) are satisfied and thus the orientation and coloring
is a Schnyder wood. By (T3') and Theorem~\ref{lem:type} it is a
Schnyder of Type 1.
\end{proof}

Note that for toroidal triangulations, there is no edges oriented in
two directions in an orientation and coloring of the edges satisfying
(T1'), by Euler's formula. So (T3') is automatically satisfied. Thus in the case of
toroidal triangulations it is sufficient to have properties (T1') and
(T2') to have a Schnyder wood.  This is not true in general as shown
by the example of Figure~\ref{fig:relax} that satisfies (T1') and
(T2') but that is not a Schnyder wood. There is a monochromatic cycle
of color $1$ that is not intersecting any monochromatic cycle of color
$2$ so (T2) is not satisfied.

\begin{figure}[!h]
\center
\includegraphics[scale=0.5]{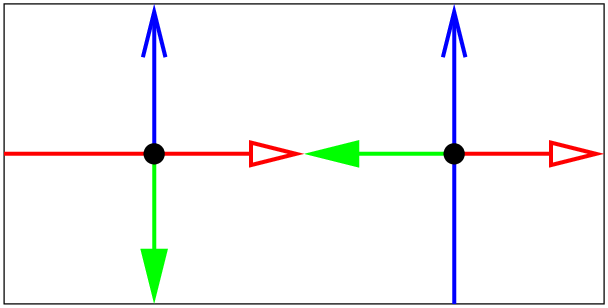}
\caption{An orientation and coloring of the edges of toroidal graph
  satisfying (T1') and (T2') but that is not a Schnyder wood.}
\label{fig:relax}
\end{figure}

  \section{Existence for simple triangulations}
\label{sec:existence}

In this section we present a short proof of existence of Schnyder
woods for simple triangulations. Sections~\ref{sec:lemma}
and~\ref{sec:existence3connected} contain the full proof of existence
for essentially 3-connected toroidal maps.

Fijavz~\cite{Fij} proved a useful result concerning existence of
particular non homotopic cycles in toroidal triangulations with no
loop and no multiple edges. (Recall that in this paper we are less
restrictive as we allow non contractible loops and non homotopic
multiple edges.)

\begin{theorem}[\cite{Fij}]
\label{th:fij} 
A simple toroidal triangulation contains three non contractible and
non homotopic cycles that all intersect on one vertex and that are
pairwise disjoint otherwise.
\end{theorem}

Theorem~\ref{th:fij} is not true  for all toroidal triangulations
as shown by the example on Figure~\ref{fig:not-connected}.

\begin{figure}[!h]
\center
\includegraphics[scale=0.5]{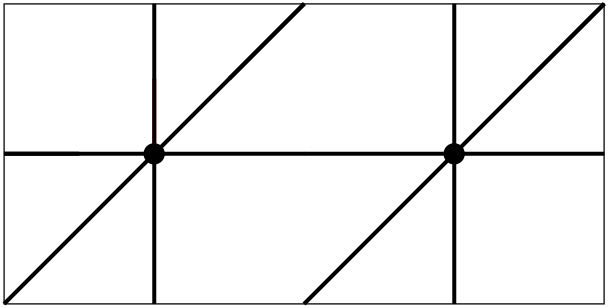}
\caption{A toroidal triangulation that does not contain three non
  contractible and non homotopic cycles that all intersect on one
  vertex and that are pairwise disjoint otherwise.}
\label{fig:not-connected}
\end{figure}

Theorem~\ref{th:fij} can be used to prove existence of particular
Schnyder woods for simple toroidal triangulations.

\begin{theorem}
\label{th:schnydersimple}
A simple toroidal triangulation admits a Schnyder wood with three
monochromatic cycles of different colors all intersecting on one
vertex and that are pairwise disjoint otherwise.
\end{theorem}

\begin{proof}
  Let $G$ be a simple toroidal triangulation.  By
  Theorem~\ref{th:fij}, let $C_0,C_1,C_2$ be three non contractible
  and non homotopic cycles of $G$ that all intersect on one vertex $x$
  and that are pairwise disjoint otherwise. By eventually shortening
  the cycles $C_i$, we can assume that the three cycles $C_i$ are
  homotopically chordless (i.e. there is no edge between two vertices
  of $C_i$ that can be continuously transformed into a part of $C_i$).
  By symmetry, we can assume that the six edges $e_i,e'_i$ of the
  cycles $C_i$ incident to $x$ appear around $x$ in the
  counterclockwise order $e_0,e'_2,e_1,e'_0,e_2,e'_1$.  The cycles
  $C_i$ divide $G$ into two regions, denoted $R_1,R_2$ such that $R_1$
  is the region situated in the counterclockwise sector between $e_0$
  and $e'_2$ of $x$ and $R_2$ is the region situated in the
  counterclockwise sector between $e'_2$ and $e_1$ of $x$. Let $G_i$
  be the subgraph of $G$ contained in the region $R_i$ (including the
  three cycles $C_i$). Let $G'_1$ (resp. $G'_2$) be the graph obtained
  from $G_1$ (resp. $G_2$) by replacing $x$ by three vertices $x_1,
  x_2, x_3$, such that $x_i$ is incident to the edges in the
  counterclockwise sector between $e_{i+1}$ and $e'_i$ (resp. $e'_i$
  and $e_{i-1}$).  The cycles $C_i$ being homotopically chordless, the two
  graphs $G'_1$ and $G'_2$ are internally 3-connected planar maps (for
  vertices $x_i$ chosen on their outer face).  The vertices
  $x_0,x_1,x_2$ appear in counterclockwise order on the outer face of
  $G'_1$ and $G'_2$.  By a result of Miller~\cite{Mil02} (see also
  \cite{Fel01, Fel03}), the two graphs $G'_i$ admit planar Schnyder
  woods rooted at $x_0,x_1,x_2$.  Orient and color the edges of $G$
  that intersect the interior of $R_i$ by giving them the same
  orientation and coloring as in a planar Schnyder wood of $G'_i$.
  Orient and color the cycle $C_i$ in color $i$ such that it is
  entering $x$ by edge $e'_i$ and leaving $x$ by edge $e_i$. We claim
  that the orientation and coloring that is obtained is a toroidal
  Schnyder wood of $G$.

  Clearly, any interior vertex of the region $R_i$ satisfies (T1).
  Let us show that (T1) is also satisfied for any vertex $v$ of a
  cycle $C_i$ distinct from $x$.  In a Schnyder wood of $G'_1$, the
  cycle $C_i$ is oriented in two direction, from $x_{i-1}$ to $x_{i}$
  in color $i$ and from $x_{i}$ to $x_{i-1}$ in color $i-1$. Thus the
  edge leaving $v$ in color $i+1$ is an inner edge of $G'_1$ and
  vertex $v$ has no edges entering in color $i+1$. Symmetrically, in
  $G'_2$ the edge leaving $v$ in color $i-1$ is an inner edge of
  $G'_2$ and vertex $v$ has no edges entering in color $i-1$. Then one
  can paste $G'_1$ and $G'_2$ along $C_i$, orient $C_i$ in color $i$
  and see that $v$ satisfies (T1). The definition of the $G'_i$, and
  the orientation of the cycles is done so that $x$ satisfies (T1).
  The cycles $C_i$ being pairwise intersecting, (T2') is satisfied, so
  by Lemma~\ref{lem:tri-relax}, the orientation and coloring is a
  Schnyder wood.
\end{proof}

Note that in the Schnyder wood obtained by
Theorem~\ref{th:schnydersimple}, we do not know if there are several
monochromatic cycles of one color or not.  So given any three
monochromatic cycles of different color, they might not all intersect
on one vertex. But for any two monochromatic cycles of different
color, we know that they intersect exactly once.  We wonder whether
Theorem~\ref{th:schnydersimple} can be modified as follow : Does a
simple toroidal triangulation admits a Schnyder wood such that there
is just one monochromatic cycle per color ? Moreover can one require
that the monochromatic cycles of different colors pairwise intersect
exactly once ?  Or like in Theorem~\ref{th:schnydersimple} that they
all intersect on one vertex and that they are pairwise disjoint
otherwise ?

\section{The contraction Lemma}
\label{sec:lemma}

We prove existence of Schnyder wood for essentially 3-connected
toroidal maps by contracting edges until we obtain a graph with just
few vertices. Then the graph can be decontracted step by step to
obtain a Schnyder wood of the original graph.

Given a toroidal map $G$, the \emph{contraction} of a non-loop-edge
$e$ of $G$ is the operation consisting of continuously contracting $e$
until merging its two ends. We note $G/e$ the obtained graph.  On
Figure~\ref{fig:contraction} the contraction of an edge $e$ is
represented. We consider three different cases corresponding to whether
the faces adjacent to the edge $e$ are triangles or not.  Note that
only one edge of each set of homotopic multiple edges that is
possibly created is preserved.

\begin{figure}[!h]
\center
\begin{tabular}{ccc}
$G$& &$G/e$ \\
& & \\
\includegraphics[scale=0.4]{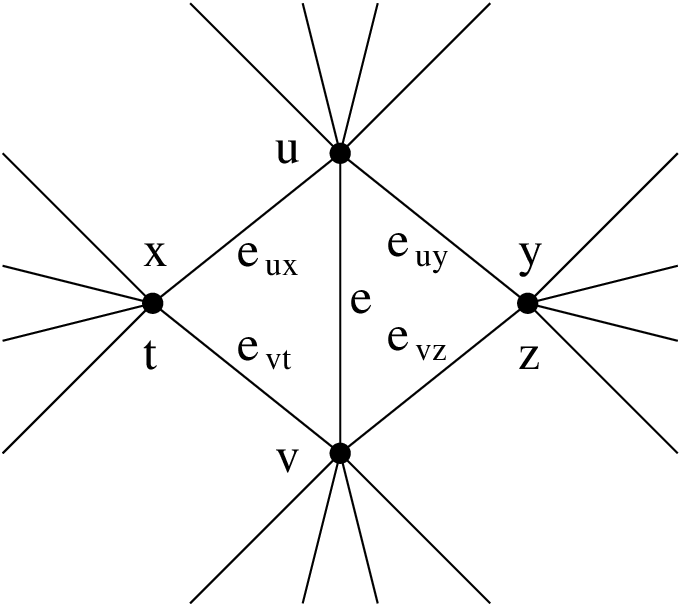}
& \hspace{3em} &
\includegraphics[scale=0.4]{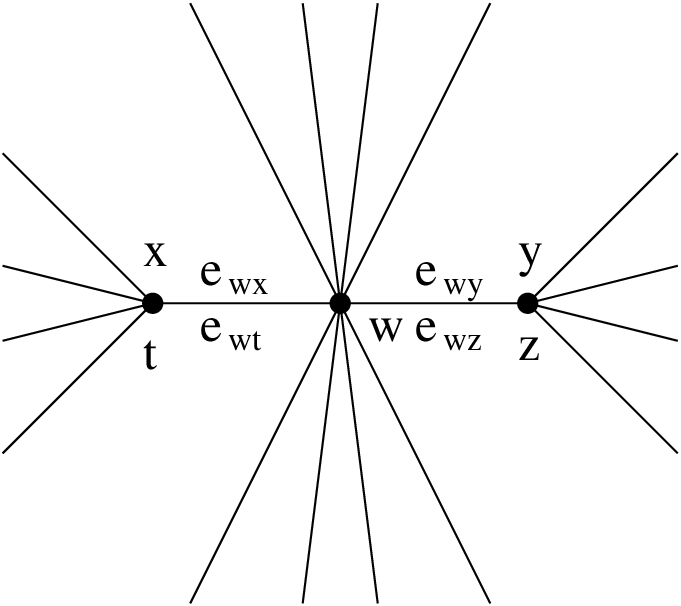}\\
(a) & & (d) \\
& & \\
\includegraphics[scale=0.4]{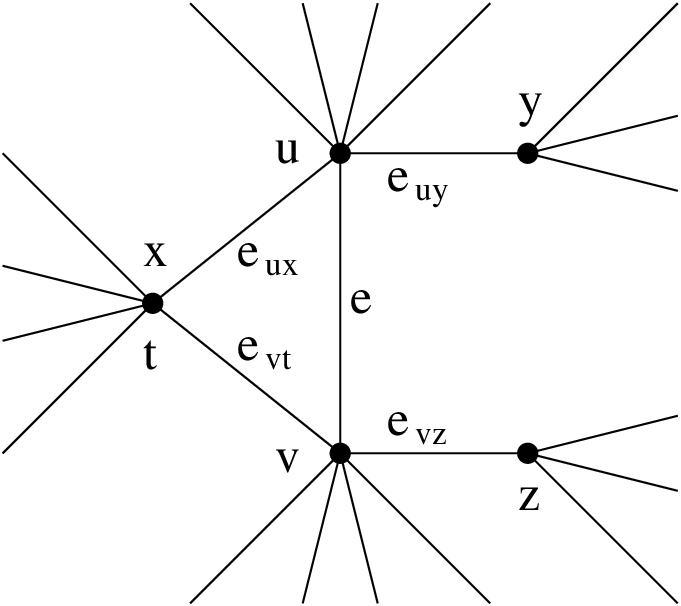}
& \hspace{3em} &
\includegraphics[scale=0.4]{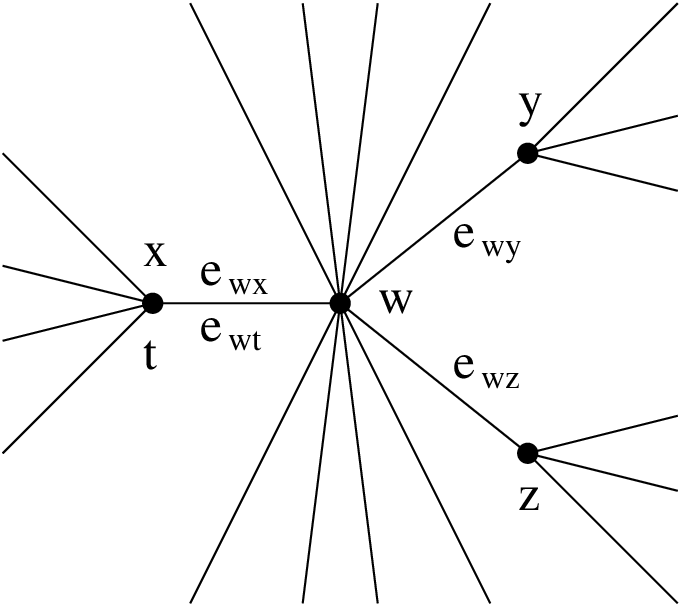}\\
(b) & & (e) \\
& & \\
\includegraphics[scale=0.4]{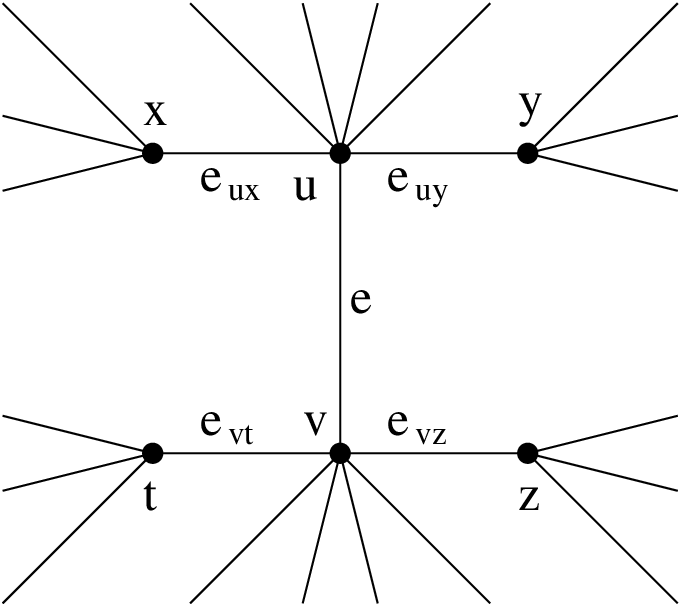}
& \hspace{3em} &
\includegraphics[scale=0.4]{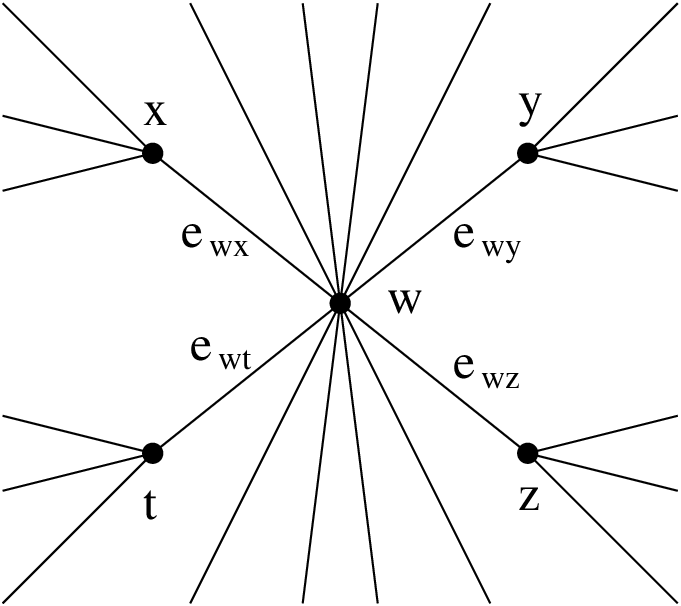}\\
(c) & & (f) \\
\end{tabular}
\caption{The contraction operation}
\label{fig:contraction}
\end{figure}

The goal of this section is to prove the following lemma, that plays a
key role in the proof of
Section~\ref{sec:existence3connected}.

\begin{lemma}
  \label{lem:contractype1}
  If $G$ is a toroidal map given with a non-loop edge $e$ whose extremities
  are of degree at least three and such that $G/ e$ admits a Schnyder
  wood of Type 1, then $G$ admits a Schnyder wood of type 1.
\end{lemma}

The proof of Lemma~\ref{lem:contractype1} is long and technical.  In
the planar case, an analogous lemma can be proved without too much
difficulty (see  Section~2.6 of~\cite{Fel-book}) as there are special
 outer-vertices where the contraction can be done to reduce
the case analysis  and as Properties (P1) and (P2)
are local and not too difficult to preserve during the decontraction
process.

In the toroidal case there is a huge case analysis for the following
reasons. One as to consider the three different kind of contraction
depicted on Figure~\ref{fig:contraction}. For each of these cases, one
as to consider the different ways that the edges can be oriented and
color (around the contraction point) in $G/ e$. For each of these
cases, one as to show that the Schnyder wood $G/e$ can be extended to
a Schnyder wood of $G$. This would be quite easy if one just have to
satisfy (T1) that is a local property, but satisfying (T2) is much
more complicated.  Instead of proving (T2), that is considering
intersections between every monochromatic cycles, we prove (T2') and
(T3'), that is equivalent for our purpose by
Lemma~\ref{lem:tri-relax}. Property (T2') is simpler than (T2) as it
is considering just one intersection for each pair of colors instead
of all the intersections. Even with this simplification proving (T2')
is the main difficulty of the proof.  For each considered cases, one
has to analyze the different ways that the monochromatic cycles go
through the contracted vertex or not and to show that there always
exists a coloring and orientation of $G$ where (T2') is
satisfied. Some cases are non trivial and involved the use of lemmas
like Lemmas~\ref{lem:magic} and~\ref{lem:magic2}.

\begin{lemma}
\label{lem:magic}
Let $G$ be a toroidal map given with a Schnyder wood and $y,w$ be two
vertices of $G$ such that $e_i(y)$ is entering $w$. Suppose that there
is a directed path $Q_{i-1}$ of color $i-1$ from $y$ to $w$, and a
directed path $Q_{i+1}$ of color $i+1$ from $y$ to $w$. Consider the
two directed cycles $C_{i-1}=Q_{i-1}\cup \{e_i(y)\}^{-1}$ and
$C_{i+1}=Q_{i+1}\cup \{e_i(y)\}^{-1}$.  Then $C_{i-1}$ and $C_{i+1}$
are not homotopic.
\end{lemma}

\begin{proof}
  By Lemma~\ref{lem:nocontractiblecycle}, the cycles $C_{i-1}$ and
  $C_{i+1}$ are not contractible.  Suppose that $C_{i-1}$ and
  $C_{i+1}$ are homotopic.  The path $Q_{i+1}$ is leaving $C_{i-1}$ at
  $y$ on the right side of $C_{i-1}$.  Since $C_{i-1}$ and $C_{i+1}$
  are homotopic, the path $Q_{i+1}$ is entering $C_{i-1}$ at least
  once from its right side. This is in contradiction with (T1).
\end{proof}

The sector $[e_1,e_2]$ of a vertex $w$, for $e_1$ and $e_2$ two edges
incident to $w$, is the counterclockwise sector of $w$ between $e_1$
and $e_2$, including the edge $e_1$ and $e_2$. The sector $]e_1,e_2]$,
$[e_1,e_2[$, and $]e_1,e_2[$ are defined analogously by excluding the
corresponding edges from the sectors.

\begin{lemma}
\label{lem:magic2}
Let $G$ be a toroidal map given with a Schnyder wood and $w,x,y$ be
three vertices such that $e_{i-1}(x)$ and $e_{i+1}(y)$ are entering
$w$.  Suppose that there is a directed path $Q_{i-1}$ of color $i-1$
from $y$ to $w$, entering $w$ in the sector $[e_i(w),e_{i-1}(x)]$, and
a directed path $Q_{i+1}$ of color $i+1$ from $x$ to $w$, entering $w$
in the sector $[e_{i+1}(y),e_i(w)]$.  Consider the two directed cycles
$C_{i-1}=Q_{i-1}\cup \{e_{i+1}(y)\}^{-1}$ and $C_{i+1}=Q_{i+1}\cup
\{e_{i-1}(x)\}^{-1}$.  Then either $C_{i-1}$ and $C_{i+1}$ are not
homotopic or $C_{i-1}=C_{i+1}^{-1}$.
\end{lemma}

\begin{proof} 
  By Lemma~\ref{lem:nocontractiblecycle}, the cycles $C_{i-1}$ and
  $C_{i+1}$ are not contractible.  Suppose that $C_{i-1}$ and
  $C_{i+1}$ are homotopic and that $C_{i-1}\neq C_{i+1}^{-1}$.  Since
  $C_{i-1}\neq C_{i+1}^{-1}$, the cycle $C_{i+1}$ is leaving
  $C_{i-1}$. By (T1) and the assumption on the sectors, $C_{i+1}$ is
  leaving $C_{i-1}$ on its right side. Since $C_{i-1}$ and $C_{i+1}$
  are homotopic, the path $Q_{i+1}$ is entering $C_{i-1}$ at least
  once from its right side. This is in contradiction with (T1).
\end{proof}

We are now able to prove Lemma~\ref{lem:contractype1}.

\emph{Proof of Lemma~\ref{lem:contractype1}.}\
  Let $u,v$ be the two extremities of $e$. Vertices $u$ and $v$ are of
  degree at least three. Let $x,y$ (resp. $z,t$) be the neighbors of
  $u$ (resp. $v$) such that $x,v,y$ (resp. $z,u,t$) appear
  consecutively and in counterclockwise order around $u$ (resp. $v$)
  (see Figure~\ref{fig:contraction}.(c)).  Note that $u$ and $v$ are
  distinct by definition of edge contraction but that $x,y,z,t$ are
  not necessarily distinct, nor necessarily distinct from $u$ and $v$. Depending
  on whether the faces incident to $e$ are triangles or not, we are,
  by symmetry, in one of the three cases of
  Figure~\ref{fig:contraction}.(a).(b).(c). Let $G'=G/e$ and consider
  a Schnyder wood of Type 1 of $G'$. Let $w$ be the vertex of $G'$
  resulting from the contraction of $e$.

  For each case (a), (b), (c).  There are different cases
  corresponding to the different possibilities of orientation and
  coloring of the edges $e_{wx}, e_{wy}, e_{wz}, e_{wt}$ in $G'$.  For
  example for case (a), there should be $6$ cases depending on if
  $e_{wx}$ and $e_{wy}$ are both entering $w$, both leaving $w$ or one
  entering $w$ and one leaving $w$ (3 cases), multiplied by the
  coloring, both of the same or not (2 cases). The case where $w$ has
  two edges leaving in the same color is impossible by (T1). So, by
  symmetry, only $5$ cases remain represented by Figure~$a.k.0$, for
  $k=1,\ldots,5$, on Figure~\ref{fig:case-a} (in the notation $\alpha.k.l$,
  $a.k$ indicates the line on the figures and $\ell$ the column).  For
  cases (b) and (c), there are more cases to consider but the analysis
  is similar. These cases are represented in the first columns of
  Figures~\ref{fig:case-b} and~\ref{fig:case-c}.  On these figure, a
  dotted half-edge represent the possibility for an edge to be uni- or
  bi-directed. In the last case of each figure, we have indicated
  where is the edge leaving in color $1$ as there are two
  possibilities (up or down).

\begin{figure}[!h]
\center
  \begin{tabular}[]{ccccccccc}
 & &  0 & & 1 & & 2 & & 3\\ \\
a.1& \hspace{2em} & \includegraphics[scale=0.4]{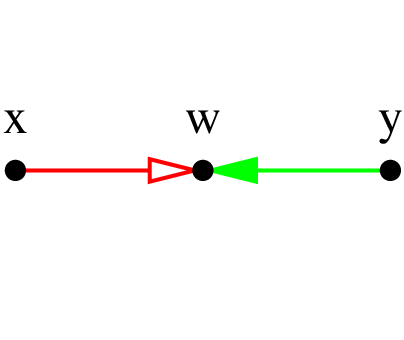} & \hspace{0em} &
\includegraphics[scale=0.4]{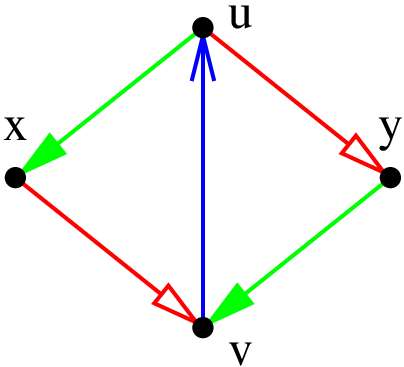} & \hspace{0em} &
\includegraphics[scale=0.4]{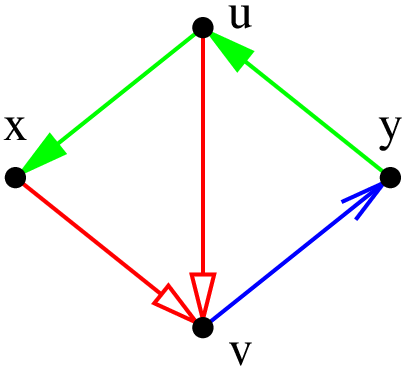} & \hspace{0em} &
\includegraphics[scale=0.4]{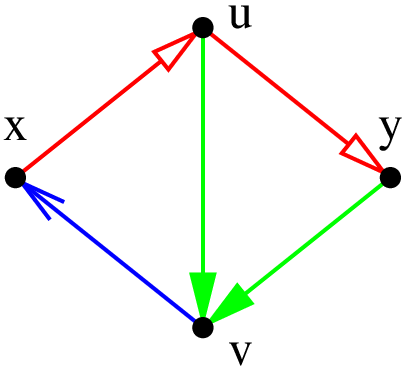} \\

a.2 & & \includegraphics[scale=0.4]{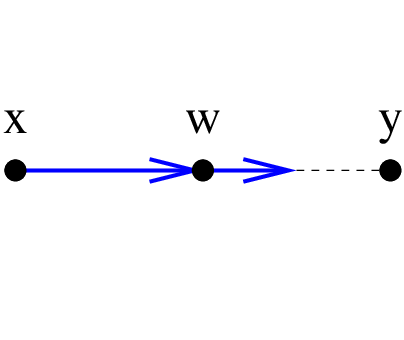} & \hspace{0em} &
\includegraphics[scale=0.4]{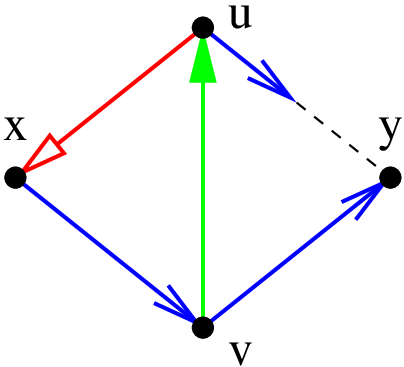} & \hspace{0em} &
\includegraphics[scale=0.4]{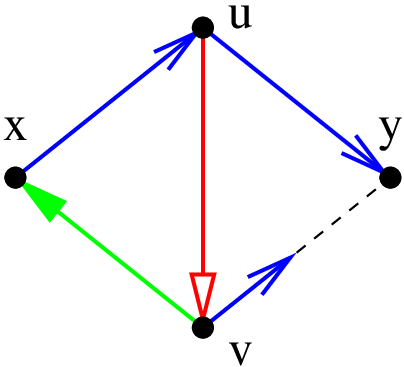} \\
 
a.3 & &\includegraphics[scale=0.4]{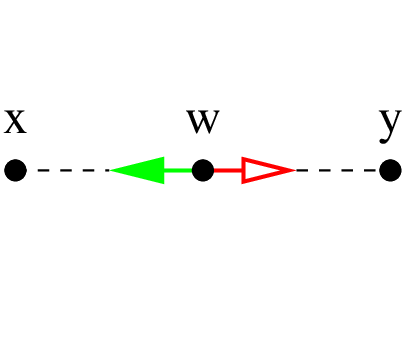} & \hspace{0em} &
\includegraphics[scale=0.4]{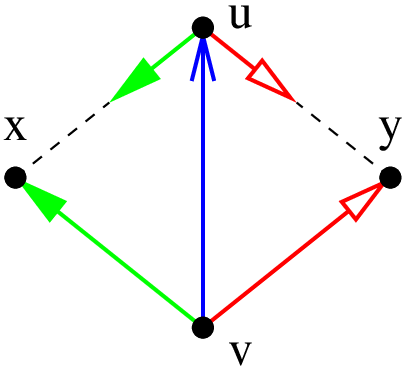} \\

a.4 & & \includegraphics[scale=0.4]{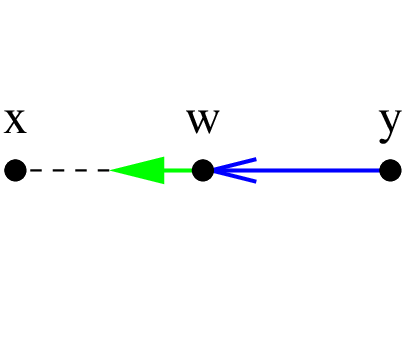} & \hspace{0em} &
\includegraphics[scale=0.4]{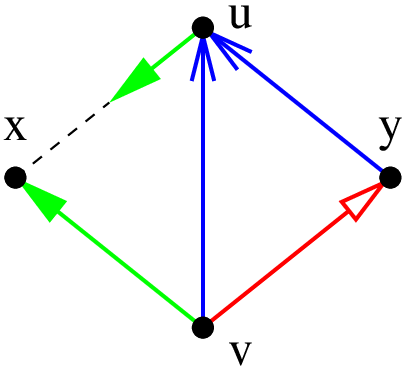} \\

a.5 & & \includegraphics[scale=0.4]{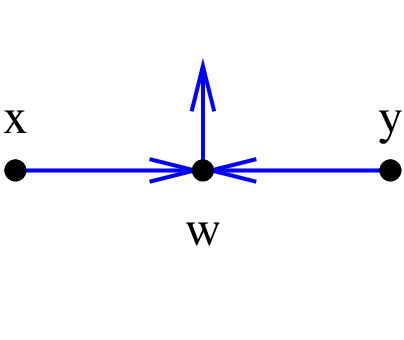} & \hspace{0em} &
\includegraphics[scale=0.4]{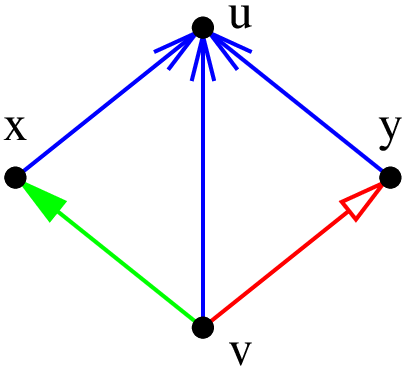} \\
  \end{tabular}

  \caption{Decontraction rules for case (a).}
\label{fig:case-a}
\end{figure}

\begin{figure}[!h]
\center
  \begin{tabular}[]{ccccccc}
 & &  0 & & 1 & & 2 \\ \\
b.1& \hspace{2em} & \includegraphics[scale=0.2]{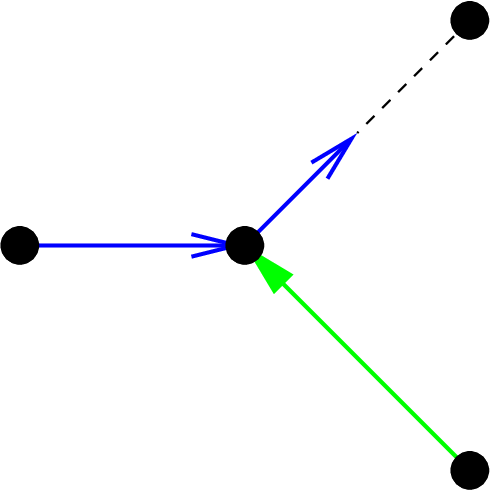} & \hspace{2em} &
\includegraphics[scale=0.2]{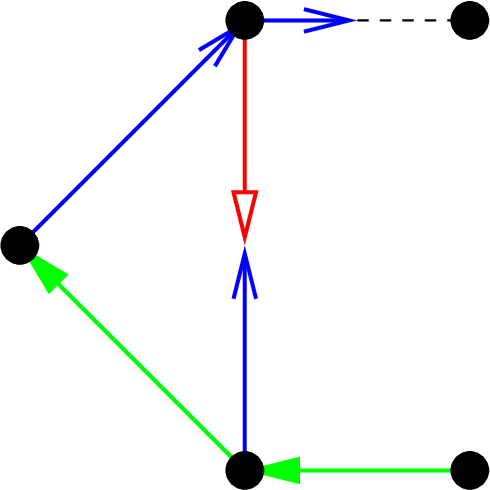} & \hspace{2em} &
\includegraphics[scale=0.2]{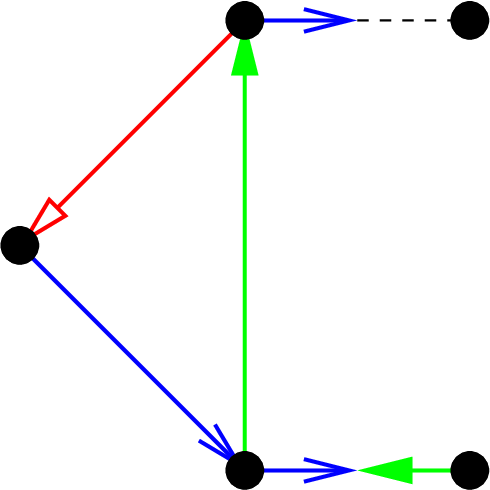} \\

b.2 & &\includegraphics[scale=0.2]{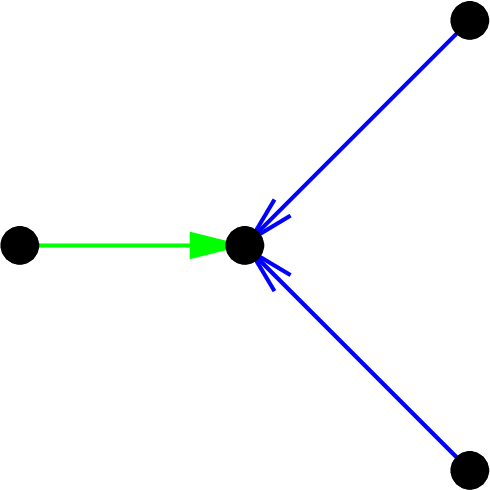} & \hspace{2em} &
\includegraphics[scale=0.2]{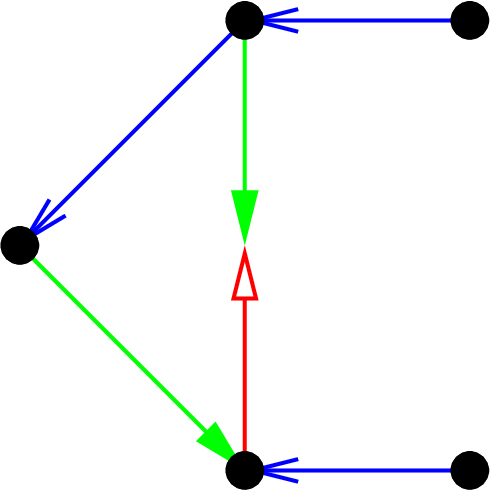} & \hspace{2em} &
\includegraphics[scale=0.2]{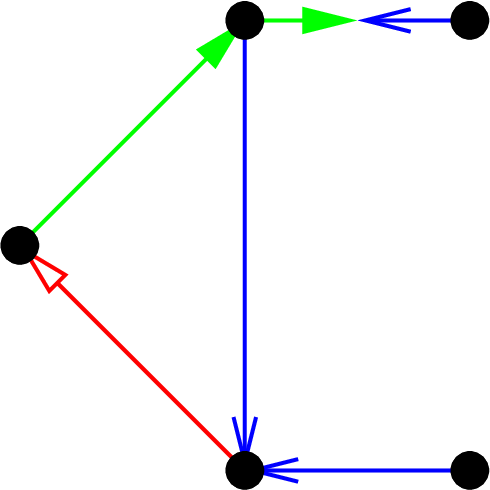} \\

b.3 & & \includegraphics[scale=0.2]{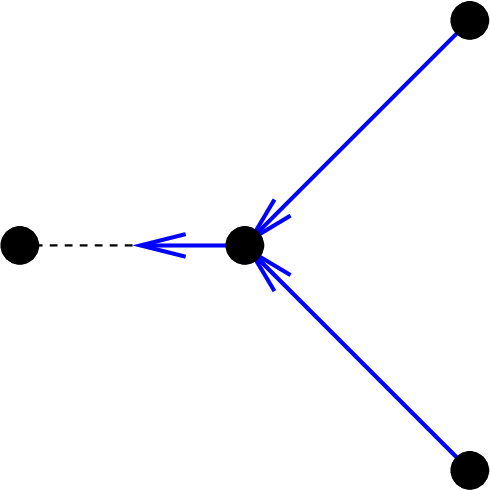} & \hspace{2em} &
\includegraphics[scale=0.2]{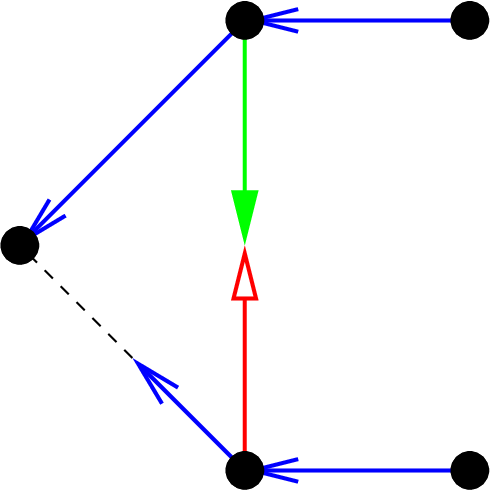} & \hspace{2em} &
\includegraphics[scale=0.2]{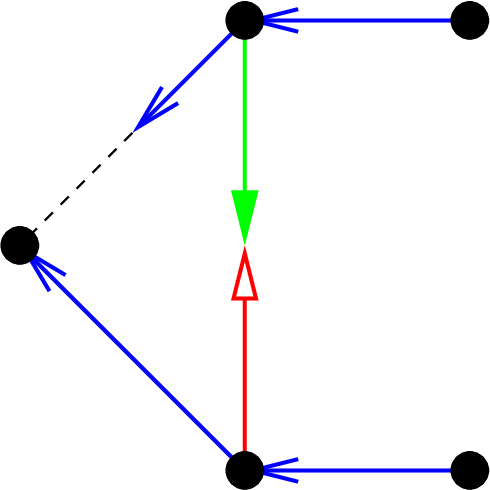} \\

b.4  & & \includegraphics[scale=0.2]{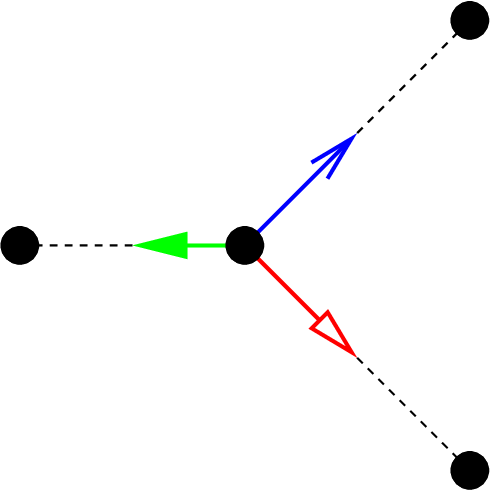} & \hspace{2em} &
\includegraphics[scale=0.2]{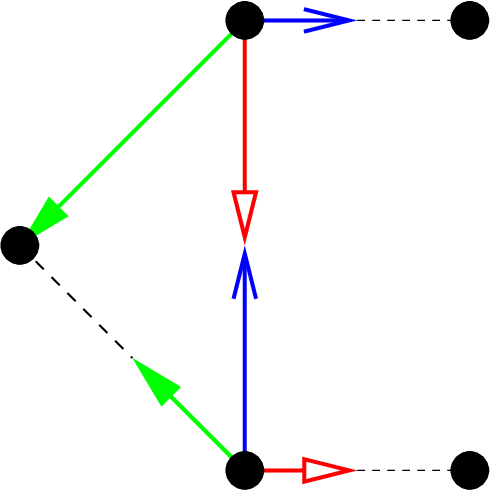} & \hspace{2em} &
\includegraphics[scale=0.2]{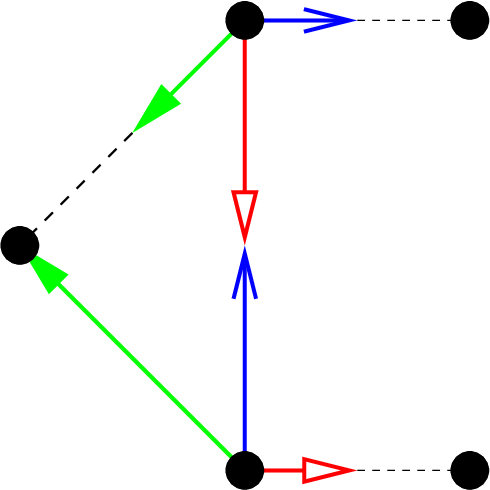} \\

b.5 & & \includegraphics[scale=0.2]{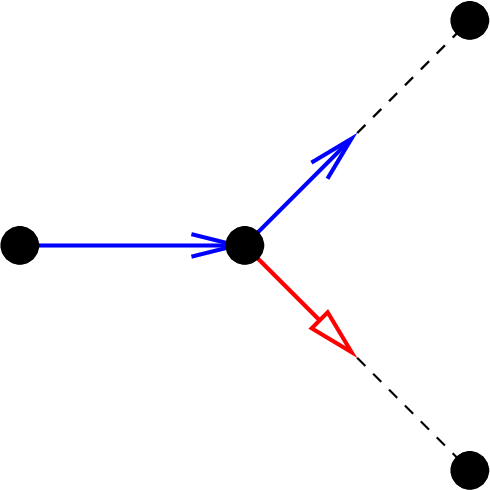} & \hspace{2em} &
\includegraphics[scale=0.2]{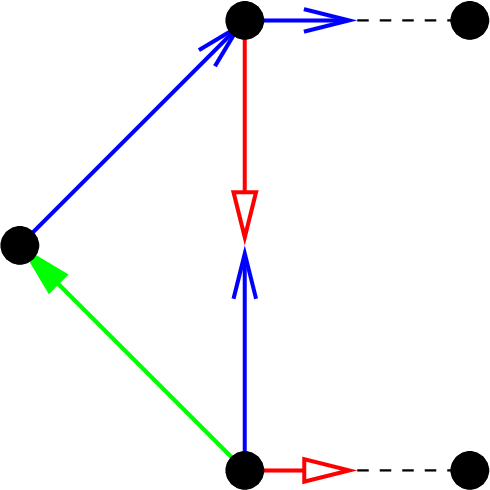} \\

b.6&  & \includegraphics[scale=0.2]{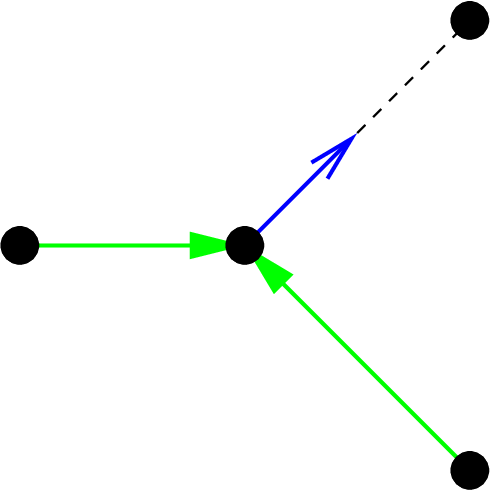} & \hspace{2em} &
\includegraphics[scale=0.2]{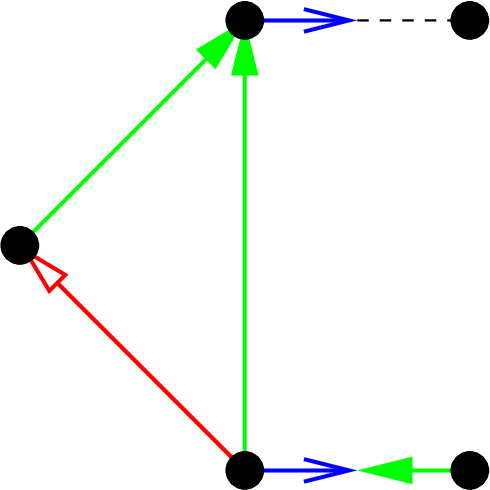} \\

b.7 & &\includegraphics[scale=0.2]{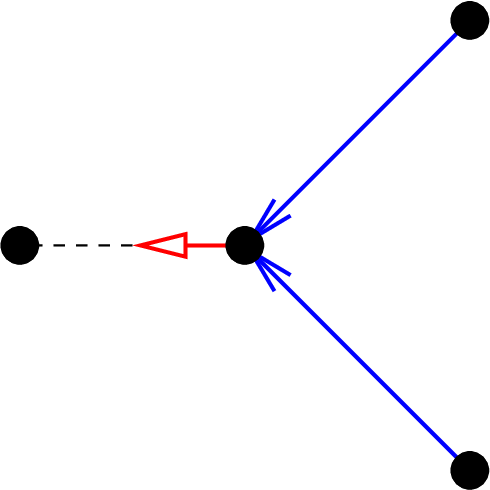} & \hspace{2em} &
\includegraphics[scale=0.2]{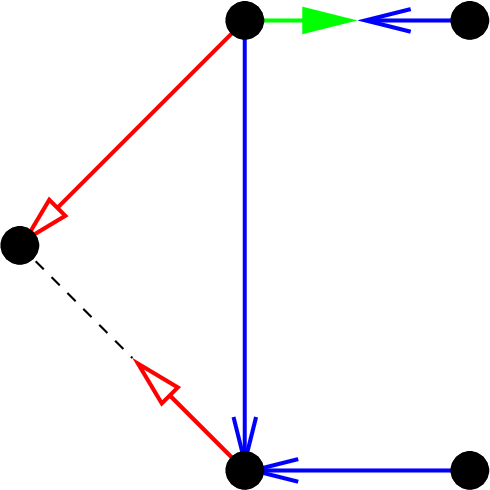} \\

b.8&  &\includegraphics[scale=0.2]{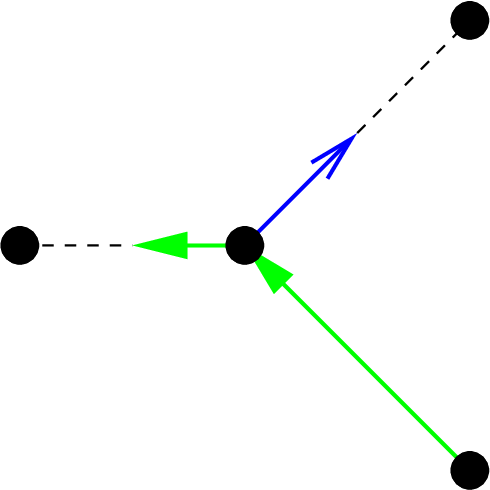} & \hspace{2em} &
\includegraphics[scale=0.2]{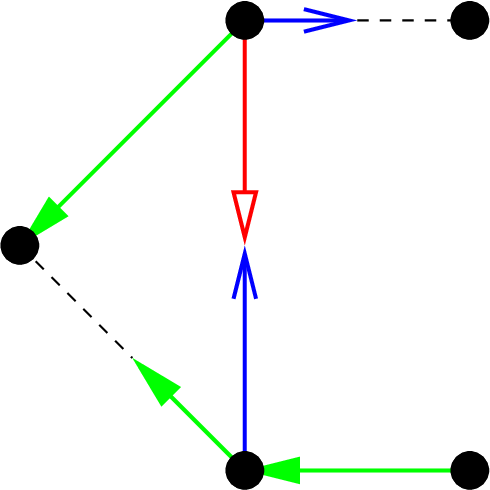} \\

b.9 & &\includegraphics[scale=0.2]{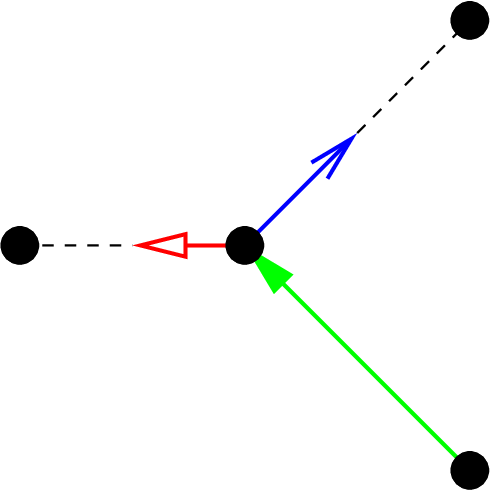} & \hspace{2em} &
\includegraphics[scale=0.2]{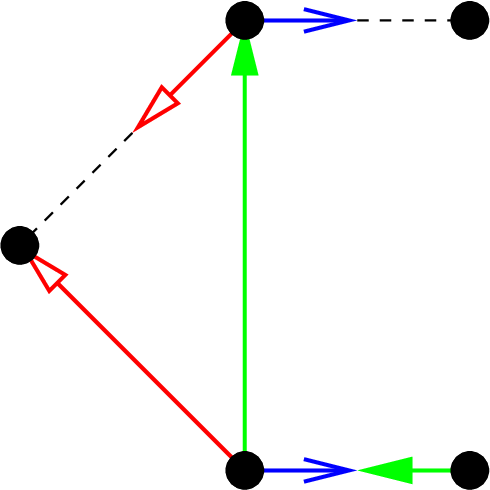} \\

b.10&  &\includegraphics[scale=0.2]{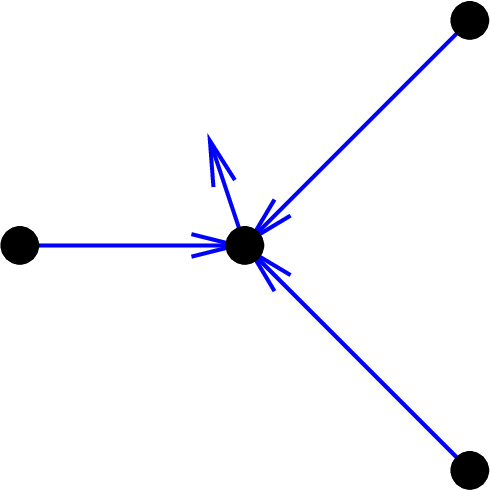} & \hspace{2em} &
\includegraphics[scale=0.2]{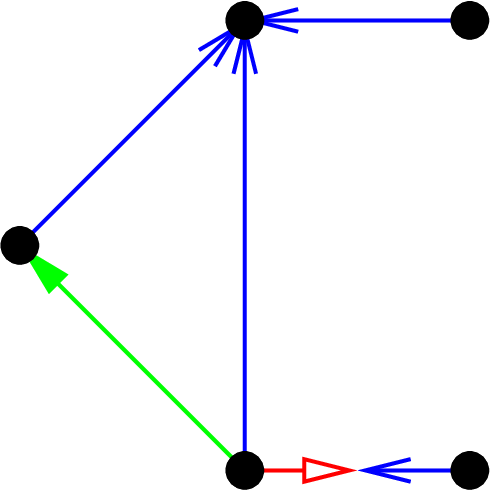} \\

  \end{tabular}
\caption{Decontraction rules for case (b)}
\label{fig:case-b}
\end{figure}

\begin{figure}[!h]
\center
  \begin{tabular}[]{ccccccc}
 & &  0 & & 1 & & 2 \\ \\
c.1& \hspace{2em} & \includegraphics[scale=0.2]{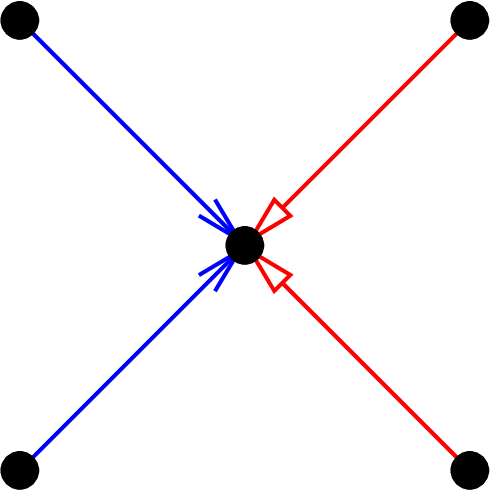} & \hspace{2em} &
\includegraphics[scale=0.2]{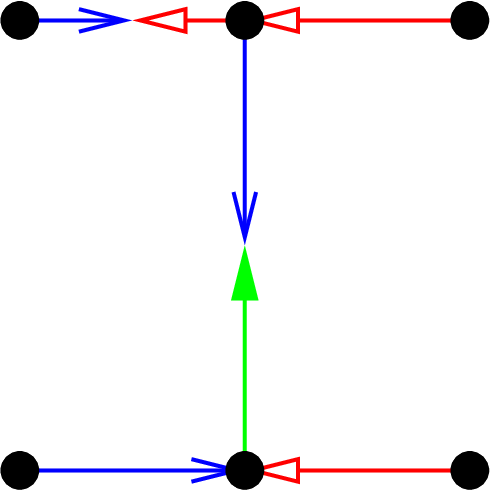} & \hspace{2em} &
\includegraphics[scale=0.2]{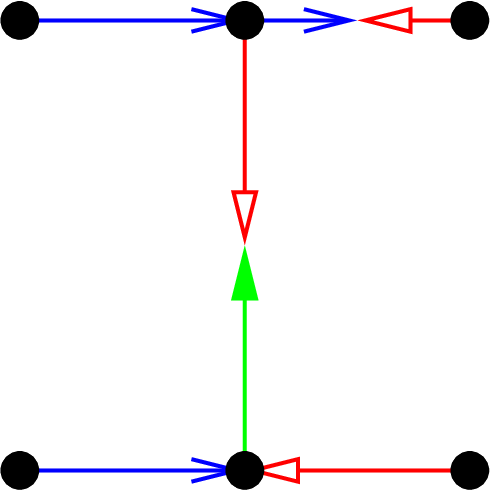} \\

c.2 & & \includegraphics[scale=0.2]{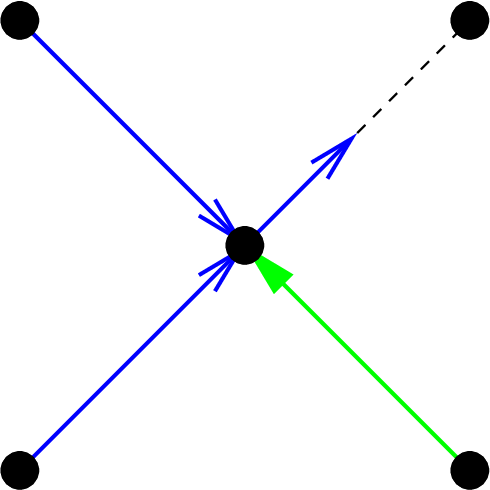} & \hspace{2em} &
\includegraphics[scale=0.2]{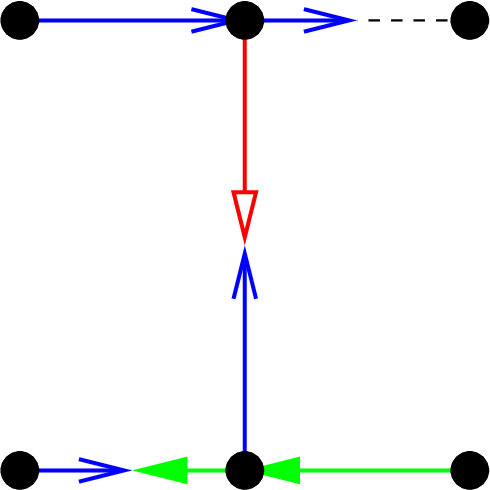} & \hspace{2em} &
\includegraphics[scale=0.2]{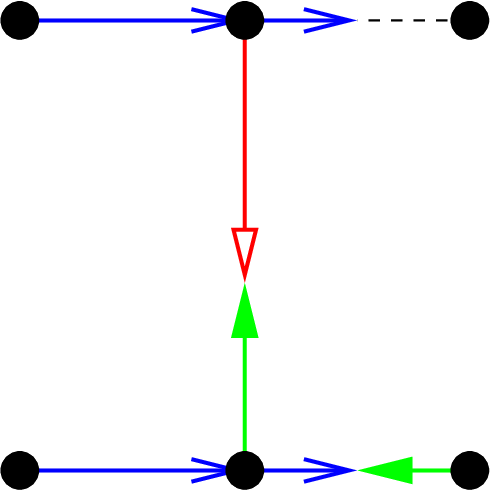} \\

c.3 & & \includegraphics[scale=0.2]{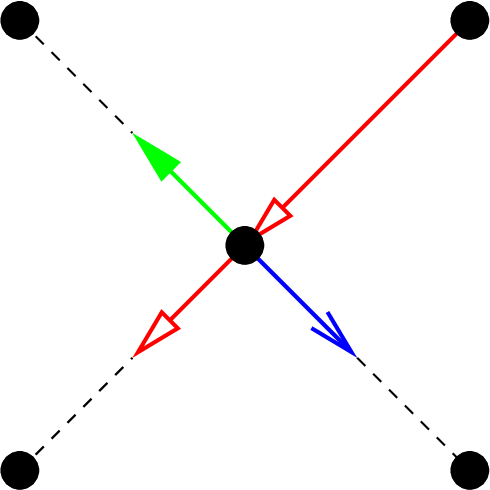} & \hspace{2em} &
\includegraphics[scale=0.2]{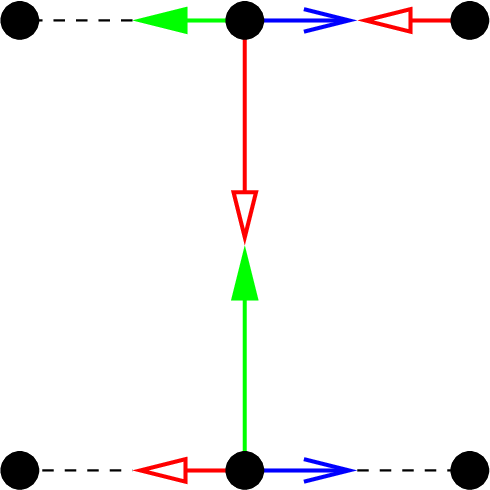} \\

c.4 & & \includegraphics[scale=0.2]{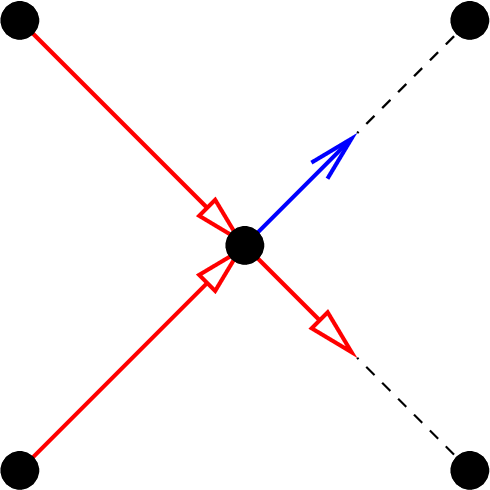} & \hspace{2em} &
\includegraphics[scale=0.2]{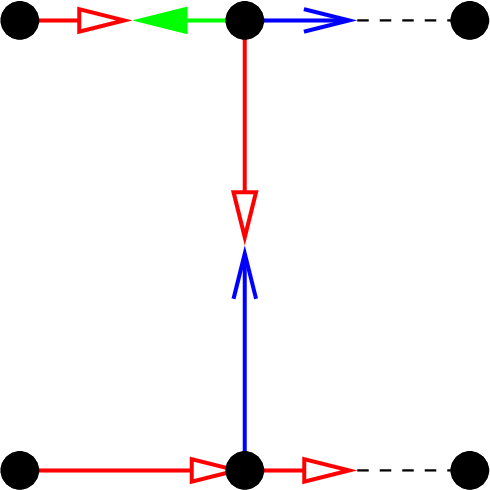} \\

c.5 & & \includegraphics[scale=0.2]{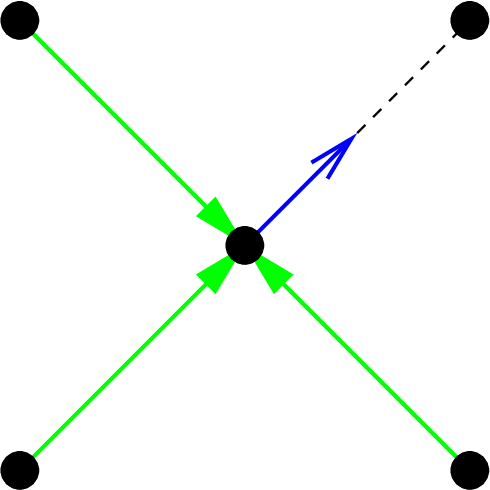} & \hspace{2em} &
\includegraphics[scale=0.2]{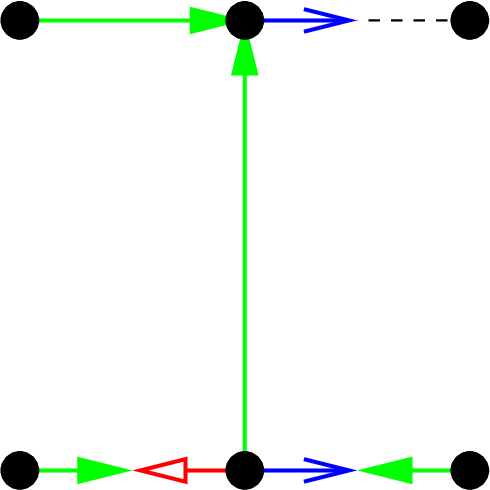} \\

c.6 & & \includegraphics[scale=0.2]{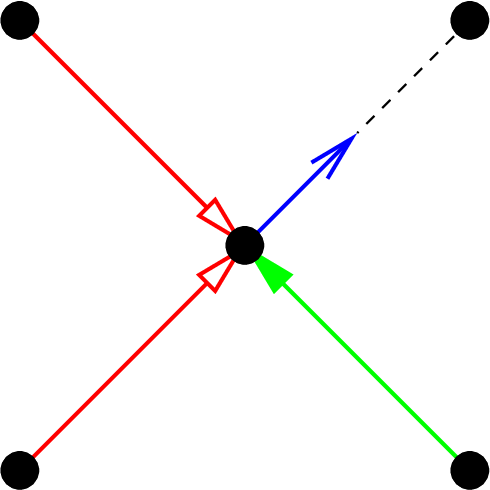} & \hspace{2em} &
\includegraphics[scale=0.2]{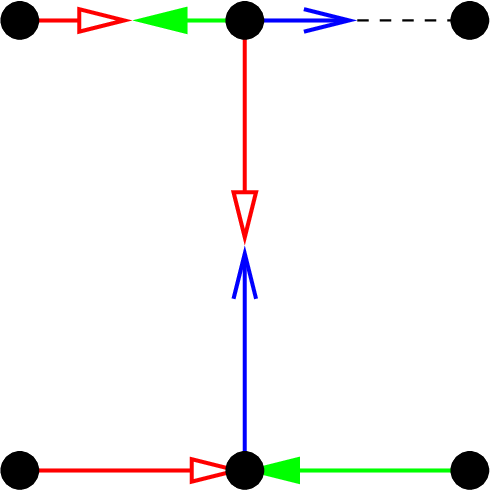} \\

c.7 & & \includegraphics[scale=0.2]{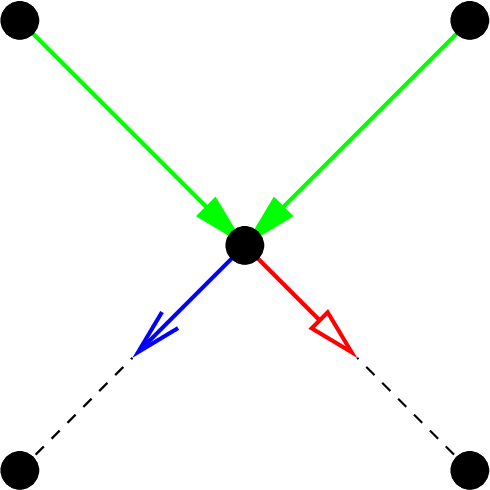} & \hspace{2em} &
\includegraphics[scale=0.2]{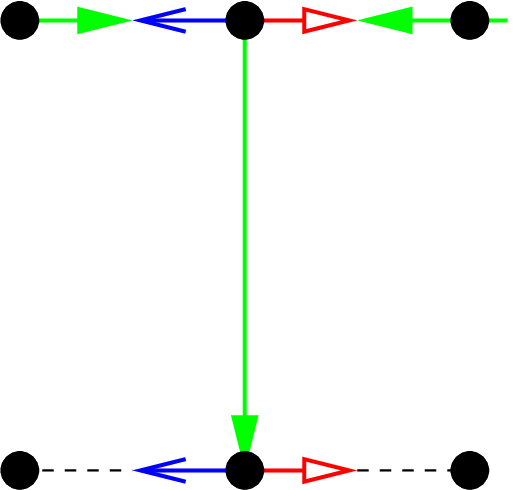} \\

c.8 & & \includegraphics[scale=0.2]{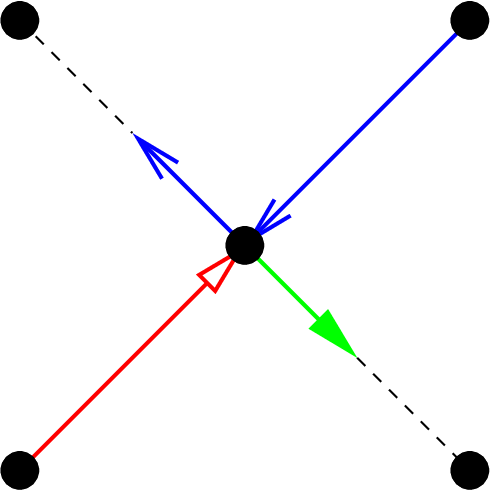} & \hspace{2em} &
\includegraphics[scale=0.2]{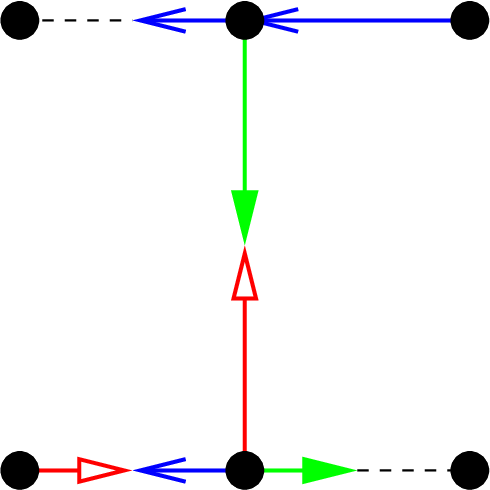} \\

c.9 & & \includegraphics[scale=0.2]{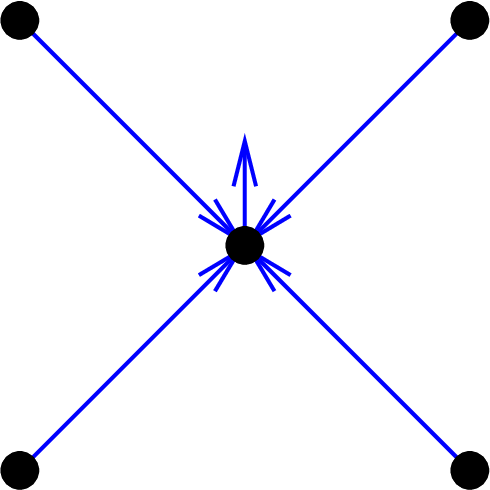} & \hspace{2em} &
\includegraphics[scale=0.2]{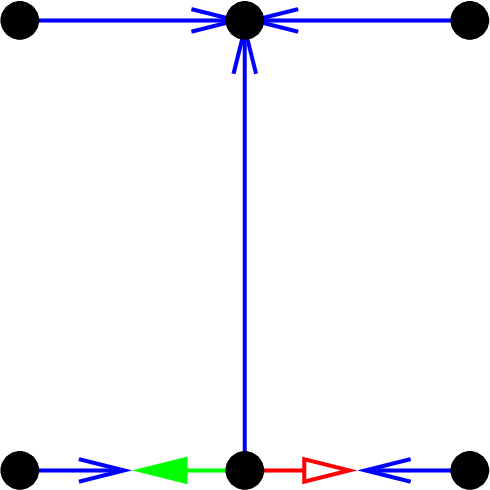} \\
  \end{tabular}
\caption{Decontraction rules for case (c)}
\label{fig:case-c}
\end{figure}

  In each case $\alpha.k$, $\alpha\in\{a,b,c\}$, we show how one can
  color and orient the edges of $G$ to obtain a Schnyder wood of $G$
  from the Schnyder wood of $G'$. Just the edges
  $e,e_{ux},e_{uy},e_{vt},e_{vz}$ of $G$ have to be specified, all the
  other edges of $G$ keep the orientation and coloring of their
  corresponding edge in $G'$.  In each case $\alpha.k$, there might be
  several possibilities to color and orient these edges to satisfy
  (T1). Just some of these possibilities, the one that are useful for
  our purpose, are represented on Figures $\alpha.k.\ell$, $\ell\geq
  1$, of Figures~\ref{fig:case-a} to~\ref{fig:case-c}. A dotted half-edge represent
  the fact that the edge is uni- or bi-directed like the corresponding
  half edge of $G'$.

  In each case $\alpha.k$, $\alpha\in\{a,b,c\}$, we show that one of
  the colorings $\alpha.k.\ell$, $\ell\geq 1$, gives a Schnyder wood
  of Type 1 of $G$. By Lemma~\ref{lem:tri-relax}, we just have to
  prove that, in each case $\alpha.k$, there is one coloring
  satisfying (T1'), (T2') and (T3'). Properties (T1') and (T3') are
  satisfied for any colorings $\alpha.k.\ell$, $\ell\geq 1$ but this
  is not the case for property (T2'). This explains why several
  possible colorings of $G$ have to be considered.

  (T1') One can easily check that in all the cases $\alpha.k.\ell$,
  Property (T1') is satisfied for every vertex of $G$. To do so one
  can consider the angle labeling around vertices $w,x,y,(z),(t)$ of
  $G'$ in the case $\alpha.k.0$. Then one can remark that this angle
  labeling exports well around vertices $u,v,x,y,(z),(t)$ of $G$ and
  thus the Schnyder property is satisfied for these vertices. On
  Figure~\ref{fig:exampledecontractionangle}, an example is given on
  how the angle labeling is modified during the decontraction process.
  It corresponds to case $c.2.1$ where the dotted half edge is
  uni-directed.  We do not detail more this part that is easy to
  check.

\begin{figure}[!h]
\center
  \begin{tabular}[]{ccccc}
\includegraphics[scale=0.4]{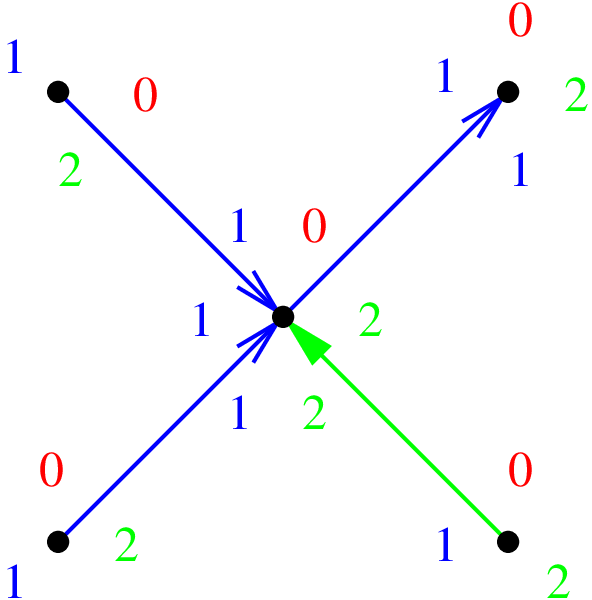} & \hspace{1em} &
\includegraphics[scale=0.4]{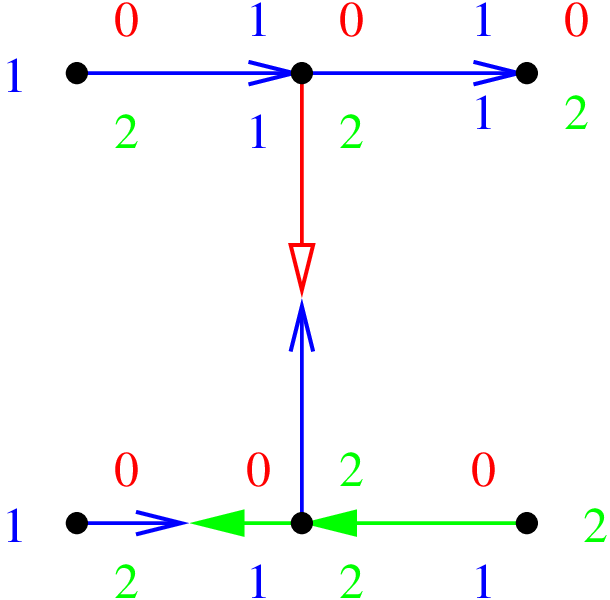} \\

\end{tabular}

\caption{Example on (T1') preservation during
  decontraction.}
\label{fig:exampledecontractionangle}
\end{figure}

(T3') One can easily check that in all the cases $\alpha.k.\ell$,
Property (T3') is satisfied for $G$. If (T3') is not satisfied in $G$
after applying one of the coloring $\alpha.k.\ell$, then there exists
two monochromatic cycles $C,C'$ of different colors that are reversal.
By property (T3') of Lemma~\ref{lem:tri-relax}, there is no reversal
cycles in $G'$. Thus $C,C'$ have to use a bi-directed edge $e'$ of the
figures that is newly created and distinct from edge $e$ and distinct
from half-dotted edges (otherwise the cycles are still reversal when
$e$ is contracted). Just some cases have such an edge and one can
remark that for all this cases, the cycle, after entering $u$ or $v$
by edge $e'$ have to use the edge $e$ that is either uni-directed or
bi-directed with different colors than $e'$, a contradiction. For
example in case $c.2.1$ of Figure~\ref{fig:exampleT3}, two reversal
cycles of $G$ that do not corresponds to reversal cycles when $e$ is
contracted, have to use edge $e_{vt}$. Then one of the two cycle is
entering $v$ by $e_{vt}$ in color $1$ and thus has to continue by
using the only edge leaving $v$ in color $1$, edge $e$. As $e$ and
$e_{vt}$ are colored differently, this is no possible.  We do not
detail more this part that is easy to check.

\begin{figure}[!h]
\center
  \begin{tabular}[]{ccccc}
\includegraphics[scale=0.4]{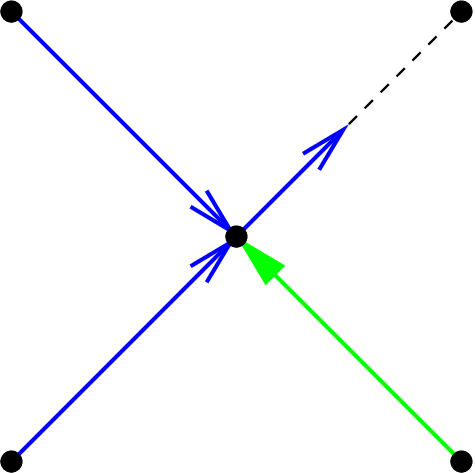} & \hspace{1em} &
\includegraphics[scale=0.4]{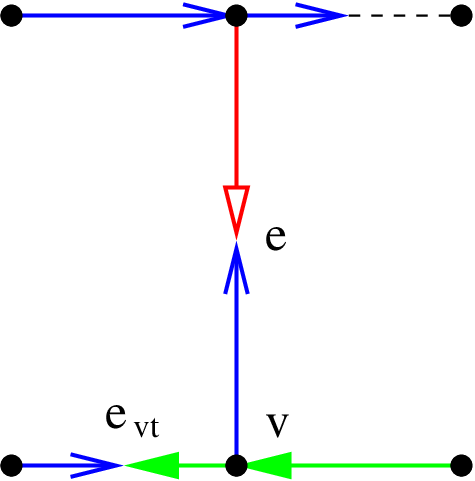} \\
\end{tabular}

\caption{Example on (T3') preservation during decontraction.}
\label{fig:exampleT3}
\end{figure}

  (T2') Proving Property (T2') is the main difficulty of the proof. For
  each case $\alpha.k$, there is a case analysis considering the
  different ways that the monochromatic cycles of $G'$ go through $w$
  or not.  We say that a monochromatic cycle $C$ of $G'$ is
  \emph{safe} if $C$ does not contain $w$. Depending on whether there
  are safe monochromatic cycles or not of each colors, it may be a
  different case $\alpha.k.\ell$ and a different argument that is used
  to prove that property (T2') is preserved.

\noindent
\emph{$\bullet$ Case a.1: $e_{wx}$ and $e_{wy}$ are entering $w$ in
  different color.}

We can assume by symmetry that $e_{wx}=e_0(x)$ and $e_{wy}=e_2(y)$
(case a.1.0 of Figure~\ref{fig:case-a}). We apply one of the colorings
a.1.1, a.1.2 and a.1.3 of Figure~\ref{fig:case-a}.

We have a case analysis corresponding to whether there are
monochromatic cycles of $G'$ that are safe.

\noindent\emph{$\star$ Subcase $a.1.\{0,1,2\}$: There are safe monochromatic cycles of
  colors $\{0,1,2\}$.}

Let $C'_0,C'_1,C'_2$ be safe monochromatic cycles of color $0,1,2$ in
$G'$. As the Schnyder wood of $G'$ is of type 1, they pairwise
intersects in $G'$.  Apply the coloring a.1.1 on $G$.  As
$C'_0,C'_1,C'_2$ do not contain vertex $w$, they are not modified in
$G$. Thus they still pairwise intersect in $G$. So (T2') is satisfied.

\noindent\emph{$\star$ Subcase $a.1.\{0,2\}$: There are safe monochromatic cycles of
  colors exactly $\{0,2\}$.}

Let $C'_0,C'_2$ be safe monochromatic cycles of color $0,2$ in
$G'$. Let $C'_1$ be a $1$-cycle in $G'$.  As the Schnyder wood of $G'$
is of type 1, $C'_0,C'_1,C'_2$ pairwise intersects in $G'$. None of
those intersections contain $w$ as $C'_0$ and $C'_2$ do not contain
$w$.  By (T1), the cycle $C'_1$ enter $w$ in the sector
$]e_{wx},e_{wy}[$ and leaves in the sector $]e_{wy},e_{wx}[$.  Apply
the coloring a.1.1 on $G$. The cycle $C'_1$ is replaced by a new cycle
$C_1=C'_1\setminus\{w\}\cup\{u,v\}$.  The cycles $C'_0,C'_1,C'_2$ were
intersecting outside $w$ in $G'$ so $C'_0,C_1,C'_2$ are intersecting
in $G$. So (T2') is satisfied.

\noindent\emph{$\star$ Subcase $a.1.\{1,2\}$: There are safe monochromatic cycles of
  colors exactly $\{1,2\}$.}

Let $C'_1,C'_2$ be safe monochromatic cycles of color $1,2$ in
$G'$. Let $C'_0$ be a $0$-cycle in $G'$.  The cycles $C'_0,C'_1,C'_2$
pairwise intersects outside $w$.  The cycle $C'_0$ enters $w$ in the
sector $[e_1(w),e_{wx}[$, $[e_{wx},e_{wx}]$ or
$]e_{wx},e_2(w)]$. Apply the coloring a.1.2 
on $G$.  Depending on which of the three sectors $C'_0$ enters, it is
is replaced by one of the three following cycle
$C_0=C'_0\setminus\{w\}\cup\{u,v\}$,
$C_0=C'_0\setminus\{w\}\cup\{x,v\}$,
$C_0=C'_0\setminus\{w\}\cup\{v\}$. In any of the three possibilities,
$C_0,C'_1,C'_2$ are intersecting in $G$. So (T2') is satisfied.

\noindent\emph{$\star$ Subcase $a.1.\{0,1\}$: There are safe
  monochromatic cycles of colors exactly $\{0,1\}$.}

This case is completely symmetric to the case $a.1.\{1,2\}$.

\noindent\emph{$\star$ Subcase $a.1.\{2\}$: There are safe
  monochromatic cycles of color $2$ only.}

Let $C'_2$ be a safe $2$-cycle in $G'$.  Let $C'_0,C'_1$ be
monochromatic cycles of color $0,1$ in $G'$.

Suppose that there exists a path $Q'_0$ of color $0$, from $y$ to $w$
such that this path does not intersect $C'_2$. Suppose also that there
exists a path $Q'_1$ of color $1$, from $y$ to $w$ such that this path does not intersect
$C'_2$.  Let $C''_{0}=Q'_{0}\cup \{e_{wy}\}$ and $C''_{1}=Q'_{1}\cup
\{e_{wy}\}$. By Lemma~\ref{lem:nocontractiblecycle}, $C''_{0}, C''_1,
C'_{2}$ are not contractible. Both of $C''_{0}, C''_{1}$ does not
intersect $C'_2$ so by Lemma~\ref{lem:intersect2}, they are both
homotopic to $C'_2$. Thus by Lemma~\ref{lem:homotopic3}, cycles
$C''_{0}, C''_{1}$ are homotopic to each other, contradicting
Lemma~\ref{lem:magic} (with $i=2,w,y,Q'_0,Q'_1$).  So we can assume
that one of $Q'_0$ or $Q'_1$ as above does not exist.

Suppose that in $G'$, there does not exist a path of color $0$, from
$y$ to $w$ such that this path does not intersect $C'_2$.  Apply the
coloring a.1.1 on $G$. Cycle $C'_1$ is replaced by
$C_1=C'_1\setminus\{w\}\cup\{u,v\}$, and intersect $C'_2$. Let $C_0$
be a $0$-cycle of $G$.  Cycle $C_0$ has to contain $u$ or $v$ or both,
otherwise it is a safe cycle of $G'$ of color $0$. In any case it
intersects $C_1$.  If $C_0$ contains $v$, then
$C'_0=C_0\setminus\{v\}\cup\{w\}$ and so $C_0$ is intersecting $C'_2$
and (T2') is satisfied.  Suppose now that $C_0$ does not contain $v$.
Then $C_0$ contains $u$ and $y$, the extremity of the edge leaving $u$
in color $0$. Let $Q_0$ be the part of $C_0$ consisting of the path
from $y$ to $u$. The path $Q'_0=Q_0\setminus\{u\}\cup\{w\}$ is from
$y$ to $w$. Thus by assumption $Q'_0$ intersects $C'_2$. So $C_0$
intersects $C'_2$ and (T2') is satisfied.

Suppose now that in $G'$, there does not exist a path of color $1$,
from $y$ to $w$ such
that this path does not intersect $C'_2$. Apply the coloring a.1.2 on
$G$.  Depending on which of the three sectors $C'_0$ enters,
$[e_1(w),e_{wx}[$, $[e_{wx},e_{wx}]$ or $]e_{wx},e_2(w)]$, it is is
replaced by one of the three following cycle
$C_0=C'_0\setminus\{w\}\cup\{u,v\}$,
$C_0=C'_0\setminus\{w\}\cup\{x,v\}$,
$C_0=C'_0\setminus\{w\}\cup\{v\}$. In any of the three possibilities,
$C_0$ contains $v$ and intersect $C'_2$.  Let $C_1$ be a $1$-cycle of
$G$.  Cycle $C_1$ has to contain $u$ or $v$ or both, otherwise it is a
safe cycle of $G'$ of color $1$.  Vertex $u$ has no edge entering it
in color $1$ so $C_1$ does not contain $u$ and thus it contains $v$
and intersects $C_0$.  Then $C_1$ contains $y$, the extremity of the
edge leaving $v$ in color $1$. Let $Q_1$ be the part of $C_1$
consisting of the path from $y$ to $v$. The path
$Q'_1=Q_1\setminus\{v\}\cup\{w\}$ is from $y$ to $w$. Thus by assumption $Q'_1$ intersects
$C'_2$. So $C_1$ intersects $C'_2$ and (T2') is satisfied.

\noindent\emph{$\star$ Subcase $a.1.\{0\}$: There are safe
  monochromatic cycles of color $0$ only.}

This case is completely symmetric to the case $a.1.\{2\}$.

\noindent\emph{$\star$ Subcase $a.1.\{1\}$: There are safe
  monochromatic cycles of color  $1$ only.}

Let $C'_1$ be a safe $1$-cycle in $G'$.  Let $C'_0$ and $C'_2$ be
monochromatic cycles of color $0$ and $2$ in $G'$.

Suppose $C'_0$ is entering $w$ in the sector $]e_{wx},e_2(w)]$.  Apply
the coloring a.1.3 on $G$.  The $0$-cycle $C'_0$ is replaced by
$C_0=C'_0\setminus\{w\}\cup\{v\}$ and thus contains $v$ and still
intersect $C'_1$.  Depending on which of the three sectors $C'_2$
enters,$[e_0(w),e_{wy}[$, $[e_{wy},e_{wy}]$ or $]e_{wy},e_1(w)]$, it
is replaced by one of the three following cycles
$C_2=C'_2\setminus\{w\}\cup\{v\}$,
$C_2=C'_2\setminus\{w\}\cup\{y,v\}$,
$C_2=C'_2\setminus\{w\}\cup\{u,v\}$.  In any case, $C_2$ contains $v$
and still intersect $C'_1$. Cycle $C_0$ and $C_2$ intersect on $v$.
So (T2') is satisfied.

The case where $C'_2$ is entering $w$ in the sector $[e_0(w),e_{wy}[$
is completely symmetric and we apply the coloring a.1.2 on $G$.

It remains to deal with the case where $C'_0$ is entering $w$ in the
sector $[e_1(w),e_{wx}]$ and $C'_2$ is entering $w$ in the sector
$[e_{wy},e_1(w)]$.  Suppose that there exists a path $Q'_0$ of color
$0$, from $y$ to $w$, entering $w$ in the sector $[e_1(w),e_{wx}]$,
such that this path does not intersect $C'_1$. Suppose also that there
exists a path $Q'_2$ of color $2$, from $x$ to $w$, entering $w$ in
the sector $[e_{wy},e_1(w)]$, such that this path does not intersect
$C'_1$.  Let $C''_{0}=Q'_{0}\cup \{e_{wy}\}$ and $C''_{2}=Q'_{2}\cup
\{e_{wx}\}$. By Lemma~\ref{lem:nocontractiblecycle}, $C''_{0}, C'_1,
C''_{2}$ are not contractible. Cycles $C''_{0}, C''_{2}$ do not
intersect $C'_1$ so by Lemma~\ref{lem:intersect2}, they are 
homotopic to $C'_1$. Thus by Lemma~\ref{lem:homotopic3}, cycles
$C''_{0}, C''_{2}$ are homotopic to each other. Thus by
Lemma~\ref{lem:magic2} (with $i=1,w,x,y,Q'_0,Q'_2$), we have $C''_{0}=
(C''_{2})^{-1}$, contradicting (T3') in $G'$.  So we can assume that
one of $Q'_0$ or $Q'_2$ as above does not exist.  By symmetry, suppose
that in $G'$, there does not exist a path of color $0$, from $y$ to
$w$, entering $w$ in the sector $[e_1(w),e_{wx}]$, such that this path
does not intersect $C'_1$.  Apply the coloring a.1.3 on $G$.
Depending on which of the two sectors $C'_2$ enters, $[e_{wy},e_{wy}]$
or $]e_{wy},e_1(w)]$, it is is replaced by one of the two following
cycle $C_2=C'_2\setminus\{w\}\cup\{y,v\}$,
$C_2=C'_2\setminus\{w\}\cup\{u,v\}$.  In any case, $C_2$ still
intersects $C'_1$. Let $C_0$ be a $0$-cycle of $G$.  Cycle $C_0$ has to
contain $u$ or $v$ or both, otherwise it is a safe cycle of $G'$ of
color $0$.  Suppose $C_0$ does not contain $u$, then
$C'_0=C_0\setminus\{v\}\cup\{w\}$ and $C'_0$ is not entering $w$ in
the sector $[e_1(w),e_{wx}]$, a contradiction.  So $C_0$ contains
$u$. Thus $C_0$ contains $y$, the extremity of the edge leaving $u$ in
color $0$, and it intersects $C_2$.  Let $Q_0$ be the part of $C_0$
consisting of the path from $y$ to $u$. The path
$Q'_0=Q_0\setminus\{u\}\cup\{w\}$ is from $y$ to $w$ and entering $w$
in the sector $[e_1(w),e_{wx}]$. Thus by assumption $Q'_0$ intersects
$C'_1$. So $C_0$ intersects $C'_1$ and (T2') is satisfied.

\noindent\emph{$\star$ Subcase $a.1.\{\}$: There are no safe
  monochromatic cycle.}

Let $C'_0,C'_1,C'_2$ be monochromatic cycle of color $0,1,2$ in
$G'$. They all pairwise intersect on $w$.  

Suppose first that $C'_0$ is entering $w$ in the sector
$]e_{wx},e_2(w)]$.  Apply the coloring a.1.3 on $G$.  The $0$-cycle
$C'_0$ is replaced by $C_0=C'_0\setminus\{w\}\cup\{v\}$ and thus
contains $v$.  Depending on which of the three sectors $C'_2$ enters,
$[e_0(w),e_{wy}[$, $[e_{wy},e_{wy}]$ or $]e_{wy},e_1(w)]$, it is is
replaced by one of the three following cycle
$C_2=C'_2\setminus\{w\}\cup\{v\}$,
$C_2=C'_2\setminus\{w\}\cup\{y,v\}$,
$C_2=C'_2\setminus\{w\}\cup\{u,v\}$. In any case, $C_2$ contains $v$.
Let $C_1$ be a $1$-cycle in $G$. Cycle $C_1$ has to contain $u$ or $v$
or both, otherwise it is a safe cycle of $G'$ of color $1$. Vertex $u$
has no edge entering it in color $1$ so $C_1$ does not contain $u$ and
thus it contains $v$. So $C_0,C_1,C_2$ all intersect on $v$ and (T2')
is satisfied.

The case where $C'_2$ is entering $w$ in the sector $[e_0(w),e_{wy}[$
is completely symmetric and we apply the coloring a.1.2 on $G$.

It remains to deal with the case where $C'_0$ is entering $w$ in the
sector $[e_1(w),e_{wx}]$ and $C'_2$ is entering $w$ in the sector
$[e_{wy},e_1(w)]$. Apply the coloring a.1.1 on $G$.  Cycle $C'_1$ is
replaced by $C_1=C'_1\setminus\{w\}\cup\{u,v\}$.  Let $C_0$ be a
$0$-cycle in $G$.  Cycle $C_0$ has to contain $u$ or $v$ or both,
otherwise it is a safe cycle of $G'$ of color $0$. Suppose
$C_0\cap\{v,x\}=\{v\}$, then $C_0\setminus\{v\}\cup\{w\}$ is a
$0$-cycle of $G'$ entering $w$ in the sector $]e_{wx},e_2(w)]$,
contradicting the assumption on $C'_0$. Suppose $C_0$ contains $u$,
then $C_0$ contains $y$, the extremity of the edge leaving $u$ in
color $0$.  So $C_0$ contains $\{v,x\}$ or $\{u,y\}$.  Similarly $C_2$
contains $\{v,y\}$ or $\{u,x\}$.  In any case $C_0,C_1,C_2$ pairwise
intersect. So (T2') is satisfied.


\noindent
\emph{$\bullet$ Case a.2: $e_{wx}$ and $e_{wy}$ have the same color, one
  is entering $w$, the other is leaving $w$.}

We can assume by symmetry that $e_{wx}=e_1(x)$ and $e_{wy}=e_1(w)$
(case a.2.0 of Figure~\ref{fig:case-a}).  
We apply one of the colorings
a.2.1 and a.2.2 of Figure~\ref{fig:case-a}.

 We have a case analysis corresponding to whether there are
monochromatic cycles of $G'$ that are safe. 

\noindent\emph{$\star$ Subcase $a.2.\{0,1,2\}$: There are safe
  monochromatic cycles of colors $\{0,1,2\}$.}

 Let $C'_0,C'_1,C'_2$ be safe monochromatic cycles of color $0,1,2$ in
$G'$. They pairwise intersects in $G'$.  Apply the coloring a.2.1 on
$G$.  $C'_0,C'_1,C'_2$ still pairwise intersect in $G$. So (T2') is
satisfied.

\noindent\emph{$\star$ Subcase $a.2.\{0,2\}$: There are safe
  monochromatic cycles of colors exactly $\{0,2\}$.}

Let $C'_0,C'_2$ be safe monochromatic cycles of color $0,2$ in
$G'$. Let $C'_1$ be a $1$-cycle in $G'$.  Cycles $C'_0,C'_2$ still
intersects in $G$.  Apply the coloring a.2.1
on $G$. Depending on which of the three sectors $C'_1$ enters,
$[e_2(w),e_{wx}[$, $[e_{wx},e_{wx}]$ or $]e_{wx},e_0(w)]$, it is is
replaced by one of the three following cycles
$C_1=C'_1\setminus\{w\}\cup\{u,y\}$,
$C_1=C'_1\setminus\{w\}\cup\{x,v,y\}$,
$C_1=C'_1\setminus\{w\}\cup\{v,y\}$.  In any of the three
possibilities, $C_1$ still intersect both $C'_0,C'_2$.  So (T2') is
satisfied.

\noindent\emph{$\star$ Subcase $a.2.\{1,2\}$: There are safe monochromatic cycles of
  colors exactly $\{1,2\}$.}

Let $C'_1,C'_2$ be safe monochromatic cycles of color $1,2$ in
$G'$. Let $C'_0$ be a $0$-cycle in $G'$. Cycles $C'_1,C'_2$ still
intersects in $G$.  Apply the coloring a.2.2 on $G$.  Depending on
which of the two sectors $C'_0$ enters, $[e_{wy},e_{wy}]$ or
$]e_{wy},e_2(w)]$, it is replaced by one of the two following cycles
$C_0=C'_0\setminus\{w\}\cup\{y,v\}$,
$C_0=C'_0\setminus\{w\}\cup\{u,v\}$.  In any of the two
possibilities, $C_0$ still intersect both $C'_1,C'_2$.  So (T2') is
satisfied.

\noindent\emph{$\star$ Subcase $a.2.\{0,1\}$: There are safe
  monochromatic cycles of colors exactly $\{0,1\}$.}

This case is completely symmetric to the case $a.2.\{1,2\}$.

\noindent\emph{$\star$ Subcase $a.2.\{2\}$: There are safe
  monochromatic cycles of color $2$ only.}

Let $C'_2$ be a safe $2$-cycle in $G'$.  Let $C'_0,C'_1$ be
monochromatic cycles of color $0,1$ in $G'$.

Apply the coloring a.2.2 on $G$.  Depending on which of the three
sectors $C'_1$ enters, $[e_2(w),e_{wx}[$, $[e_{wx},e_{wx}]$ or
$]e_{wx},e_0(w)]$, it is is replaced by one of the three following
cycle $C_1=C'_1\setminus\{w\}\cup\{u,y\}$,
$C_1=C'_1\setminus\{w\}\cup\{x,u,y\}$,
$C_1=C'_1\setminus\{w\}\cup\{v,y\}$. Depending on which of the two
sectors $C'_0$ enters, $[e_{wy},e_{wy}]$ or $]e_{wy},e_2(w)]$, it is
replaced by one of the two following cycles
$C_0=C'_0\setminus\{w\}\cup\{y,v\}$,
$C_0=C'_0\setminus\{w\}\cup\{u,v\}$.  In any case, $C_0$ and $C_1$
intersects each other and intersect $C'_2$.  So (T2')
is satisfied.

\noindent\emph{$\star$ Subcase $a.2.\{0\}$: There are safe
  monochromatic cycles of color $0$ only.}

This case is completely symmetric to the case $a.2.\{0\}$.

\noindent\emph{$\star$ Subcase $a.2.\{1\}$: There are safe
  monochromatic cycles of color  $1$ only.}

Let $C'_1$ be a safe $1$-cycle in $G'$.  Let $C'_0,C'_2$ be
monochromatic cycles of color $0,2$ in $G'$.  Suppose that there
exists a path $Q'_0$ of color $0$, from $x$ to $w$, that does not
intersect $C'_1$. Suppose also that there exists a path $Q'_2$ of
color $2$, from $x$ to $w$, that does not intersect $C'_1$.  Let
$C''_{0}=Q'_{0}\cup \{e_{wx}\}$ and $C''_{2}=Q'_{2}\cup
\{e_{wx}\}$. By Lemma~\ref{lem:nocontractiblecycle}, $C''_{0}, C'_1,
C''_{2}$ are not contractible. Both of $C''_{0}, C''_{2}$ does not
intersect $C'_1$, so by Lemma~\ref{lem:intersect2}, they are both
homotopic to $C'_1$. Thus by Lemma~\ref{lem:homotopic3}, cycles
$C''_{0}, C''_{2}$ are homotopic to each other, contradicting
Lemma~\ref{lem:magic} (with $i=1,w,x,Q'_0,Q'_2$).  So we can assume
that one of $Q'_0$ or $Q'_2$ as above does not exist.  

By symmetry, suppose that in $G'$, there does not exist a path of
color $0$, from $x$ to $w$, that does not intersect $C'_1$.  Apply the
coloring a.2.1 on $G$.  Depending on which of the two sectors $C'_2$
enters, $[e_1(w),e_{wy}[$, $[e_{wy},e_{wy}]$, it is replaced by one of
the two following cycles $C_2=C'_2\setminus\{w\}\cup\{v,u\}$,
$C_2=C'_2\setminus\{w\}\cup\{y,u\}$.  In any of the two possibilities
$C_2$ intersects $C'_1$.  Let $C_0$ be a $0$-cycle of $G$.  Cycle
$C_0$ has to contain $u$ or $v$ or both, otherwise it is a safe cycle
of $G'$ of color $0$.  Vertex $v$ has no edge entering it in color
$0$, so $C_0$ does not contain $v$ and so it contains $u$ and $x$, the
extremity of the edge leaving $u$ in color $0$.  Thus $C_0$ intersect
$C_2$.  Let $Q_0$ be the part of $C_0$ consisting of the path from $x$
to $u$. The path $Q'_0=Q_0\setminus\{u\}\cup\{w\}$ of $G'$ is from $x$
to $w$, thus by assumption $Q'_0$ intersects $C'_1$. So $C_0$
intersects $C'_1$ and (T2') is satisfied.

\noindent\emph{$\star$ Subcase $a.2.\{\}$: There are no safe
  monochromatic cycle.}

Let $C'_0,C'_1,C'_2$ be monochromatic cycle of color $0,1,2$ in
$G'$. 

Suppose $C'_1$ is entering $w$ in the sector $[e_2(w),e_{wx}]$.  Apply
the coloring a.2.1 on $G$.  Cycle $C'_1$ is replaced by
$C_1=C'_1\setminus\{w\}\cup\{u,y\}$ or
$C_1=C'_1\setminus\{w\}\cup\{x,v,y\}$.  Cycle $C'_0$ is replaced by
$C_0=C'_0\setminus\{w\}\cup\{u,x\}$ or
$C_0=C'_0\setminus\{w\}\cup\{y,u,x\}$.  Cycle $C'_2$ is replaced by
$C_2=C'_2\setminus\{w\}\cup\{v,u\}$ or
$C_2=C'_2\setminus\{w\}\cup\{y,u\}$.  So $C_0,C_1,C_2$ all intersect
each other and (T2') is satisfied.

The case where $C'_1$ is entering $w$ in the sector $[e_{wx},e_0(w)]$
is completely symmetric and we apply the coloring a.2.2 on $G$.


  \noindent
  \emph{$\bullet$ Cases a.3, a.4, a.5:}

  The proof is simpler for the remaining cases (cases a.3.0, a.4.0,
  a.5.0 on Figure~\ref{fig:case-a}). For each situation, there is only
  one way to extend the coloring to $G$ in order to preserve (T1') and
  this coloring also preserve (T2').  Indeed, in each coloring of $G$, case
  a.3.1, a.4.1, a.5.1 on Figure~\ref{fig:case-a}, one can check that
  every non safe monochromatic cycle $C'$ is replaced by a cycle $C$
  with $C'\setminus\{w\}\cup\{u\}\subseteq C$.  Thus all non safe
  cycles intersects on $u$ and a non safe cycle of color $i$ intersect
  all safe cycles of colors $i-1$ and $i+1$. So (T2') is always
  satisfied. 

  It remains to analyze the situation for the decontraction case (b)
  and (c).  The colorings that are needed are represented on
  Figures~\ref{fig:case-b} and~\ref{fig:case-c}.  The proof is similar
  to case (a) and we omit it.
\hfill $\Box$\vspace{1em}.

Note that we are not able to prove a lemma analogous to Lemma~\ref{lem:contractype1} for Type 2
Schnyder woods.
 On the example of
Figure~\ref{fig:3-connected}, it is not possible to decontract the
graph $G'$ (Figure~\ref{fig:3-connected}.(a)), and extend its Schnyder
wood to $G$ (Figure~\ref{fig:3-connected}.(b)) without modifying the
edges that are not incident to the contracted edge $e$.  Indeed, if we
keep the edges non incident to $e$ unchanged, there are only two
possible ways to extend the coloring in order to preserve (T1), but
none of them fulfills (T2).

\begin{figure}[h!]
\center
\begin{tabular}{ccc}
\includegraphics[scale=0.4]{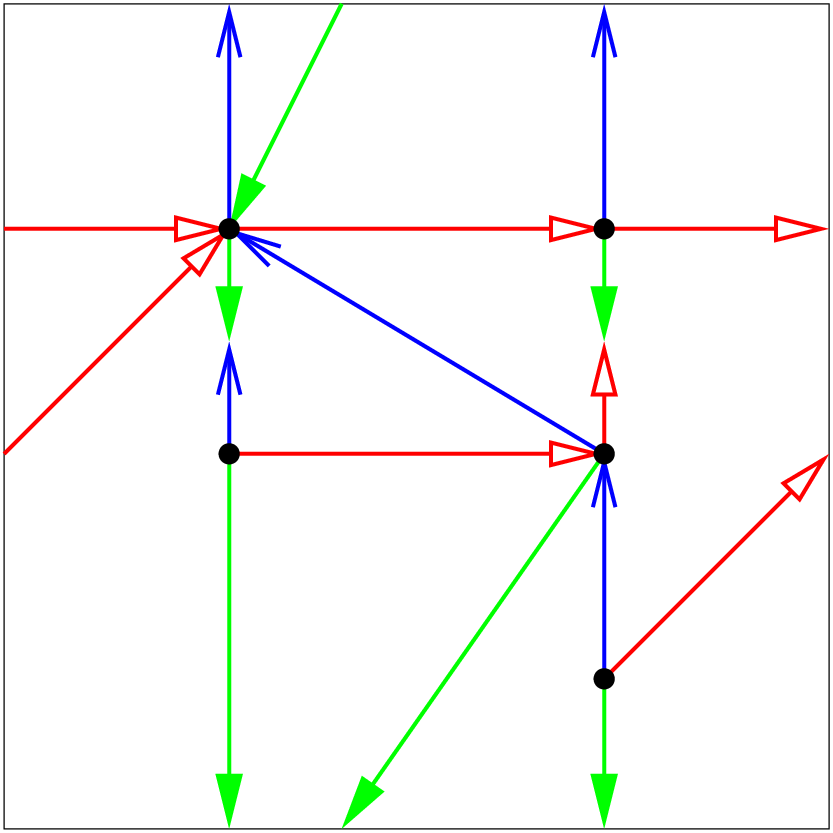}
& \hspace{1em} &
\includegraphics[scale=0.4]{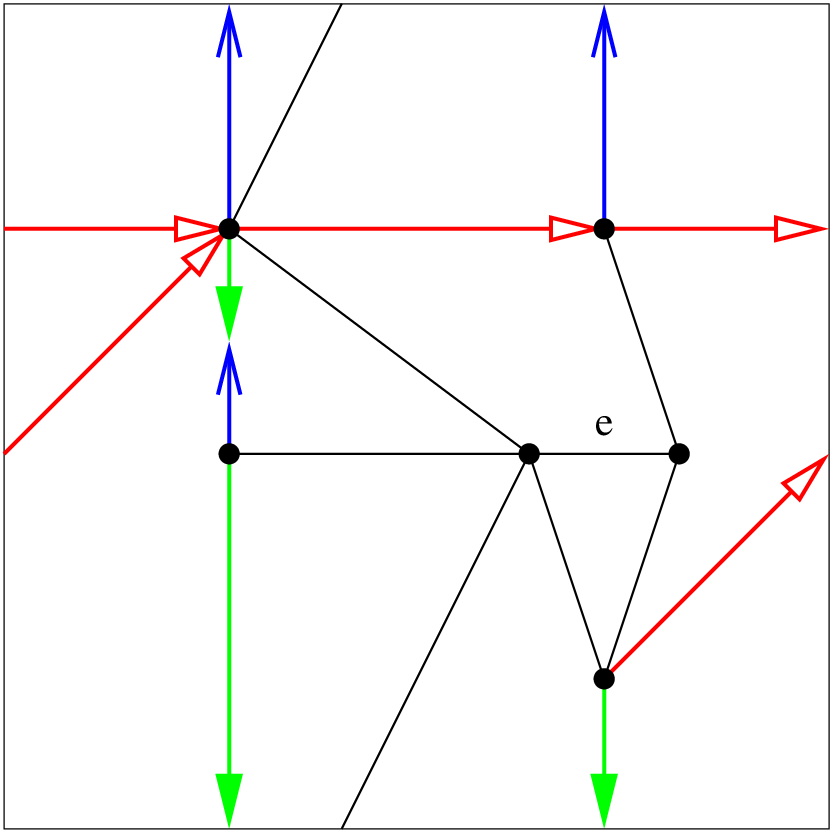}\\
(a) & &(b) \\
\end{tabular}
\caption{ (a) The graph obtained by  contracting the edge
  $e$ of the graph (b). It is not possible to color 
 and orient the black edges of
  (b) to obtain a Schnyder wood.
}
\label{fig:3-connected}
\end{figure}

  \section{Existence for essentially 3-connected toroidal maps}
\label{sec:existence3connected}

Given a map $G$ embedded on a surface.  The \emph{angle map}
\cite{Rose89} of $G$ is a map $A(G)$ on this surface whose vertices
are the vertices of $G$ plus the vertices of $G^*$ (i.e. the faces of
$G$), and whose edges are the angles of $G$, each angle being incident
with the corresponding vertex and face of $G$. Note that if $G$
contains no homotopic multiple edges, then every face of $G$ has
degree at least $3$ in $A(G)$.

Mohar and Rosenstiehl~\cite{MR98} proved that a map $G$ is essentially
2-connected if and only if the angle map $A$ of $G$ has no pair of
(multiple) edges bounding a disk (i.e. no walk of length 2 bounding a
disk). As every face in an angle map is a quadrangle, such disk
contains some vertices of $G$.  The following claim naturally extends
this characterization to essentially 3-connected toroidal maps.

\begin{lemma}
\label{lem:walk}
A toroidal map $G$ is essentially 3-connected if and only if the angle
map $A(G)$ has no walk of length at most four bounding a disk which is
not a face.
\end{lemma}

\begin{proof}
  In $G^\infty$ any minimal separator of size 1 or 2, $S=\{v_1\}$ or
  $S=\{v_1,v_2\}$, corresponds to a separating cycle of length 2 or 4
  in $A(G^\infty)$, $C=(v_1,f_1)$ or $C=(v_1,f_1,v_2,f_2)$, i.e. a
  cycle of length at most 4 bounding a disk $D$ which is not a face.

  ($\Longrightarrow$) Any walk in $A(G)$ of length at most 4 bounding
  a disk which is not a face lifts to a cycle of length at most 4
  bounding a disk which is not a face in $A(G^\infty)$. Thus such a
  walk implies the existence of a small separator in $G^\infty$,
  contradicting its 3-connectedness.

  ($\Longleftarrow$) According to~\cite{MR98}, if $G$ is essentially
  2-connected, $A(G)$ has no walk of length 2 bounding a disk.  Let us
  now show that if $G$ is essentially 2-connected but not essentially
  3-connected, $A(G)$ has a walk of length 4 bounding a disk which is
  not a face.

  If $G$ is essentially 2-connected but not essentially 3-connected,
  then $A(G^\infty)$ has a cycle $C$ of length 4 bounding a disk which
  is not a face, and this cycle corresponds to a contractible walk $W$
  of length 4 in $A(G)$.  Since, $W$ is contractible, it contains a
  subwalk bounding a disk. $A(G)$ being bipartite, this subwalk has
  even length, and since $A(G)$ is essentially 2-connected, it has no
  such walk of length 2. Thus $W$ bounds a disk. Finally, this disk is
  not a single face since otherwise $C$ would bound a single face in
  $A(G^\infty)$.
\end{proof}

A non-loop edge $e$ of an essentially 3-connected toroidal map is
\emph{contractible} if the contraction of $e$ keeps the map
essentially 3-connected. We have the following lemma:

\begin{lemma}
\label{lem:contractibleedges}
  An essentially 3-connected toroidal map that is not reduced to a
  single vertex has a contractible edge.
\end{lemma}

\begin{proof}
  Let $G$ be an essentially 3-connected toroidal map with at least 2
  vertices. Note that for any non-loop $e$, the map $A(G/e)$ has no
  walk of length $2$ bounding a disk which is not a face, otherwise,
  $A(G)$ contains a walk of length at most $4$ bounding a disk which
  is not a face and thus, by Lemma~\ref{lem:walk}, $G$ is not
  essentially 3-connected.

  Suppose by contradiction that contracting any non-loop edge $e$ of
  $G$ yields a non essentially 3-connected map $G/e$. By
  Lemma~\ref{lem:walk}, it means that the angle map $A(G)$ has no walk
  of length at most four bounding a disk which is not a face For any
  non-loop $e$, let $W_4(e)$ be the 4-walk of $A(G/e)$ bounding a
  disk, which is maximal in terms of the faces it contains.  Among all
  the non-loop edges, let $e$ be the one such that the number of faces
  in $W_4(e)$ is minimum.  Let $W_4(e)=(v_1,f_1,v_2,f_2)$ and assume
  that the endpoints of $e$, say $a$ and $b$, are contracted into
  $v_2$ (see Figure~\ref{fig:essentially}.(a)).  Note that by
  maximality of $W_4(e)$, $v_1$ and $v_2$ do not have any common
  neighbor $f$ out of $W_4(e)$, such that $(v_1,f,v_2,f_1)$ bounds a
  disk.

\begin{figure}[h!]
\center
\begin{tabular}{ccc}
\includegraphics[scale=0.4]{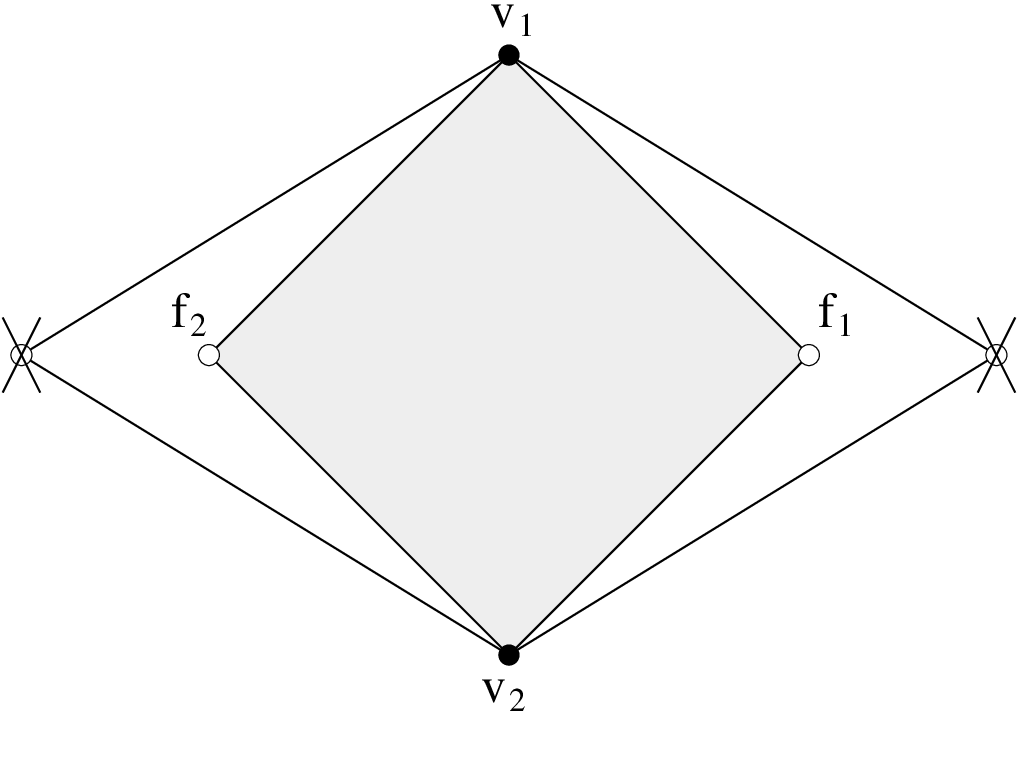}
&  &
\includegraphics[scale=0.4]{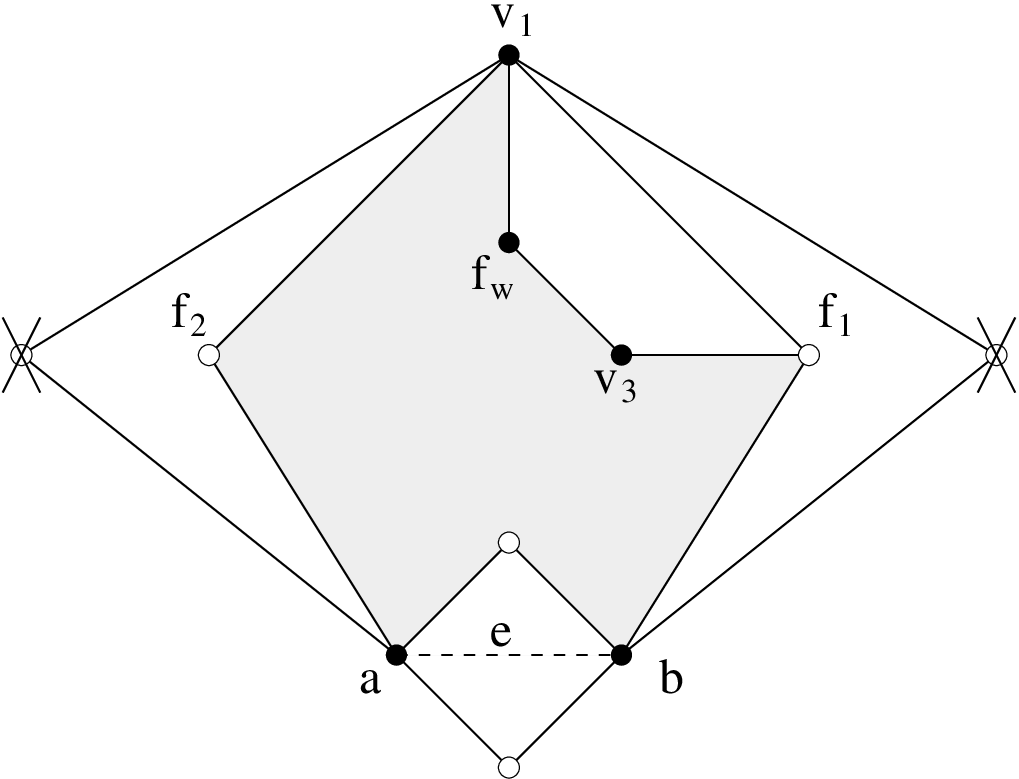}\\
(a) & & (b) \\
\end{tabular}
\caption{Notations of the proof of Lemma~\ref{lem:contractibleedges}}
\label{fig:essentially}
\end{figure}

Assume one of $f_1$ or $f_2$ has a neighbor inside $W_4(e)$. By
symmetry, assume $v_3$ is a vertex inside $W_4(e)$ such that there is
a face $F=(v_1,f_1,v_3,f_w)$ in $A(G/e)$, with eventually $f_w=f_2$.
Consider now the contraction of the edge $v_1v_3$.  Let
$P(v_1,v_3)=(v_1,f_x,v_y,f_z,v_3)$ be the path from $v_1$ to $v_3$
corresponding to $W_4(v_1v_3)$ and $P(a,b)=(a,f_2,v_1,f_1,b)$ the path
corresponding to $W_4(e)$.  Suppose that $f_z= f_2$, then
$(v_1,f_2,v_3,f_w)$ bounds a face by Lemma~\ref{lem:walk} and then
$f_w$ has degree $2$ in $A(G)$, a contradiction. So $f_z\neq f_2$.

  Suppose that all the faces of $W_4(v_1v_3)$ are in $W_4(e)$, then
  with $F$, $W_4(e)$ contains more faces than $W_4(v_1v_3)$, a
  contradiction to the choice of $e$. So in $A(G)$, $P(v_1,v_3)$ must
  cross the path $P(a,b)$. Then $v_y$ or $f_z$ have to intersect
  $P(a,b)$.  Suppose $f_z\neq f_1$. Then $v_y=a$ or $b$. In this case
  $(v_1,f_x,v_2,f_1)$ bounds a disk, a contradiction.  Thus $f_z=f_1$
  and the cycle $(v_1,f_x,v_y,f_1)$ bounds a face by
  Lemma~\ref{lem:walk}.  This implies that $W_4(v_1v_3)$ bounds a
  face, a contradiction.

  Assume now that none of $f_1$ or $f_2$ has a neighbor inside
  $W_4(e)$. Let $f'_1$, $f'_2$, $f_3$ and $f'_3$ be vertices of $A(G)$
  such that $(v_1,f_1,b,f'_1)$, $(v_1,f_2,a,f'_2)$ and
  $(a,f_3,b,f'_3)$ are faces (see
  Figure~\ref{fig:essentially2}). Suppose $f'_1 = f'_2 = f'_3$. Then
  in $A(G/e)$, the face $f'_1$ is deleted (among the 2 homotopic
  multiple edges between $v_1,v_2$ that are created, only one is kept
  in $G/e$). Then $W_4(e)$ bounds a face, a contradiction. Thus there
  exists some $i$ such that $f'_i \neq f'_{i+1}$. Assume that $i=1$
  (resp. $i=2$ or 3), and let $v_3$ and $f''$ be such that there is a
  face $(v_1,f'_1,v_3,f'')$ in $A(G)$ (resp. $(a,f'_2,v_3,f'')$ or
  $(b,f'_3,v_3,f'')$). As above considering the contraction of the
  edge $v_1v_3$ (resp. $av_3$ or $bv_3$) yields a contradiction.
\begin{figure}[h!]
\center
\includegraphics[scale=0.4]{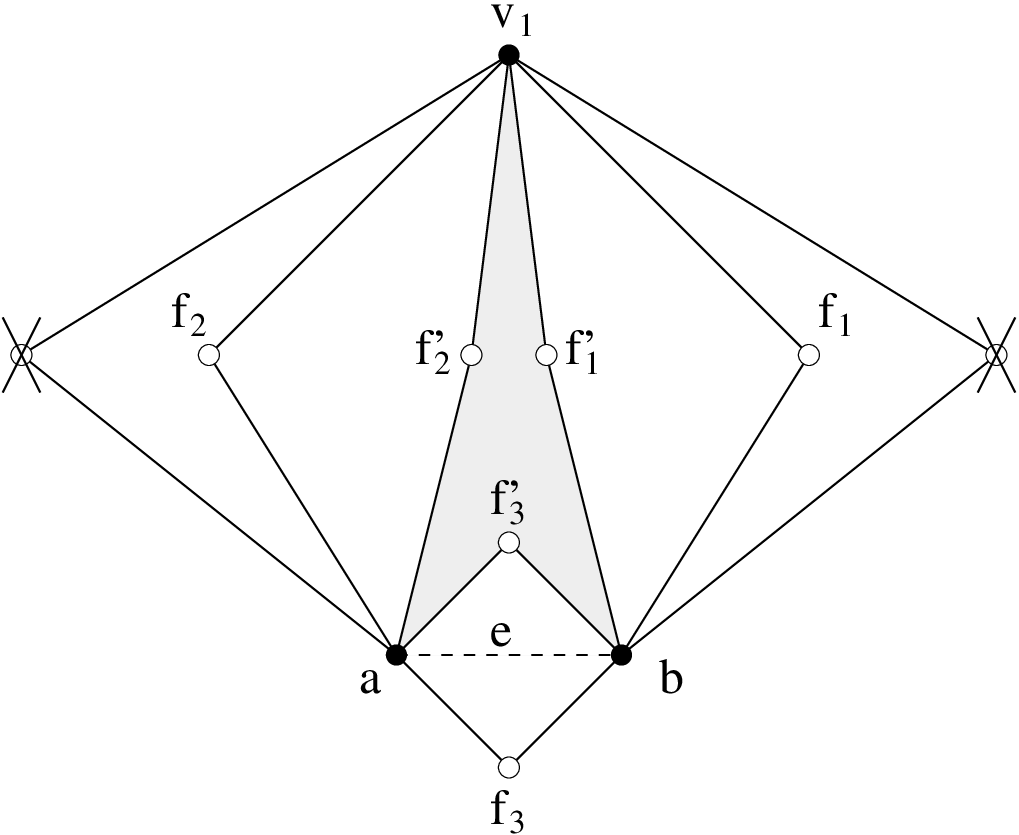}
\caption{Notations of the proof of Lemma~\ref{lem:contractibleedges}}
\label{fig:essentially2}
\end{figure}
\end{proof}

Lemma~\ref{lem:contractibleedges} shows that an essentially
3-connected toroidal map can be contracted step by step by keeping it
essentially 3-connected until obtaining a map with just one vertex.
The two essentially 3-connected toroidal maps on one vertex are
represented on Figure~\ref{fig:essential} with a Schnyder wood. The
graph of Figure~\ref{fig:essential}.(a),  the \emph{3-loops},
admits a Schnyder wood of Type 1, and the graph of
Figure~\ref{fig:essential}.(b),  the \emph{2-loops}, admits a
Schnyder wood of Type 2.

\begin{figure}[!h] 
\center
\begin{tabular}{ccc}
\includegraphics[scale=0.5]{orientation-col.eps}
& \hspace{4em} &
\includegraphics[scale=0.5]{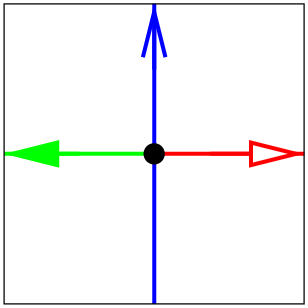}\\
(a) The 3-loops & &(b) The 2-loops 
\end{tabular}
\caption{The two essentially 3-connected toroidal maps on one vertex.}
\label{fig:essential}
\end{figure}

It would be convenient if one could contract any essentially 3-connected
toroidal map until obtaining one of the two graphs of
Figure~\ref{fig:essential} and then decontract the graph to obtain a
Schnyder wood of the original graph. Unfortunately, for Type 2
Schnyder woods we are not able to prove that property (T2) can be
preserved during the decontraction process (see
Section~\ref{sec:lemma}).  Fortunately most essentially 3-connected
toroidal maps admits Schnyder woods of Type 1.  A toroidal map is
\emph{basic} if it consists of a non contractible cycle on $n$
vertices, $n\geq 1$, plus $n$ homotopic loops (see
Figure~\ref{fig:basic}). We prove in this section that non-basic
essentially 3-connected toroidal maps admits Schnyder woods of Type 1.
For this purpose, instead of contracting these maps to one of the two
graphs of Figure~\ref{fig:essential}, we contract them to the graph of
Figure~\ref{fig:essential}.(a) or to the graph of
Figure~\ref{fig:brique},  \emph{the brick}. (One can draw the
universal cover of the brick to understand its name.)

\begin{figure}[!h]
\center
\includegraphics[scale=0.5]{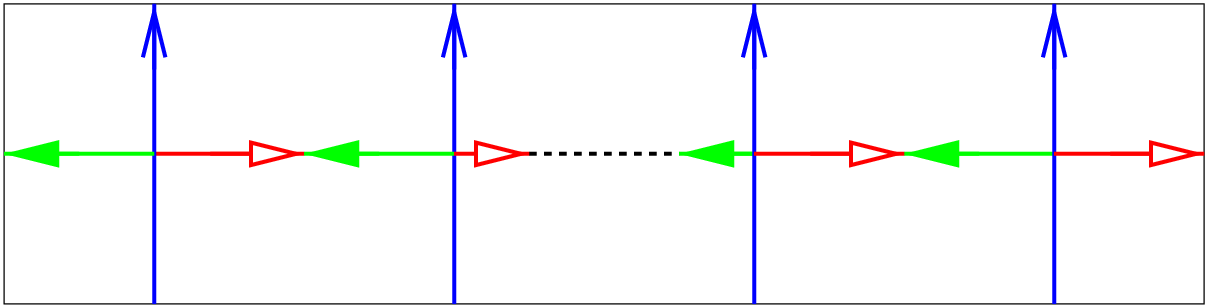}
\caption{The family of basic toroidal maps, having only Schnyder woods
  of Type 2}
\label{fig:basic}
\end{figure}

\begin{figure}[!h]
\center
\includegraphics[scale=0.5]{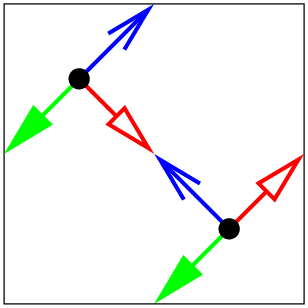}
\caption{The brick, an essentially 3-connected toroidal map with two
  vertices.}
\label{fig:brique}
\end{figure}

\begin{lemma}
\label{lem:contractionbrique}
A non-basic essentially 3-connected toroidal map can be contracted to
the 3-loops (Figure~\ref{fig:essential}.(a)) or to the brick
(Figure~\ref{fig:brique}).
\end{lemma}

\begin{proof}
  Let us prove the lemma by induction on the number of edges of the
  map.  As the 3-loop and the brick are the only non-basic essentially
  3-connected toroidal maps with at most 3 edges, the lemma holds for
  the maps with at most 3 edges.  Consider now a non-basic essentially
  3-connected toroidal map $G$ with at least $4$ edges.  As $G$ has at
  least 2 vertices, it has at least one contractible edge by
  Lemma~\ref{lem:contractibleedges}.  If $G$ has a contractible edge
  $e$ which contraction yields a non-basic map $G'$, then by induction
  hypothesis on $G'$ we are done. Let us prove that such an edge
  always exists.  We assume by contradiction, that the contraction of
  any contractible edge $e$ yields a basic map $G'$. Let us denote
  $v_i$, with $1\le i \le n$, the vertices of $G'$ in such a way that
  $(v_1,v_2,\ldots,v_n)$ is a cycle of $G'$. We can assume that
  $v_1$ is the vertex resulting of  the contraction of $e$. Let
   $u$ and $v$ be the endpoints of $e$ in $G$.

   Suppose first that $u$ or $v$ is incident to a loop in $G$. By
   symmetry, we can assume that $v$ is incident to a loop and that $u$
   is in the cylinder between the loops around $v$ and $v_n$ (note
   that if $n=1$ then $v_n=v$), and note that $u$ is the only vertex
   here. Since $G$ is non-basic and $u$ has at least 3 incident edges,
   two of them go to the same vertex but are non-homotopic. Since
   after the contraction of $e$ there is only one edge left in the
   cylinder, we can deduce that $u$ has at least two edges in common
   with $v$. On the other side since $G$ is essentially 3-connected
   $u$ has an edge $e'$ with $v_n$. This edge $e'$ is contractible
   since its contraction yields a graph containing the basic graph on
   $n$ vertices. But since this graph has 2 non-homotopic edges
   linking $(uv_n)$ and $v$, it is non-basic. So $G$ has a contractible
   edge which contraction produces a non-basic graph, contradicting
   our assumption.

   Suppose now that $u$ and $v$ do not have incident loop, we thus
   have that $G$ contains a cycle $C$ of length $2$ containing
   $e$. Let $e'$ be the other edge of $C$. Since $G$ is essentially
   3-connected, both $u$ and $v$ have at least degree 3, and at least
   one of them has an incident edge on the left (resp. right) of
   $C$. If $n=1$, since $G$ has at least 4 edges there are 2
   (non-homotopic) edges, say $f_1$ and $f_2$ between $u$ and $v$ and
   distinct from $e$ and $e'$. In this case, since the cycles
   $(e,f_1)$ and $(e,f_2)$ were not homotopic, the edges $f_1$ and
   $f_2$ remain non-homotopic in $G'$. So in this case $G$ has one
   vertex and 3 edges, it is thus non-basic. Assume now that $n\ge
   2$. In this case $u$ and $v$ are contained in a cylinder bordered
   by the loops at $v_2$ and at $v_n$ (with eventually $n=2$). In this
   case, we can assume that $u$ has at least one incident edge $f_1$
   on the left of $C$ to $v_n$, and that $v$ has at least one incident
   edge $f_2$ on the right of $C$ to $v_2$. In this case one can
   contract $f_1$ and note that the obtained graph which contain at
   least 3 non-homotopic edges around $v$ ($e$, $e'$ and $f_2$) is
   essentially 3-connected, and non-basic. So $G$ has a contractible
   edge which contraction produces a non-basic graph, contradicting
   our assumption.
\end{proof}

\begin{lemma}
\label{lem:basic}  
A basic toroidal map admits only Schnyder woods of Type 2.
\end{lemma}

\begin{proof}
  Basic toroidal map admit Schnyder woods of Type 2 as shown by
  Figure~\ref{fig:basic}. Suppose that a basic toroidal map $G$ on $n$
  vertices admits a Schnyder wood of Type 1. Consider one of the
  vertical loop $e$ and suppose by symmetry that it is oriented upward
  in color $1$. In a Schnyder wood of type 1, all the monochromatic
  cycles of different colors are not homotopic, thus all the loops
  homotopic to $e$ are also oriented upward in color $1$ and they are
  not bi-oriented. It remains just a cycle on $n$ vertices for edges
  of color $0$ and $2$. Thus the Schnyder wood is the one of
  Figure~\ref{fig:basic}, a contradiction.
\end{proof}

We are now able to prove the following:

\begin{theorem}
\label{th:existencebasic}
A toroidal graph admits a Schnyder wood of Type 1 if and only if it is
an essentially 3-connected non-basic toroidal map.
\end{theorem}

\begin{proof}
  ($\Longrightarrow$) If $G$ is a toroidal graph given with a Schnyder
  wood of Type 1. Then, by Lemma~\ref{lem:essentially}, $G$ is
  essentially 3-connected and by Lemma~\ref{lem:basic}, $G$ is not
  basic.

  ($\Longleftarrow$) Let $G$ be a non-basic essentially 3-connected
  toroidal map.  By Lemma~\ref{lem:contractionbrique}, $G$ can be
  contracted to the 3-loops or to the brick.  Both of these graphs
  admit Schnyder woods of type 1 (see Figure~\ref{fig:essential}.(a)
  and~\ref{fig:brique}). So by Lemma~\ref{lem:contractype1} applied
  successively, $G$ admits a Schnyder wood of Type 1.
\end{proof}

Theorem~\ref{th:existencebasic} and Lemma~\ref{lem:basic} imply
Theorem~\ref{th:existence}.  One related open problem is to
characterize which essentially 3-connected toroidal maps have Schnyder
woods of Type 2.

Here is a remark about how to compute a Schnyder wood for an
essentially 3-connected toroidal triangulation. Instead of looking
carefully at the technical proof of Lemma~\ref{lem:contractype1} to
know which coloring of the decontracted graph has to be chosen among
the possible choice. One can try the possible cases $\alpha.k.\ell$,
$\ell\geq 1$, and check after which obtained coloring is a Schnyder
wood. To do so, one just as to check if (T2') is satisfied. Checking
that (T2') is satisfied can be done by the following method: start
from any vertex $v$, walk along $P_0(v),P_1(v),P_2(v)$ and mark the
three monochromatic cycles $C_0,C_1,C_2$ reached by the three paths
$P_i$.  Property (T2') is then satisfied if the cycles $C_0,C_1,C_2$
pairwise intersect.

The existence of Schnyder wood for toroidal triangulations implies the
following theorem.

\begin{theorem}
\label{cor:cyclesedge-disjoint}
A toroidal triangulation contains three non contractible and non
homotopic cycles that are pairwise edge-disjoint.
\end{theorem}

\begin{proof}
  One just has to apply Theorem~\ref{th:existencebasic} to obtain a
  Schnyder wood of Type 1 and then, for each color $i$, choose
  arbitrarily a $i$-cycle. These cycles are edge-disjoint as, by
  Euler's formula, there is no bi-oriented edges in Schnyder woods of
  toroidal triangulations.
\end{proof}

The conclusion of Theorem~\ref{cor:cyclesedge-disjoint} is weaker than
the one of Theorem~\ref{th:fij} but it is not restricted to simple
toroidal triangulations. Recall that Theorem~\ref{th:fij} is not
true for general toroidal triangulations as shown by the graph of
Figure~\ref{fig:not-connected}.

A nonempty family $\mathcal R$ of linear orders on the vertex set $V$
of a simple graph $G$ is called a \emph{realizer} of $G$ if for every
edge $e$, and every vertex $x$ not in $e$, there is some order $<_i\in
\mathcal R$ so that $y<_ix$ for every $y\in e$. The
\emph{Dushnik-Miller dimension} \cite{DM41} of $G$, is defined as the
least positive integer $t$ for which $G$ has a realizer of cardinality
$t$. Realizers are usually used on finite graphs, but  here we allow
$G$ to be an infinite simple graph.

Schnyder woods where originally defined by Schnyder~\cite{Sch89} to
prove that a finite planar graph $G$ has Dushnik-Miller dimension at
most $3$.  A consequence of Theorem~\ref{th:existence} is an
analogous result for the universal cover of a toroidal graph:

\begin{theorem}
  The universal cover of a toroidal graph has Dushnik-Miller dimension
  at most three.
\end{theorem}

\begin{proof}
  By eventually adding edges to $G$ we may assume that $G$ is a
  toroidal triangulation.  By Theorem~\ref{th:existence}, it admits a
  Schnyder wood.  For $i\in\{0,1,2\}$, let $<_i$ be the order induced
  by the inclusion of the regions $R_i$ in $G^\infty$. That is
  $u<_{i}v$ if and only if $R_i(u)\subsetneq R_i(v)$. Let $<'_i$ be
  any linear extension of $<_i$ and consider $\mathcal
  R=\{<'_0,<'_1,<'_2\}$.  Let $e$ be any edge of $G^{\infty}$ and $v$
  be any vertex of $G^{\infty}$ not in $e$. Edge $e$ is in a region
  $R_i(v)$ for some $i$, thus $R_i(u)\subseteq R_i(v)$ for every $u\in
  e$ by Lemma~\ref{lem:regionss}.(i).  As there is no edges oriented
  in two directions in a Schnyder wood of a toroidal triangulation, we
  have $R_i(u)\neq R_i(v)$ and so $u <_{i}v$.  Thus $\mathcal R$ is a
  realizer of $G^\infty$.
\end{proof}

\section{Orthogonal surfaces}
\label{sec:ortho}

Given two points $u=(u_0,u_1,u_2)$ and $v=(v_0,v_1,v_2)$ in $\mathbb
R^3$, we note $u\vee v=(\max(u_i,v_i))_{i=0,1,2}$ and $u\wedge
v=(\min(u_i,v_i))_{i=0,1,2}$. We define an order $\geq$ among the
points in $\mathbb R^3$, in such a way that $u\geq v$ if $u_i\geq v_i$
for $i=0,1,2$.

Given a set $\mv$ of pairwise incomparable elements in $\mr^3$, we
define the set of vertices that dominates $\mv$ as
$\md_\mv=\{u\in\mathbb R^3\ |\ \exists\, v\in \mathcal V$ such that
$u\geq v\}$. The \emph{orthogonal surface $\ms$} generated by $\mv$ is
the boundary of $\md_{\mv}$. (Note that orthogonal surfaces are well
defined even when $\mv$ is an infinite set.)  If $u,v\in \mv$ and
$u\vee v\in \ms$, then $\ms$ contains the union of the two line
segment joining $u$ and $v$ to $u\vee v$. Such arcs are called
\emph{elbow geodesic}. The \emph{orthogonal arc} of $v\in \mv$ in the
direction of the standard basis vector $e_i$ is the intersection of
the ray $v+\lambda e_i$ with $\ms$.

Let $G$ be a planar map. A \emph{geodesic embedding} of $G$ on the
orthogonal surface $\ms$ is a drawing of $G$ on $\ms$ satisfying the
following:

\begin{itemize}
\item[(D1)] There is a bijection between the vertices of $G$ and $\mv$.

\item[(D2)] Every edge of $G$ is an elbow geodesic.

\item[(D3)] Every 
orthogonal arc in $\ms$ is part of an edge of $G$.

\item[(D4)] There are no crossing edges in the embedding of $G$ on $\ms$.
\end{itemize}

Miller~\cite{Mil02} (see also \cite{Fel03, FZ08}) proved that a
geodesic embedding of a planar map $G$ on an orthogonal surface $\ms$
induces a Schnyder wood of $G$. The edges of $G$ are colored with the
direction of the orthogonal arc contained in the edge. An orthogonal
arc intersecting the ray $v+\lambda e_i$ corresponds to the edge
leaving $v$ in color $i$. Edges represented by two orthogonal arcs
corresponds to edges oriented in two directions.

Conversely, it has been proved that a Schnyder wood of a planar map
$G$ can be used to obtain a geodesic embedding of $G$.  Let $G$ be a
planar map given with a Schnyder wood.  The method is the following
(see \cite{Fel03} for more details): For every vertex $v$, one can
divide $G$ into the three regions bounded by the three monochromatic
path going out from $v$. The \emph{region vector} associated to $v$ is
the vector obtained by counting the number of faces in each of these
three regions.  The mapping of each vertex on its region vector gives
the geodesic embedding. (Note that in this approach, the vertices are
all mapped on the same plane as the sum of the coordinates of each
region vector is equal to the total number of inner faces of the map.)

Our goal is to generalize geodesic embedding to the torus.  More
precisely, we want to represent the universal cover of a toroidal map
on an infinite and periodic orthogonal surface.  

Let $G$ be a toroidal map.  Consider any flat torus representation of
$G$ in a parallelogram $P$.  The graph $G^\infty$ is obtained by
replicating $P$ to tile the plane.  Given any of these parallelograms
$Q$, let $Q^{top}$ (resp. $Q^{right}$) be the copy of $P$ just above
(resp. on the right of) $Q$.  Given a vertex $v$ in $Q$, we note
$v^{top}$ (resp. $v^{right}$) its copies in $Q^{top}$
(resp. $Q^{right}$).

A mapping of the vertices of $G^\infty$ in $\mr^d$, $d\in\{2,3\}$, is
\emph{periodic} with respect to vectors $S$ and $S'$ of $\mr^d$, if
there exists a flat torus representation $P$ of $G$ such that for any
vertex $v$ of $G^\infty$, vertex $v^{top}$ is mapped on $v+S$ and
$v^{right}$ is mapped on $v+S'$.  A \emph{geodesic embedding} of a
toroidal map $G$ is a geodesic embedding of $G^\infty$ on $\mathcal
S_{\mathcal V^\infty}$, where $\mathcal V^\infty$ is a periodic
mapping of $G^\infty$ with respect to two non collinear vectors (see
example of Figure~\ref{fig:example-primal}).

\begin{figure}[h!]
\center
\includegraphics[scale=0.2]{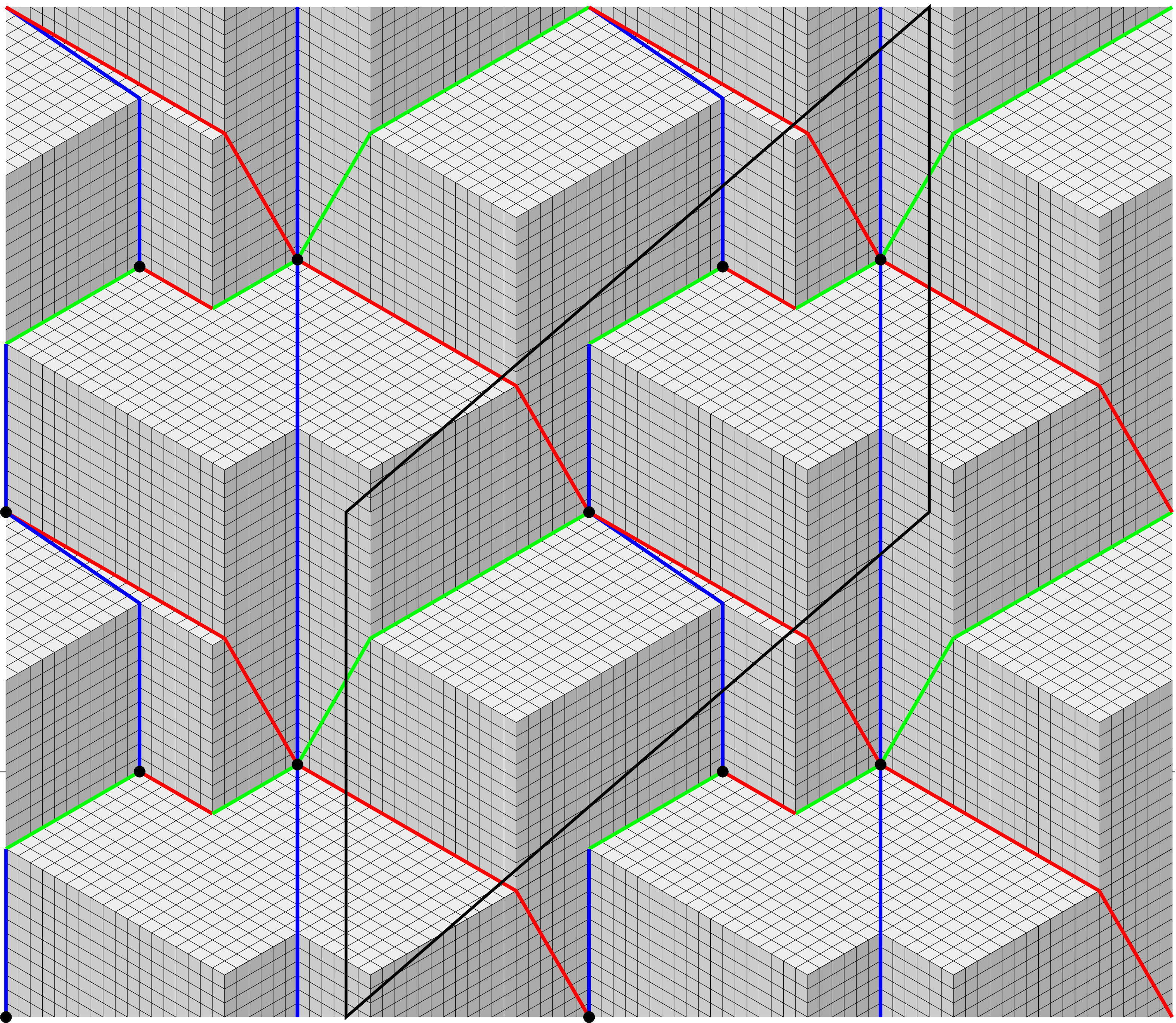}
\caption{Geodesic embedding of the toroidal map of Figure~\ref{fig:example-schnyder}.}
\label{fig:example-primal}
\end{figure}

Like in the plane Schnyder woods can be used to obtain geodesic
embeddings of toroidal maps. For that purpose, we need to generalize
the region vector method.  The idea is to use the regions $R_i(v)$ to
compute the coordinates of the vertex $v$ of $G^\infty$. The problem
is that contrarily to the planar case, these regions are unbounded and
contains an infinite number of faces. The method is thus generalized
by the following.

Let $G$ be a toroidal map, given with a Schnyder wood and a flat torus
representation in a parallelogram $P$.

Recall that $\mathcal C_i=\{C_i^0,\ldots,C_i^{k_i-1}\}$ denotes the
set of $i$-cycles of $G$ such that there is no $i$-cycle in the region
$R(C_i^j,C_i^{j+1})$. Recall that $\mathcal L_i^j$ denotes the set of
$i$-lines of $G^\infty$ corresponding to $C_i^j$. The \emph{positive
  side} of a $i$-line is define as the right side while ``walking''
along the directed path by following the orientation of the edges
colored $i$.

\begin{lemma}
\label{lem:linesintersection}
For any vertex $v$, the two monochromatic lines $L_{i-1}(v)$ and
$L_{i+1}(v)$ intersect. Moreover, if the Schnyder wood is of Type 2.i,
then $L_{i+1}(v)=(L_{i-1}(v))^{-1}$ and $v$ is situated on the right
of $L_{i+1}(v)$.
\end{lemma}

\begin{proof}
  Let $j,j'$ be such that $L_{i-1}(v)\in \mathcal L_{i-1}^{j}$ and
  $L_{i+1}(v)\in \mathcal L_{i+1}^{j'}$.  If the Schnyder wood is of
  Type 1 or Type 2.j with $j\neq i$, then the two cycles $C_{i-1}^{j}$
  and $C_{i+1}^{j'}$ are not homotopic, and so the two lines
  $L_{i-1}(v)$ and $L_{i+1}(v)$ intersect.  

  If the Schnyder wood is of Type 2.i, we consider the case where
  $v\in L_{i-1}(v)$, and the case where $v$ does not belong to
  $L_{i-1}(v)$ nor $L_{i+1}(v)$.  Then $v$ lies between two
  consecutive $(i+1)$-lines (which are also $(i-1)$-lines). Let us
  denote those two lines $L_{i+1}$ and $L'_{i+1}$, such that
  $L'_{i+1}$ is situated on the right of $L_{i+1}$ and $v\notin
  L'_{i+1}$. By property (T1), $P_{i+1}(v)$ and $P_{i-1}(v)$ cannot
  reach $L'_{i+1}$. Thus $L_{i+1}=L_{i+1}(v)=(L_{i-1}(v))^{-1}$.
\end{proof}

The \emph{size} of the region $R(C_i^j,C_i^{j+1})$ of $G$, denoted
$f_i^j=|R(C_i^j,C_i^{j+1})|$, is equal to the number of faces in
$R(C_i^j,C_i^{j+1})$.  Remark that for each color, we have
$\sum_{j=0}^{k_i-1} f_i^j$ equals the total number of faces $f$ of
$G$.  If $L$ and $L'$ are consecutive $i$-lines of $G^\infty$ with
$L\in \mathcal L_i^j$ and $L'\in \mathcal L_i^{j+1}$, then the
\emph{size} of the (unbounded) region $R(L,L')$, denoted $|R(L,L')|$,
is equal to $f_i^j$. If $L$ and $L'$ are any $i$-lines, the
\emph{size} of the (unbounded) region $R(L,L')$, denoted $|R(L,L')|$,
is equal to the sum of the size of all the regions delimited by
consecutive $i$-lines inside $R(L,L')$.  For each color $i$, choose
arbitrarily a $i$-line $L_i^*$ in $\mathcal L_i^0$ that is used as an
origin for $i$-lines.  Given a $i$-line $L$, we define the value
$f_i(L)$ of $L$ as follows: $f_i(L)=|R(L,L_i^*)|$ if $L$ is on the
positive side of $L_i^*$ and $f_i(L)=-|R(L,L_i^*)|$ otherwise.

Consider two vertices $u,v$ such that $L_{i-1}(u)=L_{i-1}(v)$ and
$L_{i+1}(u)=L_{i+1}(v)$. Even if the two regions $R_{i}(u)$ and
$R_{i}(v)$ are unbounded, their \emph{difference} is bounded.  Let
$d_i(u,v)$ be the number of faces in $R_{i}(u)\setminus R_{i}(v)$
minus the number of faces in $R_{i}(v)\setminus R_{i}(u)$.  For any
vertex, by Lemma~\ref{lem:linesintersection}, there exists $z_i(v)$ a
vertex on the intersection of the two lines $L_{i-1}(v)$ and
$L_{i+1}(v)$.  Let $N$ be a constant $\geq n$ (in this section we can
have $N=n$ but in Section~\ref{sec:straight} we need to choose $N$
bigger).  We are now able to define the region vector of a vertex 
of $G^\infty$, that is a mapping of this vertex in $\mathbb R^3$.

\begin{definition}[Region vector]
  The $i$-th coordinate of the \emph{region vector} of a vertex $v$ of
  $G^\infty$ is equal to $v_i\ =\ d_i(v,z_i(v))\ +\
  N\times\big(f_{i+1}(L_{i+1}(v))-f_{i-1}(L_{i-1}(v))\big)$ (see
  Figure~\ref{fig:region}).
\end{definition}

\begin{figure}[!h]
\center
\input{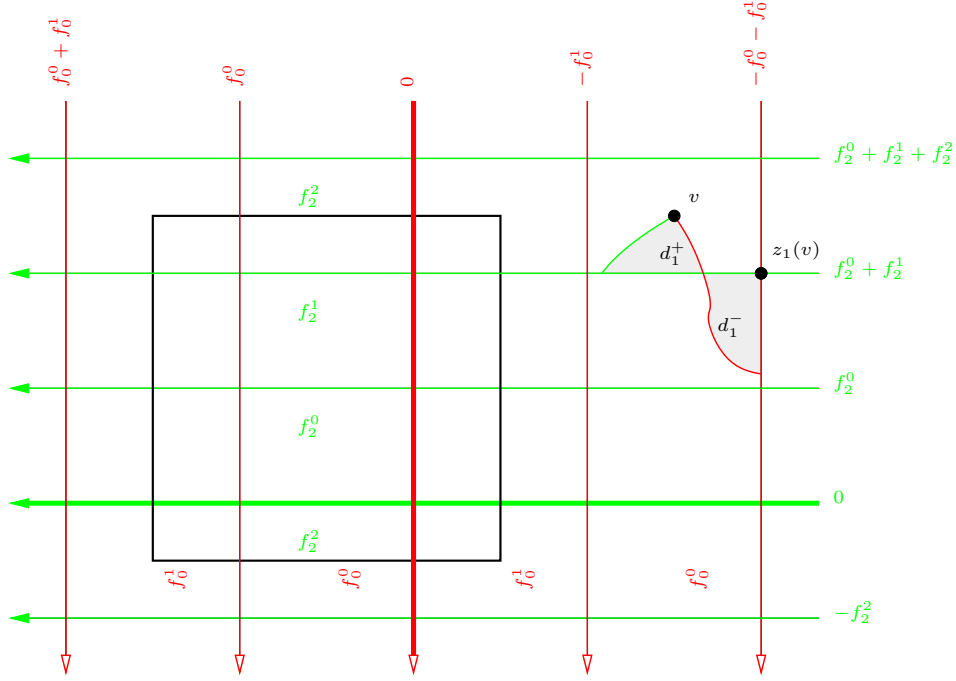}
\caption{Coordinate $1$ of vertex $v$, is equal to the number of faces
  in the region $d_1^+$, minus the number of faces in the region
  $d_1^-$, plus $N$ times $(f_2^0+f_2^1)-(-f_0^0-f_0^1)$.}
\label{fig:region}
\end{figure}

\begin{lemma}
\label{lem:sum}
The sum of the coordinates of a vertex $v$ equals the number of faces
in the bounded region delimited by the lines $L_0(v)$, $L_1(v)$ and
$L_2(v)$ if the Schnyder wood is of Type 1 and this sum equals zero if
the Schnyder wood is of Type 2.
\end{lemma}

\begin{proof}
  We have $v_0+v_1+v_2 = d_0(v,z_0(v))+ d_1(v,z_1(v))+ d_2(v,z_2(v))=
  \sum_i(|R_i(v) \setminus R_i(z_i(v))|-|R_i(z_i(v)) \setminus
  R_i(v)|) $.  We use the characteristic function $\bf 1$ to deal with
  infinite regions. We note ${\bf 1}(R)$, the function defined on the
  faces of $G^\infty$ that has value $1$ on each face of region $R$
  and $0$ elsewhere. Given a function $g:F(G^\infty)\longrightarrow
  \mathbb Z$, we note $|g|=\sum_{F\in F(G^\infty)} g(F)$ (when the sum
  is finite).  Thus 
$\sum_i v_i
=\sum_i(|\mathbf 1(R_i(v) \setminus
  R_i(z_i(v)))|-|\mathbf 1(R_i(z_i(v)) \setminus R_i(v))|)
=
  |\sum_i(\mathbf 1(R_i(v) \setminus R_i(z_i(v)))-\mathbf
  1(R_i(z_i(v)) \setminus R_i(v)))|$. Now we compute $g=\sum_i(\mathbf
  1(R_i(v) \setminus R_i(z_i(v)))-\mathbf 1(R_i(z_i(v)) \setminus
  R_i(v)))$.  We have:
$$ g =\sum_i(\mathbf 1(R_i(v) \setminus
R_i(z_i(v)))+ \mathbf 1(R_i(v) \cap R_i(z_i(v))) -\mathbf
1(R_i(z_i(v)) \setminus R_i(v)) -\mathbf 1(R_i(v) \cap R_i(z_i(v))))
$$
As $R_i(v) \setminus
R_i(z_i(v))$ and $R_i(z_i(v)) \setminus R_i(v)$
are disjoint from $R_i(v) \cap R_i(z_i(v))$, we have
$$
g= \sum_i(\mathbf 1(R_i(v)) -\mathbf
1(R_i(z_i(v))))=\sum_i\mathbf 1(R_i(v)) -
\sum_i\mathbf 1(R_i(z_i(v)))
$$
The interior of the three regions $R_i(v)$, for $i=0,1,2$, being
disjoint and spanning the whole plane $\mathbb P$ (by definition), we
have $\sum_i\mathbf 1(R_i(v))= \mathbf 1(\cup_i(R_i(v)))=\mathbf
1(\mathbb P)$.  Moreover the regions $R_i(z_i(v))$, for $i=0,1,2$, are also disjoint
and $\sum_i\mathbf 1(R_i(z_i(v)))= \mathbf
1(\cup_i(R_i(z_i(v))))=\mathbf 1(\mathbb P\setminus T)$ where $T$ is
the bounded region delimited by the lines $L_0(v)$, $L_1(v)$ and
$L_2(v)$.  So $ g= \mathbf 1(\mathbb P) - \mathbf 1(\mathbb P\setminus
T)= \mathbf 1( T) $.  And thus $\sum_i v_i=|g|=| \mathbf 1( T)|$.
\end{proof}

Lemma~\ref{lem:sum} shows that if the Schnyder wood is of Type 1,
then the set of points are not necessarily coplanar like in the planar
case~\cite{Fel-book}, but all the copies of a vertex lies on the
same plane (the bounded region delimited by the lines $L_0(v)$,
$L_1(v)$ and $L_2(v)$ has the same number of faces for any copies of a
vertex $v$). Surprisingly, for Schnyder woods of Type 2, all the
points are coplanar.

For each color $i$, let $c_i$ (resp. $c'_i$), be the algebraic number
of times a $i$-cycle is traversing the vertical (resp. horizontal)
side of the parallelogram $P$ (that was the parallelogram containing
the flat torus representation of $G$) from right to left (resp. from top to
bottom).  This number increases by one each time a monochromatic cycle
traverses the side in the given direction and decreases by one when it
traverses in the other way.  Let $S$ and $S'$ be the two vectors of
$\mr^3$ with coordinates $S_i=N(c_{i+1}-c_{i-1})f$ and
$S'_i=N(c'_{i+1}-c'_{i-1})f$.  Note that $S_0+S_1+S_2=0$ and
$S'_0+S'_1+S'_2=0$

\begin{lemma}
\label{lem:periodic}
The mapping is periodic with respect to $S$ ans $S'$.
\end{lemma}

\begin{proof}
  Let $v$ be any vertex of $G^\infty$.  Then $v^{top}_i-v_i
  =N(f_{i+1}(L_{i+1}(v^{top}))-f_{i+1}(L_{i+1}(v)))-N(f_{i-1}(L_{i-1}(v^{top}))-f_{i-1}(L_{i-1}(v)))=N(c_{i+1}-c_{i-1})f$. So
  $v^{top}=v+S$. Similarly $v^{right}=v+S'$.
\end{proof}

For each color $i$, let $\gamma_i$ be the integer such that two
monochromatic cycles of $G$ of respective colors $i-1$ and $i+1$
intersect exactly $\gamma_i$ times, with the convention that
$\gamma_i=0$ if the Schnyder wood is of Type 2.i. By
Lemma~\ref{lem:allhomotopic}, $\gamma_i$ is properly defined and do
not depend on the choice of the monochromatic cycles.  Note that if
the Schnyder wood is of Type 2.i, then $\gamma_{i-1}=\gamma_{i+1}$ and
if the Schnyder wood is not of Type 2.i, then $\gamma_i\neq 0$. Let
$\gamma = \max(\gamma_0,\gamma_1,\gamma_2)$. Let $Z_0=(
(\gamma_1+\gamma_2) Nf, - \gamma_1 Nf, - \gamma_2 Nf)$ and $Z_1=(-
\gamma_0Nf, (\gamma_0+\gamma_2) Nf, - \gamma_2 Nf)$ and $Z_2=(-
\gamma_0Nf,- \gamma_1 Nf,(\gamma_0+\gamma_1)Nf)$.

\begin{lemma}
\label{lem:copies}
For any vertex $u$, we have $\{u+k_0Z_0+k_1Z_1+k_2Z_2\, |\,
k_0,k_1,k_2\in \mz\} \subseteq \{u+kS+k'S'\, |\, k,k'\in \mz\}$.
\end{lemma}

\begin{proof}
  Let $u,v$ be two copies of the same vertex, such that $v$ is the
  first copy of $u$ in the direction of $L_0(u)$.  (That is
  $L_0(u)=L_0(v)$ and on the path $P_0(u)\setminus P_0(v)$ there is no
  two copies of the same vertex.)  Then $v_i-u_i
  =N(f_{i+1}(L_{i+1}(v))-f_{i+1}(L_{i+1}(u)))-N(f_{i-1}(L_{i-1}(v))-f_{i-1}(L_{i-1}(u)))$.
  We have $|R(L_{0}(v),L_{0}(u))|=0$, $|R(L_{1}(v),L_{1}(u))|=\gamma_2
  f$ and $|R(L_{2}(v),L_{2}(u))|=\gamma_1 f$. So $v_0-u_0= N
  (\gamma_1+\gamma_2) f$ and $v_1-u_1= -N \gamma_1 f$ and $v_2-u_2= -
  N \gamma_2 f$. So $v=u+Z_0$.  Similarly for the other colors. So the
  first copy of $u$ in the direction of $L_i(u)$ is equal to $u+Z_i$.
  By Lemma~\ref{lem:periodic}, all the copies of $u$ are mapped on
  $\{u+kS+k'S'\, |\, k,k'\in \mz\}$, so we have the result.
\end{proof}

\begin{lemma}
\label{lem:dim}
We have $dim(Z_0,Z_1,Z_2)=2$ and if the Schnyder wood is not of Type
2.i, then $dim(Z_{i-1},Z_{i+1})=2$.
\end{lemma}

\begin{proof}
  We have $\gamma_0Z_0+\gamma_1Z_1+\gamma_2Z_2=0$ and so
  $dim(Z_0,Z_1,Z_2)\leq 2$.  We can assume by symmetry that the
  Schnyder wood is not of Type 2.1 and so $\gamma_1\neq 0$. Thus
  $Z_0\neq 0$ and $Z_2\neq 0$.  Suppose by contradiction that
  $dim(Z_0,Z_2)=1$. Then there exist $\alpha\neq 0$, $\beta\neq 0$,
  such that $\alpha Z_0 + \beta Z_2=0$.  The sum of this equation for
  the coordinates $0$ and $2$ gives $(\alpha+\beta)\gamma_1=0$ and
  thus $\alpha = -\beta$. Then the equation for coordinate $0$ gives
  $\gamma_0+\gamma_1+\gamma_2=0$ contradicting the fact that
  $\gamma_1>1$ and $\gamma_0,\gamma_2\geq 0$.
\end{proof}

\begin{lemma}
\label{lem:notcolinear}
The vectors $S,S'$ are not collinear.
\end{lemma}

\begin{proof}
  By Lemma~\ref{lem:copies}, the set $\{u+k_0Z_0+k_1Z_1+k_2Z_2\, |\,
  k_0,k_1,k_2\in \mz\}$ is a subset of $\{u+kS+k'S'\, |\, k,k'\in
  \mz\}$.  By Lemma~\ref{lem:dim}, we have $dim(Z_0,Z_1,Z_2)=2$, thus
  $dim(S,S')=2$.
\end{proof}

\begin{lemma}
  \label{lem:regionslongue}
  If $u,v$ are two distinct vertices such that $v$ is in $L_{i-1}(v)$,
  $u$ is in $P_{i-1}(v)$, both $u$ and $v$ are in the region
  $R(L_{i+1}(u),L_{i+1}(v))$ and $L_{i+1}(u)$ and $L_{i+1}(v)$ are two
  consecutive $(i+1)$-lines with $L_{i+1}(u)\in \mathcal L_{i+1}^{j}$
  (see Figure~\ref{fig:longregion}).  Then
  $d_i(z_i(v),v)+d_i(u,z_i(u))< (n-1)\times f_{i+1}^{j}$.
\end{lemma}

\begin{figure}[!h]
\center
\input{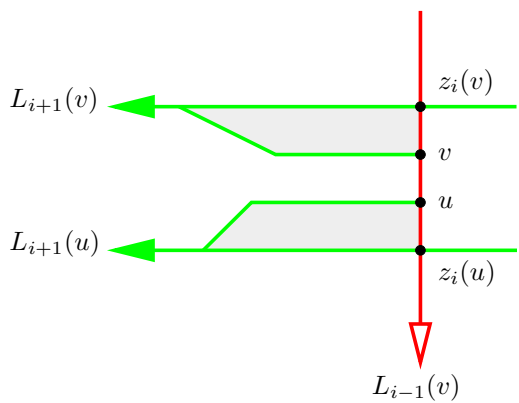}
\caption{The gray area, corresponding to the quantity
  $d_i(z_i(v),v)+d_i(u,z_i(u))$, has size bounded by $(n-1)\times
  f_{i+1}^{j}$.
}
\label{fig:longregion}
\end{figure}

\begin{proof}
  Let $Q_{i+1}(v)$ the subpath of $P_{i+1}(v)$ between $v$ and
  $L_{i+1}(v)$ (maybe $Q_{i+1}(v)$ has length $0$ if $v=z_i(v)$).  Let
  $Q_{i+1}(u)$ the subpath of $P_{i+1}(u)$ between $u$ and $L_{i+1}(u)$
  (maybe $Q_{i+1}(u)$ has length $0$ if $u=z_i(u)$).  The path
  $Q_{i+1}(v)$ cannot contain two different copies of a vertex of $G$,
  otherwise $Q_{i+1}(v)$ will corresponds to a non contractible cycle
  of $G$ and thus will contain an edge of $L_{i+1}(v)$. So the length
  of $Q_{i+1}(v)$ is $\leq n-1$.
  
  The total number of times a copy of a given face of $G$ can appear
  in the region $R=R_i(z_i(v))\setminus R_i(v)$, corresponding to
  $d_i(z_i(v),v)$, can be bounded as follow. Region $R$ is between two
  consecutive copies of $L_{i+1}(u)$. So in $R$, all the copies of a
  given face are separated by a copie of $L_{i-1}(v)$. Each copy of
  $L_{i-1}(v)$ intersecting $R$ have to intersect $Q_{i+1}(v)$ on a
  specific vertex. As $Q_{i+1}(v)$ has at most $n$ vertices. A given
  face can appear at most $n-1$ times in $R$.  Similarly, the total
  number of times that a copy of a given face of $G$ can appear in the
  region $R_i(u)\setminus R_i(z_i(u))$, corresponding to
  $d_i(u,z_i(u))$, is $\leq (n-1)$.

  A given face of $G$ can appear in only one of the two gray regions
  of Figure~\ref{fig:longregion}. So a face is counted $\leq n-1$
  times in the quantity $d_i(z_i(v),v)+d_i(u,z_i(u))$.
  Only the faces of the region $R(C_{i+1}^{j},C_{i+1}^{j+1})$ can be
  counted. And there is at least one face of
  $R(C_{i+1}^{j},C_{i+1}^{j+1})$ (for example one incident to $v$)
  that is not counted. So in total $d_i(z_i(v),v)+d_i(u,z_i(u))\leq
  (n-1)\times (f_{i+1}^{j}-1)< (n-1)\times f_{i+1}^{j}$.
\end{proof}

Clearly, the symmetric of Lemma~\ref{lem:regionslongue},  where
the role of $i+1$ and $i-1$ are exchanged, is also true.

The bound of Lemma~\ref{lem:regionslongue} is somehow sharp. In the
example of Figure~\ref{fig:examplelongue}, the rectangle represent a
toroidal map $G$ and the universal cover is partially represented.  If
the map $G$ has $n$ vertices and $f$ faces ($n=5$ and $f=5$ in the
example), then the gray region, representing the quantity
$d_1(z_1(v),v)+d_1(u,z_1(u))$, has size
$\frac{n(n-1)}{2}=\Omega(n\times f)$.

 \begin{figure}[!h]
 \center
 \includegraphics[scale=0.3]{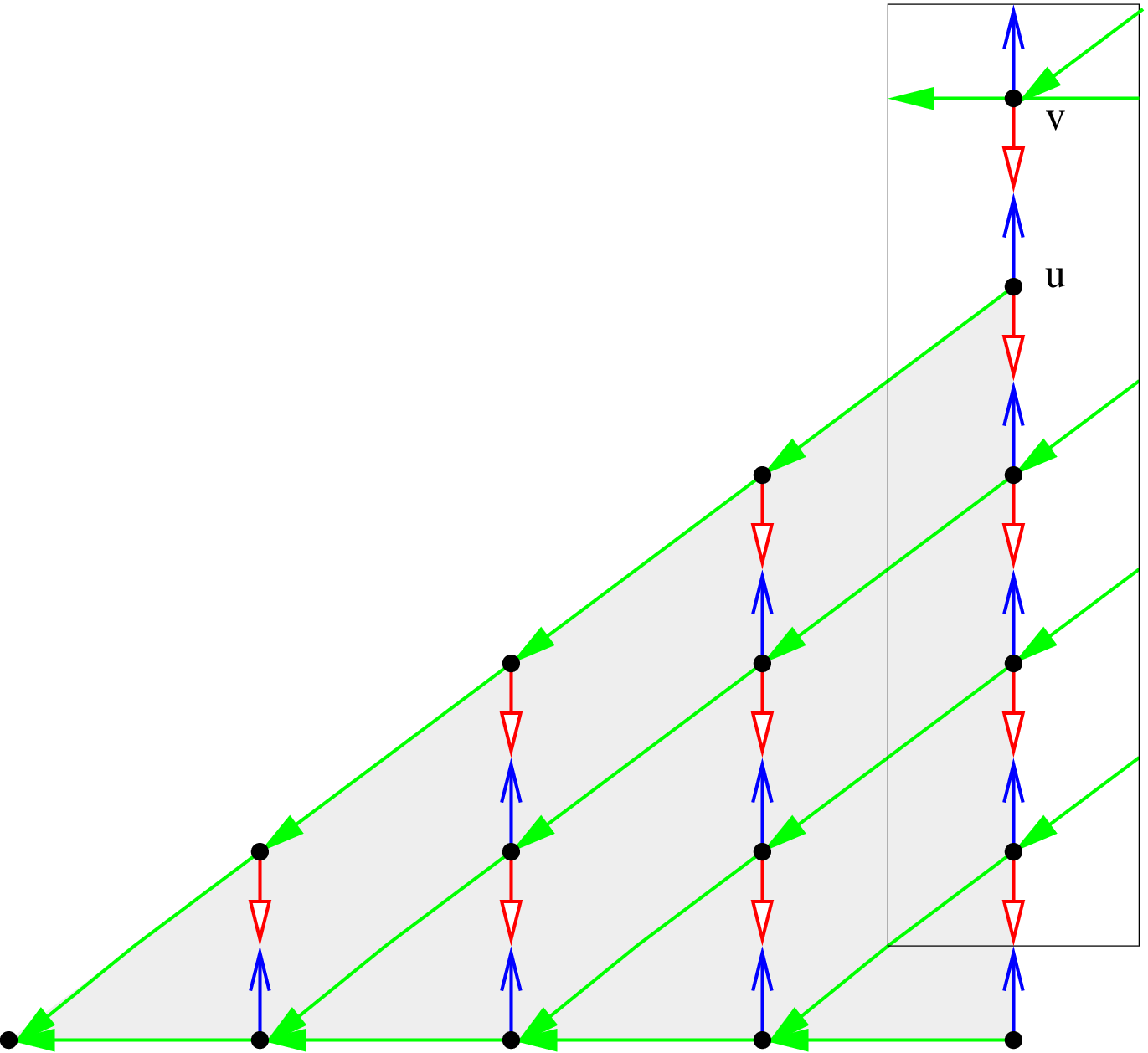}
 \caption{Example of a toroidal map where $d_1(u,z_1(u))$ has size
   $\Omega(n\times f)$.}
 \label{fig:examplelongue}
 \end{figure} 

\begin{lemma}
  \label{lem:regionorder}
  Let $u,v$ be vertices of $G^\infty$ such that $R_i(u)\subseteq
  R_i(v)$, then $u_i\leq v_i$. Moreover if $R_i(u)\subsetneq R_i(v)$,
  then
  $v_i-u_i>(N-n)(|R(L_{i-1}(u),L_{i-1}(v))|+|R(L_{i+1}(u),L_{i+1}(v))|)\geq
  0$.
\end{lemma}

\begin{proof}
  We distinguish two cases depending of the fact that the Schnyder
  wood is of type 2.i or not.

  \noindent
\emph{$\bullet$ Case 1: The Schnyder wood is not of Type 2.i.}

  Suppose first that $u$ and $v$ are both in a region delimited by two
  consecutive lines of color $i-1$ and two consecutive lines of color
  $i+1$.  Let $L_{i-1}^j,L_{i-1}^{j+1}$, $L_{i+1}^{j'},L_{i+1}^{j'+1}$
  be these lines such that $L_{i-1}^{j+1}$ is on the positive side of
  $L_{i-1}^j$, $L_{i+1}^{j'+1}$ is on the positive side of
  $L_{i+1}^{j'}$, and $L_k^\ell \in \mathcal L_k^\ell$ (see
  Figure~\ref{fig:regionrectangle}).
  We distinguish cases corresponding to equality or not between lines
  $L_{i-1}(u)$, $L_{i-1}(v)$ and $L_{i+1}(u)$, $L_{i+1}(v)$.

\begin{figure}[!h]
\center
\input{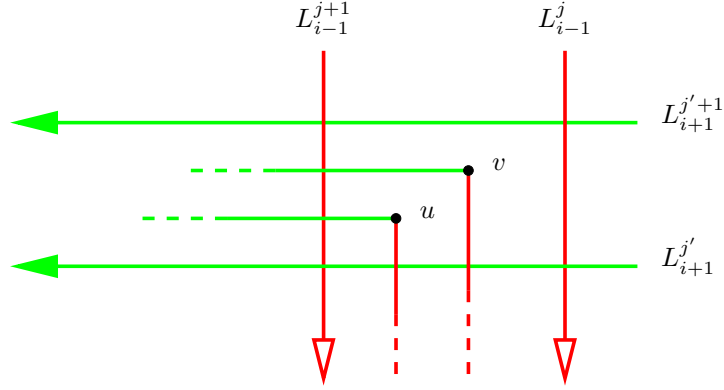}
\caption{Position of $u$ and $v$ in the proof of
  Lemma~\ref{lem:regionorder}}
\label{fig:regionrectangle}
\end{figure}

\noindent
\emph{$\star$ Case 1.1: $L_{i-1}(u)=L_{i-1}(v)$ and
  $L_{i+1}(u)=L_{i+1}(v)$.}  Then
$v_i-u_i=d_i(v,z_i(v))-d_i(u,z_i(u))=d_i(v,u)$. Thus clearly, if
$R_i(u)\subseteq R_i(v)$, then $u_i\leq v_i$ and if $R_i(u)\subsetneq
R_i(v)$, 
$v_i-u_i>0=(N-n)(|R(L_{i-1}(u),L_{i-1}(v))|+|R(L_{i+1}(u),L_{i+1}(v))|)$.

\noindent
\emph{$\star$ Case 1.2: $L_{i-1}(u)=L_{i-1}(v)$ and $L_{i+1}(u)\neq
  L_{i+1}(v)$.}  As $u\in R_i(v)$, we have $L_{i+1}(u)=L_{i+1}^{j'}$
and $L_{i+1}(v)=L_{i+1}^{j'+1}$.  Then
$v_i-u_i=d_i(v,z_i(v))-d_i(u,z_i(u)) +
N(f_{i+1}(L_{i+1}(v))-f_{i+1}(L_{i+1}(u)))=d_i(v,z_i(v))-d_i(u,z_i(u))
+ N f_{i+1}^{j'}$. Let $u'$ be the intersection of $P_{i+1}(u)$ with
$L_{i-1}^{j+1}$ (maybe $u=u'$). Let $v'$ be the intersection of
$P_{i+1}(v)$ with $L_{i-1}^{j+1}$ (maybe $v=v'$). Since
$L_{i+1}(u)\neq L_{i+1}(v)$, we have $u'\neq v'$. Since $u\in R_i(v)$,
we have $u'\in R_i(v')$ and so $u'\in P_{i-1}(v')$. Then, by
Lemma~\ref{lem:regionslongue}, $d_i(z_i(v'),v')+d_i(u',z_i(u'))<(n-1)
f_{i+1}^{j'}$.  If $L_{i-1}(u)=L_{i-1}^{j+1}$ then, one can see that
$d_i(v,z_i(v))-d_i(u,z_i(u))\geq d_i(v',z_i(v'))-d_i(u',z_i(u'))$.  If
$L_{i-1}(u)=L_{i-1}^j$, one can see that
$d_i(v,z_i(v))-d_i(u,z_i(u))\geq
d_i(v',z_i(v'))-d_i(u',z_i(u'))-f_{i+1}^{j'}$.  So finally,
$v_i-u_i=d_i(v,z_i(v))-d_i(u,z_i(u)) + N f_{i+1}^{j'}\geq
d_i(v',z_i(v'))-d_i(u',z_i(u')) + (N-1) f_{i+1}^{j'} >(N-n)
f_{i+1}^{j'}=(N-n)(|R(L_{i-1}(u),L_{i-1}(v))|+|R(L_{i+1}(u),L_{i+1}(v))|)\geq 0$.

\noindent
\emph{$\star$ Case 1.3: $L_{i-1}(u)\neq L_{i-1}(v)$ and $L_{i+1}(u)=
  L_{i+1}(v)$.}  
This case is completely symmetric to the previous
case.

\noindent
\emph{$\star$ Case 1.4: $L_{i-1}(u)\neq L_{i-1}(v)$ and
  $L_{i+1}(u)\neq L_{i+1}(v)$.}  As $u\in R_i(v)$, we have
$L_{i+1}(u)=L_{i+1}^{j'}$, $L_{i+1}(v)=L_{i+1}^{j'+1}$,
$L_{i-1}(u)=L_{i-1}^{j+1}$, and $L_{i-1}(v)=L_{i-1}^{j}$.  Then
$v_i-u_i=d_i(v,z_i(v))-d_i(u,z_i(u)) +
N(f_{i+1}(L_{i+1}(v))-f_{i+1}(L_{i+1}(u)))-
N(f_{i-1}(L_{i-1}(v))-f_{i-1}(L_{i-1}(u)))
=d_i(v,z_i(v))-d_i(u,z_i(u)) + N f_{i+1}^{j'}+Nf_{i-1}^{j}$.  Let $u'$
be the intersection of $P_{i+1}(u)$ with $L_{i-1}^{j+1}$ (maybe
$u=u'$).  Let $u''$ be the intersection of $P_{i-1}(u)$ with
$L_{i+1}^{j'}$ (maybe $u=u''$).  Let $v'$ be the intersection of
$P_{i+1}(v)$ with $L_{i-1}^{j+1}$ (maybe $v=v'$).  Let $v''$ be the
intersection of $P_{i-1}(v)$ with $L_{i+1}^{j'}$ (maybe $v=v''$).
Since $L_{i+1}(u)\neq L_{i+1}(v)$, we have $u'\neq v'$. Since $u\in
R_i(v)$, we have $u'\in R_i(v')$ and so $u'\in P_{i-1}(v')$. Then, by
Lemma~\ref{lem:regionslongue}, $d_i(z_i(v'),v')+d_i(u',z_i(u'))<(n-1)
f_{i+1}^{j'}$.  Symmetrically,
$d_i(z_i(v''),v'')+d_i(u'',z_i(u''))<(n-1) f_{i-1}^{j}$.  Moreover, we
have $d_i(v,z_i(v))-d_i(u,z_i(u))\geq
d_i(v',z_i(v'))-d_i(u',z_i(u'))+d_i(v'',z_i(v''))-d_i(u'',z_i(u''))-f_{i+1}^{j'}-f_{i-1}^j$.
So finally, $v_i-u_i=d_i(v,z_i(v))-d_i(u,z_i(u)) + N
f_{i+1}^{j'}+Nf_{i-1}^{j}\geq
d_i(v',z_i(v'))-d_i(u',z_i(u'))+d_i(v'',z_i(v''))-d_i(u'',z_i(u'')) +
(N-1) f_{i+1}^{j'} + (N-1) f_{i-1}^{j}> (N-n) f_{i+1}^{j'}+ (N-n)
f_{i-1}^{j}=(N-n)(|R(L_{i-1}(u),L_{i-1}(v))|+|R(L_{i+1}(u),L_{i+1}(v))|) \geq 0$.

Suppose now that $u$ and $v$ do not lie in a region delimited by two
consecutive lines of color $i-1$ and/or in a region delimited by two
consecutive lines of color $i+1$.  One can easily find distinct
vertices $w_0,\ldots, w_r$ ($w_i$, $1\leq i< r$  chosen at
intersections of monochromatic lines of colors $i-1$ and $i+1$) such
that $w_0=u$, $w_r=v$, and for $0\leq \ell \leq r-1$, we have
$ R_i(w_\ell)\subsetneq R_i(w_{\ell+1})$ and $w_\ell,w_{\ell+1}$ are both in a
region delimited by two consecutive lines of color $i-1$ and in a
region delimited by two consecutive lines of color $i+1$. Thus by the
first part of the proof, $(w_\ell)_i-(w_{\ell+1})_i>
(N-n)(|R(L_{i-1}(w_{\ell+1}),L_{i-1}(w_\ell))|+|R(L_{i+1}(w_{\ell+1}),L_{i+1}(w_\ell))|)$.
Thus $v_i-u_i> (N-n)\sum_\ell
(|R(L_{i-1}(w_{\ell+1}),L_{i-1}(w_\ell))|+|R(L_{i+1}(w_{\ell+1}),L_{i+1}(w_\ell))|)$.
For any $a,b,c$ such $R_i(a)\subseteq R_i(b)\subseteq R_i(c)$, we have
$|R(L_{j}(a),L_{j}(b))|+|R(L_{j}(b),L_{j}(c))|=|R(L_{j}(a),L_{j}(c))|$.
Thus we obtain the result by summing the size of the regions.

\noindent
\emph{$\bullet$ Case 2: The Schnyder wood is of Type 2.i.}

Suppose first that $u$ and $v$ are both in a region delimited by two
consecutive lines of color $i+1$. 

Let $L_{i+1}^{j},L_{i+1}^{j+1}$ be these lines such that
$L_{i+1}^{j+1}$ is on the positive side of $L_{i+1}^{j}$, and
$L_{i+1}^\ell \in \mathcal L_{i+1}^\ell$.  We can assume that we do
not have both $u$ and $v$ in $L_{i+1}^{j+1}$ (by eventually choosing
other consecutive lines of color $i+1$). We consider two cases :

\noindent
\emph{$\star$ Case 2.1: $v\notin L_{i+1}^{j+1}$.}  Then by
Lemma~\ref{lem:linesintersection},
$L_{i+1}^{j}=L_{i+1}(u)=(L_{i-1}(u))^{-1}=L_{i+1}(v)=(L_{i-1}(v))^{-1}$.
Then $v_i-u_i=d_i(v,z_i(v))-d_i(u,z_i(u))=d_i(v,u)$. Thus clearly, if
$R_i(u)\subseteq R_i(v)$, then $u_i\leq v_i$ and if $R_i(u)\subsetneq
R_i(v)$, then
$v_i-u_i>0=(N-n)(|R(L_{i-1}(u),L_{i-1}(v))|+|R(L_{i+1}(u),L_{i+1}(v))|)$.

\noindent
\emph{$\star$ Case 2.2: $v\in L_{i+1}^{j+1}$.}  Then
$L_{i+1}^{j+1}=L_{i+1}(v)=(L_{i-1}(v))^{-1}$ and $d_i(v,z_i(v))=0$.
By assumption $u\notin L_{i+1}^{j+1}$ and by
Lemma~\ref{lem:linesintersection},
$L_{i+1}^{j}=L_{i+1}(u)=(L_{i-1}(u))^{-1}$.  Then
$v_i-u_i=d_i(v,z_i(v))-d_i(u,z_i(u))+
N(f_{i+1}(L_{i+1}(v))-f_{i+1}(L_{i+1}(u)))-
N(f_{i-1}(L_{i-1}(v))-f_{i-1}(L_{i-1}(u))) = - d_i(u,z_i(u))+ 2N
f_{i+1}^{j}$.  Let $L_i$ and $L'_i$ be two consecutive $i$-lines such
that $u$ lies in the region between them and $L'_i$ is on the right of
$L_i$.  Let $u'$ be the intersection of $P_{i+1}(u)$ with $L_{i}$
(maybe $u=u'$).  Let $u''$ be the intersection of $P_{i-1}(u)$ with
$L'_{i}$ (maybe $u=u''$).  Then, by Lemma~\ref{lem:regionslongue},
$d_i(u',z_i(u'))<(n-1) f_{i+1}^{j}$ and $d_i(u'',z_i(u''))<(n-1)
f_{i+1}^{j}$.  Thus we have $d_i(u,z_i(u))\leq
d_i(u',z_i(u'))+d_i(u'',z_i(u''))+ f_{i+1}^{j}< (2(n-1)+1)
f_{i+1}^{j}$.  So finally, $v_i-u_i> -(2n-1) f_{i+1}^{j} + 2N
f_{i+1}^{j}> 2(N-n)f_{i+1}^{j} =
(N-n)(|R(L_{i-1}(u),L_{i-1}(v))|+|R(L_{i+1}(u),L_{i+1}(v))|)\geq 0$.

If $u$ and $v$ do not lie in a region delimited by two consecutive
lines of color $i+1$, then as in case 1, one can find intermediate
vertices to obtain the result.
\end{proof}

 \begin{lemma}
\label{lem:edgesbounded}
   If two vertices $u,v$ are adjacent, then for each color $i$, we
   have $|v_i - u_i |\leq 2Nf$.
 \end{lemma}

\begin{proof}
  Since $u,v$ are adjacent, they are both in a region delimited by two
  consecutive lines of color $i-1$ and in a region delimited by two
  consecutive lines of color $i+1$.  Let $L_{i-1}^j,L_{i-1}^{j+1}$ be
  these two consecutive lines of color $i-1$ and
  $L_{i+1}^{j'},L_{i+1}^{j'+1}$ these two consecutive lines of color
  $i+1$ with $L_k^\ell \in \mathcal L_k^\ell$, $L_{i-1}^{j+1}$ is on
  the positive side of $L_{i-1}^j$ and $L_{i+1}^{j'+1}$ is on the
  positive side of $L_{i+1}^{j'}$ (see
  Figure~\ref{fig:regionrectangle} when the Schnyder wood is not of
  Type 2.i).  If the Schnyder wood is of Type 2.i we assume that
  $L_{i-1}^{j+1}=(L_{i+1}^{j'})^{-1}$ and
  $L_{i-1}^{j}=(L_{i+1}^{j'+1})^{-1}$.  Let $z$ be a vertex on the
  intersection of $L_{i-1}^{j+1}$ and $L_{i+1}^{j'}$.  Let $z'$ be a
  vertex on the intersection of $L_{i-1}^{j}$ and $L_{i+1}^{j'+1}$.
  Thus we have $R_i(z)\subseteq R_i(u)\subseteq R_i(z')$ and
  $R_i(z)\subseteq R_i(v)\subseteq R_i(z')$. So by
  Lemma~\ref{lem:regionorder}, $z_i\leq u_i\leq z'_i$ and $z_i\leq
  v_i\leq z'_i$. So $|v_i-u_i|\leq
  z'_i-z_i=N(f_{i+1}(L_{i+1}^{j'+1})-f_{i+1}(L_{i+1}^{j'})-
  N(f_{i-1}(L_{i-1}^{j})-f_{i-1}(L_{i-1}^{j+1})) =N
  f_{i+1}^{j'}+Nf_{i-1}^{j}\leq 2Nf$.
\end{proof}

We are now able to prove the following :

\begin{theorem}
\label{th:geodesic}
If $G$ is a toroidal map given with a Schnyder wood, then the mapping
of each vertex of $G^\infty$ on its region vector gives a geodesic
embedding of $G$.
\end{theorem}

\begin{proof}
  By Lemmas~\ref{lem:periodic} and~\ref{lem:notcolinear}, the mapping
  of $G^\infty$ on its region vector is periodic with respect to $S$,
  $S'$ that are not collinear.  For any pair $u,v$ of distinct vertices
  of $G^\infty$, by Lemma~\ref{lem:regionss}.(iii), there exists $i,j$
  with $R_i(u) \subsetneq R_i(v)$ and $R_j(v) \subsetneq R_j(u)$ thus,
  by Lemma~\ref{lem:regionorder}, $u_i<v_i$ and $v_j<u_j$. So
  $\mv^\infty$ is a set of pairwise incomparable elements of $\mr^3$.

  (D1) $\mv^\infty$ is a set of pairwise incomparable elements so
  the mapping between vertices of $G^\infty$ and
  $\mv^\infty$ is a bijection.

  (D2) Let $e=uv$ be an edge of $G^\infty$. We show that $w=u\vee v$
  is on the surface $S_{\mathcal V^\infty}$. By definition $u\vee v$
  is in $\mathcal D_{\mathcal V^\infty}$. Suppose, by contradiction
  that $w\notin S_{\mathcal V^\infty}$. Then there exist $x\in
  \mathcal V^\infty$ with $x<w$.  Let $x$ also denote the
  corresponding vertex of $G^\infty$.  Edge $e$ is in a region
  $R_i(x)$ for some $i$. So $u,v\in R_i(x)$ and thus by
  Lemma~\ref{lem:regionss}.(i), $R_i(u)\subseteq R_i(x)$ and
  $R_i(v)\subseteq R_i(x)$. Then by Lemma~\ref{lem:regionorder},
  $w_i=\max(u_i,v_i)\leq x_i$, a contradiction.  Thus the elbow
  geodesic between $u$ and $v$ is  on the surface.

  (D3) Consider a vertex $v\in\mathcal V$ and a color $i$.  Let $u$ be
  the extremity of the arc $e_i(v)$. We have $u\in R_{i-1}(v)$ and
  $u\in R_{i+1}(v)$, so by Lemma~\ref{lem:regionss}.(i),
  $R_{i-1}(u)\subseteq R_{i-1}(v)$ and $R_{i+1}(u)\subseteq
  R_{i+1}(v)$. Thus by Lemma~\ref{lem:regionss}.(iii),
  $R_{i}(v)\subsetneq R_{i}(u)$. So, by Lemma~\ref{lem:regionorder},
  $v_i<u_i$, $u_{i-1}\leq v_{i-1}$ and $u_{i+1}\leq v_{i+1}$. So the
  orthogonal arc of vertex $v$ in direction of the basis vector $e_i$
  is part of the elbow geodesic of the edge $e_i(v)$.

  (D4) Suppose there exists a pair of crossing edges $e=uv$ and
  $e'=u'v'$ on the surface $S_{\mathcal V^\infty}$. The two edges
  $e,e'$ cannot intersect on orthogonal arcs so they intersects on a
  plane orthogonal to one of the coordinate axis. Up to symmetry we
  may assume that we are in the situation of Figure~\ref{fig:crossing}
  with $u_1=u'_1$, $u_2>u'_2$ and $v_2<v'_2$. Between $u$ and $u'$,
  there is a path consisting of orthogonal arcs only. With (D3), this
  implies that there is a bi-directed path $P^*$ colored $0$ from $u$
  to $u'$ and colored $2$ from $u'$ to $u$. We have $u\in R_2(v)$, so
  by Lemma~\ref{lem:regionss}.(i), $R_2(u)\subseteq R_2(v)$.  We have
  $u'\in R_2(u)$, so $u'\in R_2(v)$.  If $P_0(v)$ contains $u'$, then
  there is a contractible cycle containing $v,u,u'$ in $G_1\cup
  G_{0}^{-1}\cup G_{2}^{-1}$, contradicting
  Lemma~\ref{lem:nocontractiblecycle}, so $P_0(v)$ does not contain
  $u'$.  If $P_1(v)$ contains $u'$, then $u'\in P_1(u)\cap P_0(u)$,
  contradicting Lemma~\ref{lem:nocommon}.  So $u'\in
  R_2^\circ(v)$. Thus the edge $u'v'$ implies that $v'\in R_2(v)$. So
  by Lemma~\ref{lem:regionorder}, $v'_2\leq v_2$, a contradiction.
\end{proof}

\begin{figure}[h!]
\center
\includegraphics[scale=0.5]{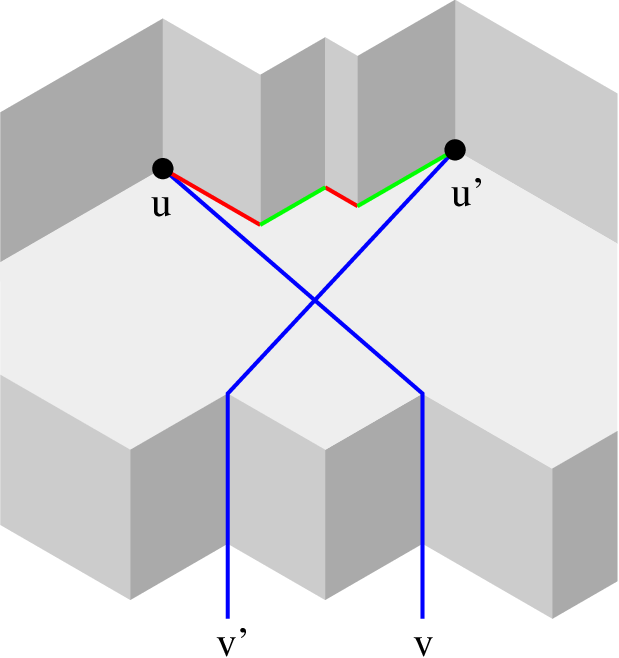}
\caption{A pair of crossing elbow geodesic}
\label{fig:crossing}
\end{figure}


Theorem~\ref{th:existence} and~\ref{th:geodesic} implies
Theorem~\ref{th:triortho}.

One can ask what is the ``size'' of the obtained geodesic embedding of
Theorem~\ref{th:geodesic} ?  Of course this mapping is infinite so
there is no real size, but as the object is periodic one can consider
the smallest size of the vectors such that the mapping is periodic
with respect to them.  There are several such pairs of vectors, one is
$S,S'$.  Recall that $S_i=N(c_{i+1}-c_{i-1})f$ and
$S'_i=N(c'_{i+1}-c'_{i-1})f$. Unfortunately the size of $S,S'$ can be
arbitrarily large. Indeed, the values of $c_{i+1}-c_{i-1}$ and
$c'_{i+1}-c'_{i-1}$ are unbounded as a toroidal map can be
artificially ``very twisted'' in the considered flat torus
representation (independently from the number of vertices or faces).
Nevertheless we can prove existence of bounded size vectors for which
the mapping is periodic with respect to them.

\begin{lemma}
\label{th:gamma}
If $G$ is a toroidal map given with a Schnyder wood, then the mapping
of each vertex of $G^\infty$ on its region vector gives a periodic
mapping of $G^\infty$ with respect to non collinear vectors $Y$ and
$Y'$ where the size of $Y$ and $Y'$ is in $\mathcal O(\gamma
Nf)$.  In general we have $\gamma\leq n$ and in the case where $G$ is
a simple toroidal triangulation given with a Schnyder wood obtained by
Theorem~\ref{th:schnydersimple}, we have $\gamma=1$.
\end{lemma}

\begin{proof}
  By Lemma~\ref{lem:copies}, the vectors $Z_{i-1},Z_{i+1}$ (when the
  Schnyder wood is not of Type 2.i) span a subset of $S,S'$ (it can
  happen that this subset is strict).  Thus in the parallelogram
  delimited by the vectors $Z_{i-1},Z_{i+1}$ (that is a parallelogram
  by Lemma~\ref{lem:dim}), there is a parallelogram with sides $Y,Y'$
  containing a copy of $V$. The size of the vectors $Z_i$ is in
  $\mathcal O(\gamma Nf)$ and so $Y$ and $Y'$ also.

  In general we have $\gamma_i\leq n$ as each intersection between two
  monochromatic cycles of $G$ of color $i-1$ and $i+1$ corresponds to
  a different vertex of $G$ and thus $\gamma \leq n$.  In the case of
  simple toroidal triangulation given with a Schnyder wood obtained by
  Theorem~\ref{th:schnydersimple}, we have, for each color $i$,
  $\gamma_i=1$, and thus $\gamma=1$.
\end{proof}

We use the example of the toroidal map $G$ of
Figure~\ref{fig:example-schnyder} to illustrate the region vector
method.  This toroidal map has $3$ vertices, $4$ faces and $7$
edges. There are two edges that are oriented in two directions. The
Schnyder wood is of Type 1, with two $1$-cycles.  We choose as origin
the three bold monochromatic lines of
Figure~\ref{fig:example-coordinate}.  Then the region vectors of the
vertices of $G^\infty$ are $\{u\in \mr^3\,|\,\exists\, v \in \mv,\
k_1,k_2\in \mz$ such that $u=v+kS+k'S\}$, with $N=n=3$, $\mathcal
V=\{(0,0,0), (0,12,-11), (6,12,-18)\}$, $S=(-12,24,-12)$,
$S'=(12,24,-36)$, $c_0=-1$, $c_1=0$, $c_2=1$, $c'_0=-2$, $c'_1=1$,
$c'_2=0$. The points are not coplanar. They are on the two different
planes of equation $x+y+z=0$ and $x+y+z=1$.  The geodesic embedding
that is obtained by mapping each vertex to its region vector is the
geodesic embedding of Figure~\ref{fig:example-primal}. The
parallelogram has sides the vectors $S,S'$.

\begin{figure}[h!]
\center
\includegraphics[scale=0.3]{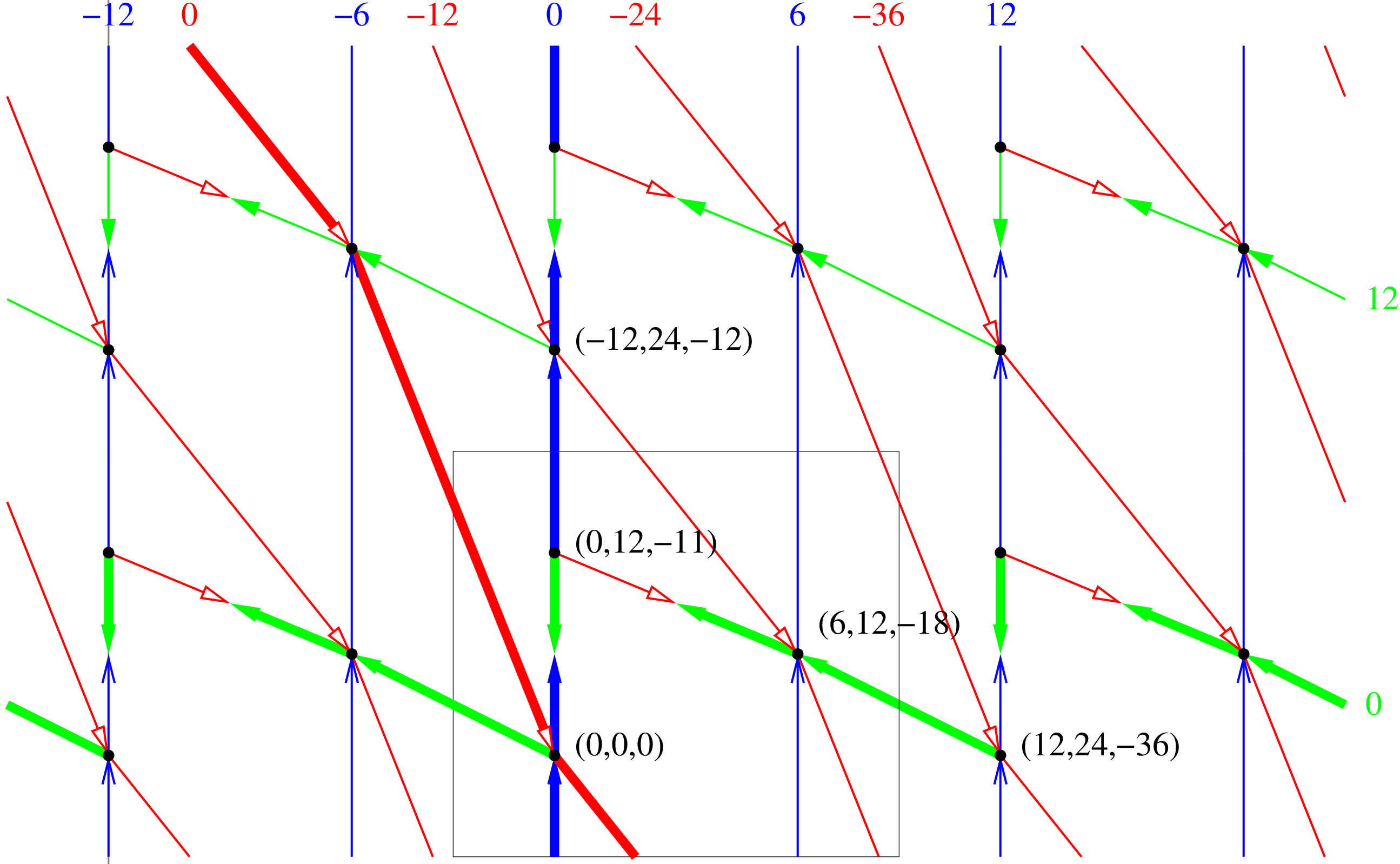}
\caption{Coordinates of the vertices.}
\label{fig:example-coordinate}
\end{figure}

Like in the plane, one can give weights to faces of $G$.
Then all their copies in $G^\infty$ have the same weight and instead
of counting the number of faces in each region one can compute the
weighted sum. 

Note that the geodesic embeddings of Theorem~\ref{th:geodesic} are not
necessarily rigid. A geodesic embedding is \emph{rigid}~\cite{FZ08,Mil02} if for every
pair $u,v\in \mathcal V$ such that $u\vee v$ is in $\mathcal
S_{\mathcal V}$, then $u$ and $v$ are the only elements of $\mathcal
V$ that are dominated by $u\vee v$.  The
geodesic embedding of Figure~\ref{fig:example-primal} is not rigid has
the bend corresponding to the loop of color $1$ is dominated by three
vertices of $G^\infty$. We do not know if it is possible to build a
rigid geodesic embedding from the Schnyder wood of a toroidal
map. Maybe a technique similar to the one presented in \cite{FZ08} can
be generalized to the torus. 

It has been already mentioned that in the geodesic embeddings of
Theorem~\ref{th:geodesic} the points corresponding to vertices are not
coplanar. The problem to build a coplanar geodesic embedding from the
Schnyder wood of a toroidal map is open. In the plane, there are some
examples of maps $G$~\cite{FZ08} for which it is no possible to require both
rigidity and co-planarity. Thus the same is  true
in the torus for the graph $G^+$.

Another question related to co-planarity is whether one can require that
the points of the orthogonal surface corresponding to edges of the
graph (i.e. bends) are coplanar.  This property is related to contact
representation by homotopic triangles \cite{FZ08}. It is known
that in the plane, not all Schnyder woods are supported by such
surfaces.  Kratochvil's conjecture~\cite{Kra07}, recently
proved~\cite{GLP11b}, states that every 4-connected planar
triangulation admits a contact representation by homothetic
triangles. Can this be extended to the torus ?

When considering non necessarily homothetic triangles, it has been
proved~\cite{FOR94} that there is a bijection between Schnyder woods
of planar triangulations and contact representations by
triangles. This results has been generalized to internally 3-connected
planar map~\cite{GLP11} by exhibiting a bijection between Schnyder
woods of internally 3-connected planar maps and primal-dual contact
representations by triangles (i.e. representations where both the
primal and the dual are represented). It would be interesting to
generalize these results to the torus.

\section{Duality of orthogonal surfaces}
\label{sec:dualortho}

Given an orthogonal surface generated by $\mv$, let $\mathcal
F_{\mathcal V}$ be the maximal points of $\ms$, i.e. the points of
$\ms$ that are not dominated by any vertex of $\ms$.  If $A,B\in
\mathcal F_{\mathcal V}$ and $A\wedge B\in \ms$, then $\ms$ contains
the union of the two line segments joining $A$ and $B$ to $A\wedge
B$. Such arcs are called \emph{dual elbow geodesic}.  The \emph{dual
  orthogonal arc} of $A\in \mathcal F_{\mathcal V}$ in the direction
of the standard basis vector $e_i$ is the intersection of the ray
$A+\lambda e_i$ with $\ms$.

Given a toroidal map $G$, let $G^{\infty*}$ be the dual of $G^\infty$.
A \emph{dual geodesic embedding} of $G$ is a drawing of $G^{\infty*}$
on the orthogonal surface $\mathcal S_{\mathcal V^\infty}$, where
$\mathcal V^\infty$ is a periodic mapping of $G^\infty$ with respect
to two non collinear vectors, satisfying the following (see example of
Figure~\ref{fig:example-dual}):

\begin{itemize}
\item[(D1*)] There is a bijection between the vertices of
  $G^{\infty*}$ and $\mathcal F_{\mathcal V^\infty}$.

\item[(D2*)] Every edge of $G^{\infty*}$ is a dual elbow geodesic.

\item[(D3*)] Every dual orthogonal arc in $\mathcal S_{\mathcal V^\infty}$ is part of an edge of $G^{\infty*}$.

\item[(D4*)] There are no crossing edges in the embedding of $G^{\infty*}$ on $\mathcal S_{\mathcal V^\infty}$.
\end{itemize}

\begin{figure}[h!]
\center
\includegraphics[scale=0.2]{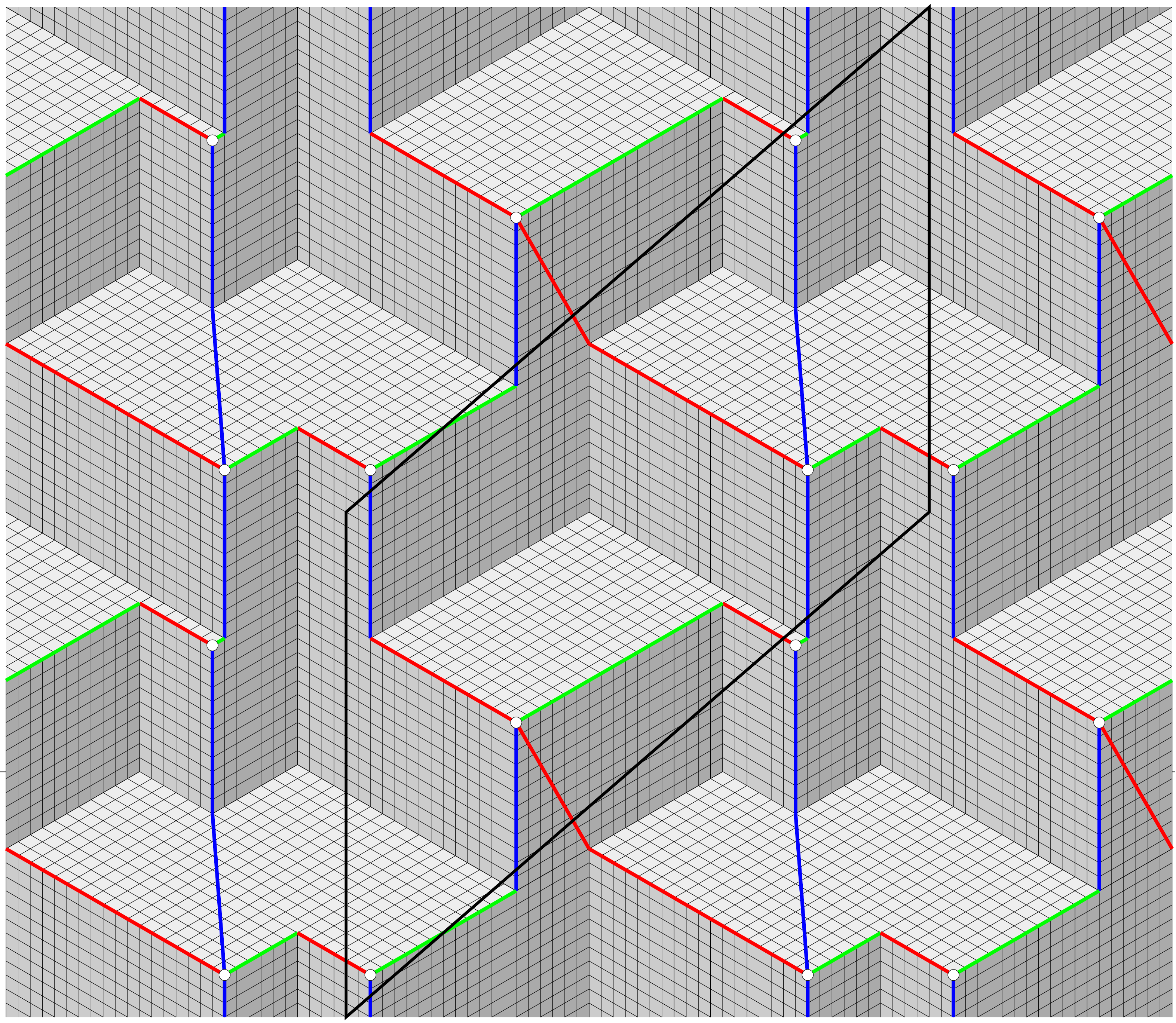}
\caption{Dual geodesic embedding of the toroidal map of
  Figure~\ref{fig:example-schnyder}.}
\label{fig:example-dual}
\end{figure}

Let $G$ be a toroidal map given with a Schnyder wood. Consider the
mapping of each vertex on its region vector.  We consider the dual of
the Schnyder wood of $G$. By Lemma~\ref{lem:dual}, it is a Schnyder
wood of $G^*$.  A face $F$ of $G^{\infty}$ is mapped on the point
$\bigvee_{v\in F}v$.  Let $\widetilde{G^\infty}$ be a simultaneous
drawing of $G^\infty$ and $G^{\infty*}$ such that only dual edges
intersect. To avoid confusion, we note $R_i$ the regions of the primal
Schnyder wood and $R_i^*$ the regions of the dual Schnyder wood.

\begin{lemma}
\label{lem:maximalpoint}
  For any face $F$ of $G^{\infty}$, we have that $\bigvee_{v\in F}v$ is
  a maximal point of $\mathcal S_{\mathcal V^\infty}$.
\end{lemma}

\begin{proof}
  Let $F$ be a face of $G^{\infty}$.  For any vertex $u$ of $\mathcal
  V^\infty$, there exists a color $i$, such that the face $F$ is in
  the region $R_i(u)$. Thus for $v\in F$, we have $v\in R_i(u)$. By
  Lemma~\ref{lem:regionorder}, we have $v_i\leq u_i$ and so $F_i\leq
  u_i$.  So $F=\bigvee_{v\in F} v$ is a point of $\mathcal S_{\mathcal
    V^\infty}$.

  Suppose, by contradiction, that $F$ is not a maximal point of
  $\mathcal S_{\mathcal V^\infty}$. Then there is a point $\alpha \in
  \mathcal S_{\mathcal V^\infty}$ that dominates $F$ and for at least
  one coordinate $j$, we have $F_j < \alpha_j$.  By
  Lemma~\ref{lem:angle}, the angles at $F$ form, in counterclockwise
  order, nonempty intervals of $0$'s, $1$'s and $2$'s.  For each color,
  let $z^i$ be a vertex of $F$ with angle $i$. We have $F$ is in the
  region $R_i(z^i)$. So $z^{i-1}\in R_i(z^i)$ and by
  Lemma~\ref{lem:regionss}.(i), we have $R_i(z^{i-1})\subseteq
  R_i(z^i)$.  Since $F$ is in $R_{i-1}(z^{i-1})$, it is not in
  $R_i(z^{i-1})$ and thus $R_i(z^{i-1})\subsetneq R_i(z^i)$.  Then by
  Lemma~\ref{lem:regionorder}, we have $(z^{i-1})_i<(z^{i})_i$ and
   symmetrically $(z^{i+1})_i<(z^{i})_i$.
  So $F_{j-1}=(z^{j-1})_{j-1}>(z^j)_{j-1}$ and 
  $F_{j+1}>(z^j)_{j+1}$. Thus $\alpha$ strictly dominates $z^j$, a
  contradiction to $\alpha \in \mathcal S_{\mathcal V^\infty}$. Thus
  $F$ is a maximal point of $\mathcal S_{\mathcal V^\infty}$
\end{proof}

\begin{lemma}
\label{lem:dualorder}
   If two faces $A,B$ are such that $R_i^*(B)\subseteq R_i^*(A)$,
    then $A_i\leq B_i$.
\end{lemma}

\begin{proof}
  Let $v\in B$ be a vertex whose angle at $B$ is labeled $i$. We have
  $v\in R^*_i(B)$ and so $v\in R^*_i(A)$. In $\widetilde{G^\infty}$,
  the path $P_{i}(v)$ cannot leave $R^*_i(A)$, the path $P_{i+1}(v)$
  cannot intersect $P_{i+1}(A)$ and the path $P_{i-1}(v)$ cannot
  intersect $P_{i-1}(A)$. Thus $P_{i+1}(v)$ intersect $P_{i-1}(A)$ and
  the path $P_{i-1}(v)$ cannot intersect $P_{i+1}(A)$.  So $A\in
  R_i(v)$. Thus for all $u\in A$, we have $u\in R_i(v)$, so
  $R_i(u)\subseteq R_i(v)$, and so $u_i\leq v_i$. Then $A_i=\max_{u\in
    A} u_i\leq v_i\leq \max_{w\in B} w_i=B_i$.
\end{proof}

\begin{theorem}
\label{th:dualgeodesic}
If $G$ is a toroidal map given with a Schnyder wood and each vertex of
$G^\infty$ is mapped on its region vector, then the mapping of each
face of $G^{\infty*}$ on the point $\bigvee _{v\in F} v$ gives a dual
geodesic embedding of $G$.
\end{theorem}

\begin{proof}
  By Lemmas~\ref{lem:periodic} and~\ref{lem:notcolinear}, the mapping
  is periodic with respect to non collinear vectors.

  (D1*) Consider a counting of elements on the orthogonal surface,
  where we count two copies of the same object just once (note that we
  are on an infinite and periodic object).  We have that the sum of
  primal orthogonal arcs plus dual ones is exactly $3m$. There are
  $3n$ primal orthogonal arcs and thus there are $3m-3n=3f$ dual
  orthogonal arcs.  Each maximal point of $\mathcal S_{\mathcal
    V^\infty}$ is incident to $3$ dual orthogonal arcs and there is no
  dual orthogonal arc incident to two distinct maximal points. So
  there is $f$ maximal points. Thus by Lemma~\ref{lem:maximalpoint},
  we have a bijection between faces of $G^\infty$ and maximal points
  of $\mathcal S_{\mathcal V^\infty}$.

  Let $\mathcal V^{\infty*}$ be the maximal points of $\mathcal
  S_{\mathcal V^\infty}$.  Let $\mathcal D^*_{\mathcal V^\infty}
  =\{A\in\mathbb R^3\ |\ \exists\, B\in \mathcal V^{\infty*}$ such
  that $A\leq B\}$. Note that the boundary of $\mathcal D^*_{\mathcal V^\infty}$
  is $\mathcal S_{\mathcal V^\infty}$.

  (D2*) Let $e=AB$ be an edge of $G^{\infty*}$.  We show that
  $w=A\wedge B$ is on the surface $S_{\mathcal V^\infty}$.  By
  definition $w$ is in $\mathcal D^*_{\mathcal V^\infty}$.  Suppose,
  by contradiction that $w\notin S_{\mathcal V^\infty}$. Then there
  exist $C$ a maximal point of $S_{\mathcal V^\infty}$ with $w<C$.  By
  the bijection (D1*) between maximal point and vertices of
  $G^{\infty*}$, the point $C$ corresponds to a vertex of
  $G^{\infty*}$, also denoted $C$.  Edge $e$ is in a region $R_i^*(C)$
  for some $i$. So $A,B\in R_i^*(C)$ and thus, by
  Lemma~\ref{lem:regionss}.(i), $R_i^*(A)\subseteq R_i^*(C)$ and
  $R_i^*(B)\subseteq R_i^*(C)$. Then by Lemma~\ref{lem:dualorder}, we
  have $C_i\leq \min(A_i,B_i)=w_i$, a contradiction.  Thus the dual
  elbow geodesic between $A$ and $B$ is also on the surface.

  (D3*) Consider a vertex $A$ of $G^{\infty*}$ and a color $i$.  Let
  $B$ be the extremity of the arc $e_i(A)$. We have $B\in
  R_{i-1}^*(A)$ and $B\in R_{i+1}^*(A)$, so by
  Lemma~\ref{lem:regionss}.(i), $R_{i-1}^*(B)\subseteq R_{i-1}^*(A)$
  and $R_{i+1}^*(B)\subseteq R_{i+1}^*(A)$. So by
  Lemma~\ref{lem:dualorder}, $A_{i-1}\leq B_{i-1}$ and $A_{i+1}\leq
  B_{i+1}$. As $A$ and $B$ are distinct maximal point of $\mathcal
  S_{\mathcal V^\infty}$, they are incomparable, thus $B_i<A_i$.  So
  the dual orthogonal arc of vertex $A$ in direction of the basis
  vector $e_i$ is part of edge $e_i(A)$.

  (D4*) Suppose there exists a pair of crossing edges $e=AB$ and
  $e'=A'B'$ of $G^{\infty*}$ on the surface $S_{\mathcal
    V^\infty}$. The two edges $e,e'$ cannot intersect on orthogonal
  arcs so they intersects on a plane orthogonal to one of the
  coordinate axis. Up to symmetry we may assume that we are in the
  situation $A_1=A'_1$, $A'_0>A_0$ and $B'_0<B_0$. Between $A$ and
  $A'$, there is a path consisting of orthogonal arcs only. With
  (D3*), this implies that there is a bi-directed path $P^*$ colored
  $2$ from $A$ to $A'$ and colored $0$ from $A'$ to $A$. We have $A\in
  R_0(B)$, so by Lemma~\ref{lem:regionss}.(i), $R_0(A)\subseteq
  R_0(B)$.  We have $A'\in R_0(A)$, so $A'\in R_0(B)$.  If $P_2(B)$
  contains $A'$, then there is a contractible cycle containing
  $A,A',B$ in $G_1^*\cup G_{0}^{*-1}\cup G_{2}^{*-1}$, contradicting
  Lemma~\ref{lem:nocontractiblecycle}, so $P_2(B)$ does not contain
  $A'$.  If $P_1(B)$ contains $A'$, then $A'\in P_1(A)\cap P_2(A)$,
  contradicting Lemma~\ref{lem:nocommon}.  So $A'\in
  R_0^\circ(B)$. Thus the edge $A'B'$ implies that $B'\in R_0(B)$. So
  by Lemma~\ref{lem:dualorder}, $B'_0\geq B_0$, a contradiction.
\end{proof}

Theorems~\ref{th:geodesic} and~\ref{th:dualgeodesic} can be combined to
obtain a simultaneous representation of a Schnyder wood and its dual
on an orthogonal surface. The projection of this 3-dimensional object
on the plane of equation $x+y+z=0$ gives a representation of the
primal and the dual where edges are allowed to have one bend and two
dual edges have to cross on their bends (see example of
Figure~\ref{fig:example-primal-dual}). 

\begin{theorem}
\label{cor:primal-dual}
An essentially 3-connected toroidal map admits a simultaneous flat
torus representation of the primal and the dual where edges are
allowed to have one bend and two dual edges have to cross on their
bends. Such a representation is contained in a (triangular) grid of
size $\mathcal O(n^2f)\times \mathcal O(n^2f)$ in general and
$\mathcal O(nf)\times\mathcal O(nf)$ if the map is a simple
triangulation. Furthermore the length of the edges are in $\mathcal
O(nf)$.
\end{theorem}

\begin{proof}
  Let $G$ be an essentially 3-connected toroidal map.  By
  Theorems~\ref{th:existence} (or Theorem~\ref{th:schnydersimple} if
  $G$ is a simple triangulation), $G$ admits a Schnyder wood (where
  monochromatic cycles of different colors intersect just once if $G$
  is simple).  By Theorems~\ref{th:geodesic}
  and~\ref{th:dualgeodesic}, the mapping of each vertex of $G^\infty$
  on its region vector gives a primal and dual geodesic
  embedding. Thus the projection of this embedding on the plane of
  equation $x+y+z=0$ gives a representation of the primal and the dual
  of $G^\infty$ where edges are allowed to have one bend and two dual
  edges have to cross on their bends.

  By Lemma~\ref{th:gamma}, the obtained mapping is a periodic
  mapping of $G^\infty$ with respect to non collinear vectors $Y$ and
  $Y'$ where the size of $Y$ and $Y'$ is in $\mathcal O(\gamma Nf)$,
  with $\gamma\leq n$ in general and $\gamma=1$ in case of a simple
  triangulation. Let $N=n$. The embedding gives a representation in
  the flat torus of sides $Y,Y'$ where the size of the vectors $Y$ and
  $Y'$ is in $\mathcal O(n^2f)$ in general and in $\mathcal O(nf)$ if
  the graph is simple and the Schnyder wood is obtained by
  Theorem~\ref{th:schnydersimple}.  By Lemma~\ref{lem:edgesbounded}
  the length of the edges in this representation are in $\mathcal
  O(nf)$.
\end{proof}

\begin{figure}[h!]
\center
\includegraphics[scale=0.2]{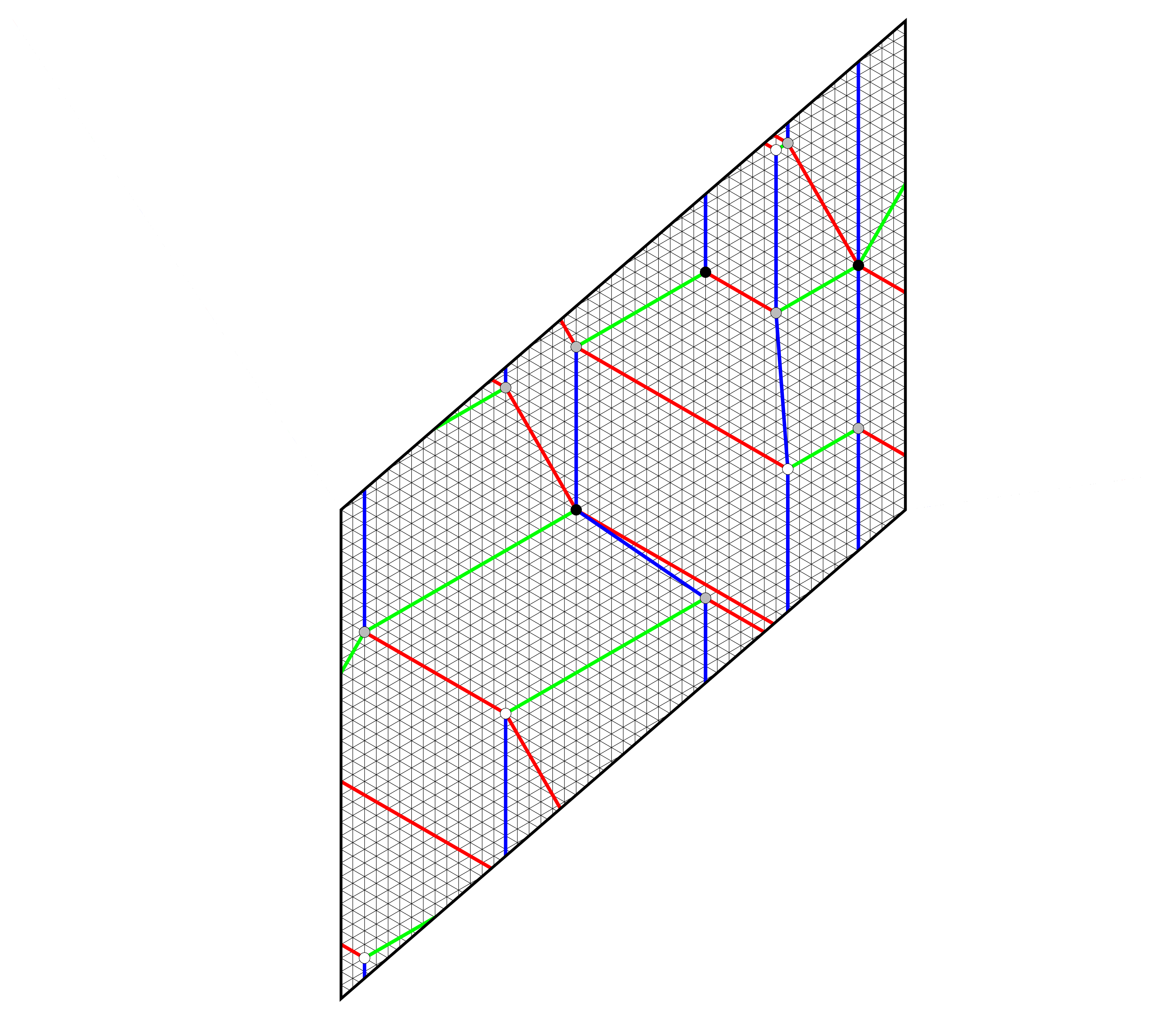}
\caption{Simultaneous representation of the primal and the dual of the
  toroidal map of Figure~\ref{fig:example-schnyder} with edges having
  one bend (in gray).}
\label{fig:example-primal-dual}
\end{figure}

\section{Straight-line representation of toroidal maps}
\label{sec:straight} 

The geodesic embedding obtained by the region vector method can be
used to obtain a straight-line representation of a toroidal map (see
Figure~\ref{fig:example-droit-primal}). For this purpose, we have to
choose $N$ bigger than previously.  Note that
Figure~\ref{fig:example-droit-primal} is the projection of the
geodesic embedding of Figure~\ref{fig:example-primal} obtained with
the value of $N=n$. In this particular case this gives a straight-line
representation but in this section we only prove that such a technique
works for triangulations and for $N$ sufficiently large. To obtain a
straight-line representation of a general toroidal map, one first has
to triangulate it.

\begin{figure}[h!]
\center
\includegraphics[scale=0.2]{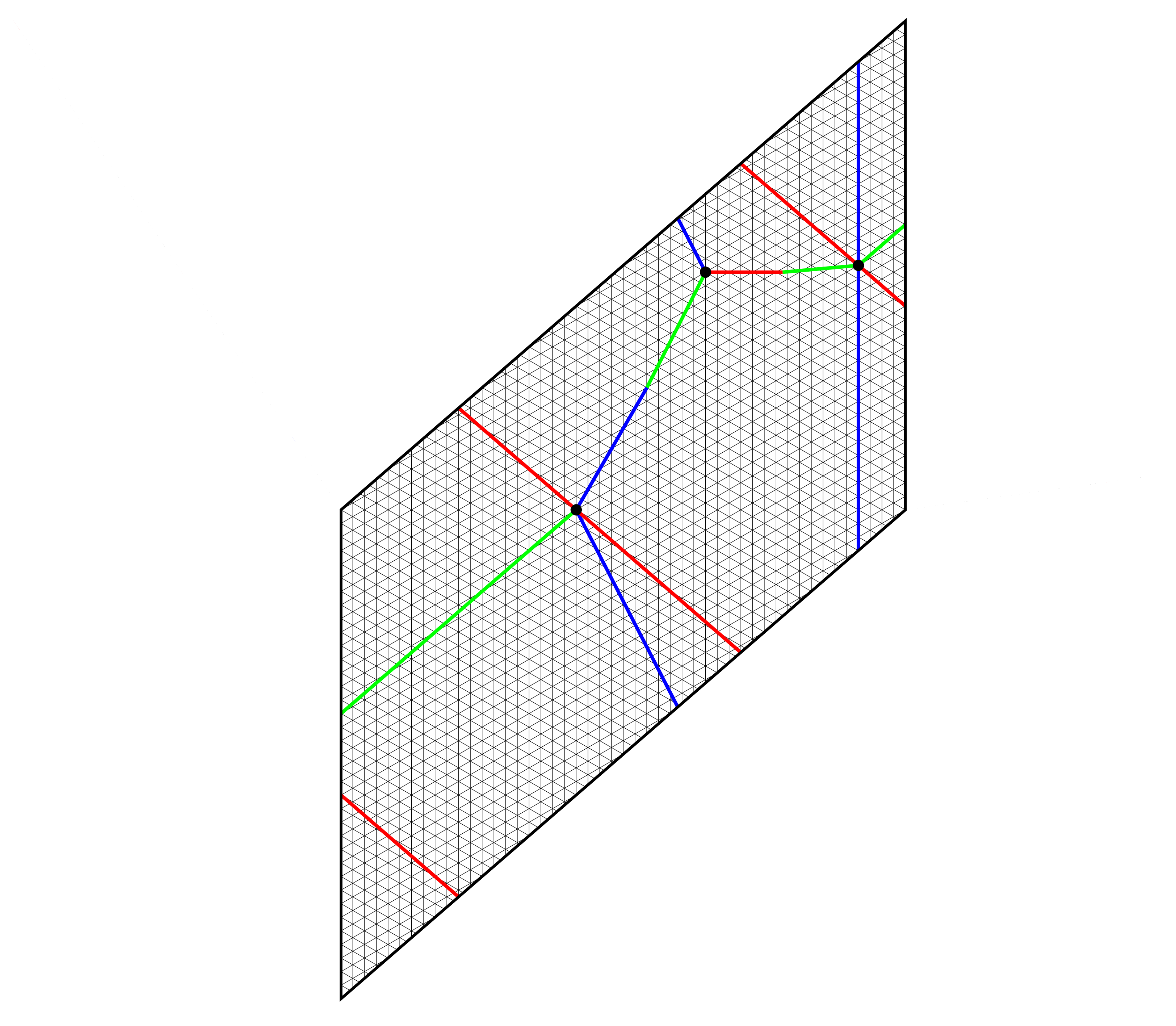}
\caption{Straight-line representation of the graph of
  Figure~\ref{fig:example-schnyder} obtained by projecting the
  geodesic embedding  of Figure~\ref{fig:example-primal}}
\label{fig:example-droit-primal}
\end{figure}

Let $G$ be a toroidal triangulation given with a Schnyder wood and
$V^\infty$ the set of region vectors of vertices of $G^\infty$.  The
Schnyder wood is of Type 1 by Theorem~\ref{lem:type}.  Recall that
$\gamma_i$ is the integer such that two monochromatic cycles of $G$ of
colors $i-1$ and $i+1$ intersect exactly $\gamma_i$ times.

\begin{lemma}
\label{lem:boundedtriangle}
For any vertex $v$, the number of faces in the bounded region
delimited by the three lines $L_{i}(v)$ is strictly less than
$(5\min(\gamma_i)+ \max(\gamma_i))f$.
\end{lemma}

\begin{proof}
  Suppose by symmetry that $\min(\gamma_i)=\gamma_1$.  Let $L_i=L_i(v)$
  and $z_i=z_i(v)$. Let $T$ be the bounded region delimited by the three
  monochromatic lines $L_{i}$.  The boundary of $T$ is a cycle $C$
  oriented clockwise or counterclockwise. Assume that $C$ is oriented
  counterclockwise (the proof is similar if oriented clockwise).  The
  region $T$ is on the left sides of the lines $L_i$.  We have
  $z_{i-1}\in P_{i}(z_{i+1})$.

  We define, for $j,k \in \mathbb{N}$, monochromatic lines $L_2(j)$,
  $L_0(k)$ and vertices $z(j,k)$ as follows (see
  Figure~\ref{fig:triangleprooff}).  Let $L_2(1)$ be the first
  $2$-line intersecting $L_0\setminus\{z_1\}$ while walking from
  $z_1$, along $L_0$ in the direction of $L_0$.  Let $L_0(1)$ be the
  first $0$-line of color $0$ intersecting $L_2\setminus\{z_1\}$ while
  walking from $z_1$, along $L_2$ in the reverse direction of $L_2$.
  Let $z(1,1)$ be the intersection between $L_2(1)$ and $L_0(1)$.  Let
  $z(j,1)$, $j\geq 0$, be the consecutive copies of $z(1,1)$ along
  $L_0(1)$ such that $z(j+1,1)$ is after $z(j,1)$ in the direction of
  $L_0(1)$. Let $L_2(j)$, $j\geq 0$, be the $2$-line of color $2$
  containing $z(j,1)$. Note that we may have $L_2=L_2(0)$ but in any
  case $L_2$ is between $L_2(0)$ and $L_2(1)$.  Let $z(j,k)$, $k\geq
  0$, be the consecutive copies of $z(j,1)$ along $L_2(j)$ such that
  $z(j,k+1)$ is after $z(j,k)$ in the reverse direction of $L_2(j)$.
  Let $L_0(k)$, $k\geq 0$, be the $0$-line containing $z(1,k)$. Note
  that we may have $L_0=L_0(0)$ but in any case $L_0$ is between
  $L_0(0)$ and $L_0(1)$.  Let $S(j,k)$ be the region delimited by
  $L_2(j),L_2(j+1),L_0(k),L_0(k+1)$. All the region $S(j,k)$ are
  copies of $S(0,0)$. The region $S(0,0)$ may contain several copies
  of a face of $G$ but the number of copies of a face in $S(0,0)$ is
  equal to $\gamma_1$.  Let $R$ be the unbounded region situated on
  the right of $L_0(1)$ and on the right of $L_2(1)$. As $P_0(v)$
  cannot intersect $L_0(1)$ and $P_2(v)$ cannot intersect $L_2(1)$,
  vertex $v$ is in $R$.  Let $P(j,k)$ be the subpath of $L_0(k)$
  between $z(j,k)$ and $z(j+1,k)$. All the lines $L_0(k)$ are composed
  only of copies of $P(0,0)$.  The interior vertices of the path
  $P(0,0)$ cannot contains two copies of the same vertex, otherwise
  there will be a vertex $z(j,k)$ between $z(0,0)$ and $z(1,0)$. Thus
  all interior vertices of a path $P(j,k)$ corresponds to distinct
  vertices of $G$.

  The Schnyder wood is of Type 1, thus $1$-lines are crossing
  $0$-lines.  As a line $L_0(k)$ is composed only of copies of
  $P(0,0)$, any path $P(j,k)$ is crossed by a $1$-line. Let $L'_1$ be
  the first $1$-line crossing $P(1,1)$ on a vertex $x$ while walking
  from $z(1,1)$ along $L_0(1)$. By (T1), line $L'_1$ is not
  intersecting $R\setminus \{z(1,1)\}$. As $v\in R$ we have $L_1$ is
  on the left of $L'_1$ (maybe $L_1=L'_1$). Thus the region $T$ is
  included in the region $T'$ delimited by $L_0,L'_1,L_2$.

  Let $y$ be the vertex where $L'_1$ is leaving $S(1,1)$.  We claim
  that $y\in L_2(1) $.  Note that by (T1), we have $y\in L_2(1)\cup
  P(1,2)$.  Suppose by contradiction that $y$ is an interior vertex of
  $P(1,2)$. Let $d_x$ be the length of the subpath of $P(1,1)$ between
  $z(1,1)$ and $x$.  Let $d_y$ be the length of the subpath of
  $P(1,2)$ between $z(1,2)$ and $y$.  Suppose $d_y<d_x$, then there
  should be a distinct copy of $L'_1$ intersecting $P(1,1)$ between
  $z(1,1)$ and $x$ on a copy of $y$, a contradiction to the choice of
  $L'_1$. So $d_x\leq d_y$. Let $A$ be the subpath of $L'_1$ between
  $x$ and $y$. Let $B$ be the subpath of $P(1,1)$ between $x$ and the
  copy of $y$ (if $d_x=d_y$, then $B$ is just a single vertex).
  Consider all the copies of $A$ and $B$ between lines $L_2(1)$ and
  $L_2(2)$, they form an infinite line $L$ situated on the right of
  $L_2(1)$ that prevents $L'_1$ from crossing $L_2(1)$, a
  contradiction.

  By the position of $x$ and $y$. We have $L'_1$ intersects $S(0,1)$
  and $S(1,0)$. We claim that $L'_1$ cannot intersect both $S(0,3)$
  and $S(3,0)$.  Suppose by contradiction that $L'_1$ intersects both
  $S(0,3)$ and $S(3,0)$. Then $L'_1$ is crossing $S(0,2)$ without
  crossing $L_2(0)$ or $L_2(1)$.  Similarly $L'_1$ is crossing
  $S(2,0)$ without crossing $L_0(0)$ or $L_0(1)$. Thus by superposing
  what happen in $S(0,2)$ and $S(2,0)$ in a square $S(j,k)$, we have
  that there are two crossing $1$-lines, a contradiction. Thus $L'_1$
  intersects at most one of $S(0,3)$ and $S(3,0)$.

  Suppose that $L'_1$ does not intersect $S(3,0)$.  Then the part of
  $T'$ situated right of $L_0(2)$ (left part on
  Figure~\ref{fig:triangleprooff}) is strictly included in
  $(S(0,0)\cup S(1,0) \cup S(2,0)\cup S(0,1)\cup S(1,1))$.  Thus this
  part of $T'$ contains at most $5\gamma_1 f$ faces.  Now consider the
  part of $T'$ situated on the left of $L_0(2)$ (right part on
  Figure~\ref{fig:triangleprooff}).  Let $y'$ be the intersection of
  $L'_1$ with $L_2$. Let $Q$ be the subpath of $L'_1$ between $y$ and
  $y'$. By definition of $L_2(1)$, there are no $2$-lines between
  $L_2$ and $L_2(1)$. So $Q$ cannot intersect a $2$-line on one of its
  interior vertices.  Thus $Q$ is crossing at most $\gamma_2$
  consecutive $0$-lines (that are not necessarily lines of type
  $L_0(k)$). Let $L'_0$ be the $\gamma_2+1$-th consecutive $0$-line
  that is on the left of $L_0(2)$ (counting $L_0(2)$).  Then the part
  of $T'$ situated on the left of $L_0(2)$ is strictly included in the
  region delimited by $L_0(2),L'_0,L_2,L_2(1)$, and thus contains at
  most $\gamma_2$ copies of a face of $G$.  Thus $T'$ contains at most
  $(\gamma_2+5\gamma_1)f$ faces.

  Symmetrically if $L'_1$ does not intersect $S(0,3)$ we have that
  $T'$ contains at most $(\gamma_0+5\gamma_1)f$ faces.  Then in any
  case, $T'$ contains at most $(\max(\gamma_0,\gamma_2)+5\gamma_1)f$
  faces and the lemma is true.
\end{proof}

\begin{figure}[!h]
\center
\input{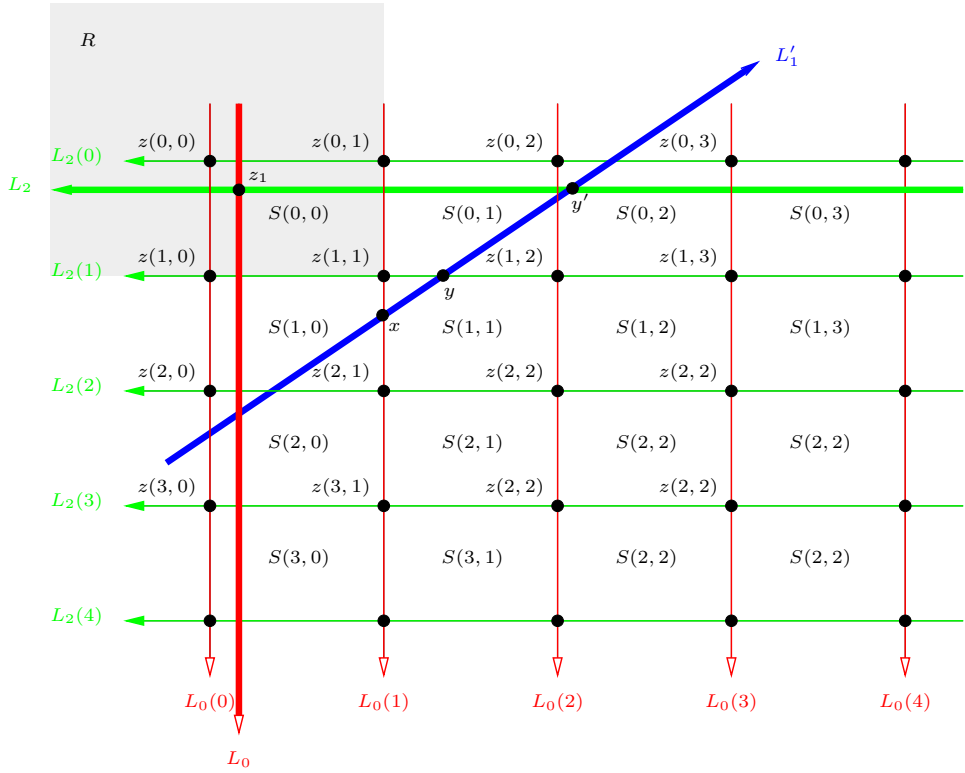}
\caption{Notations of the proof of Lemma~\ref{lem:boundedtriangle}.}
\label{fig:triangleprooff}
\end{figure}

The bound of Lemma~\ref{lem:boundedtriangle} is somehow sharp. In the
example of Figure~\ref{fig:exampletriangle}, the rectangle represent a
toroidal triangulation $G$ and the universal cover is partially
represented.  For each value of $k\geq 0$, there is a toroidal
triangulation $G$ with $n=4(k+1)$ vertices, where the gray region,
representing the region delimited by the three monochromatic lines
$L_{i}(v)$ contains $4\sum_{j=1}^{2k+1}+3(2k+2)=\Omega(n\times f)$
faces. Figure~\ref{fig:exampletriangle} represent such a triangulation
for $k=2$.

 \begin{figure}[!h]
 \center
 \includegraphics[scale=0.3]{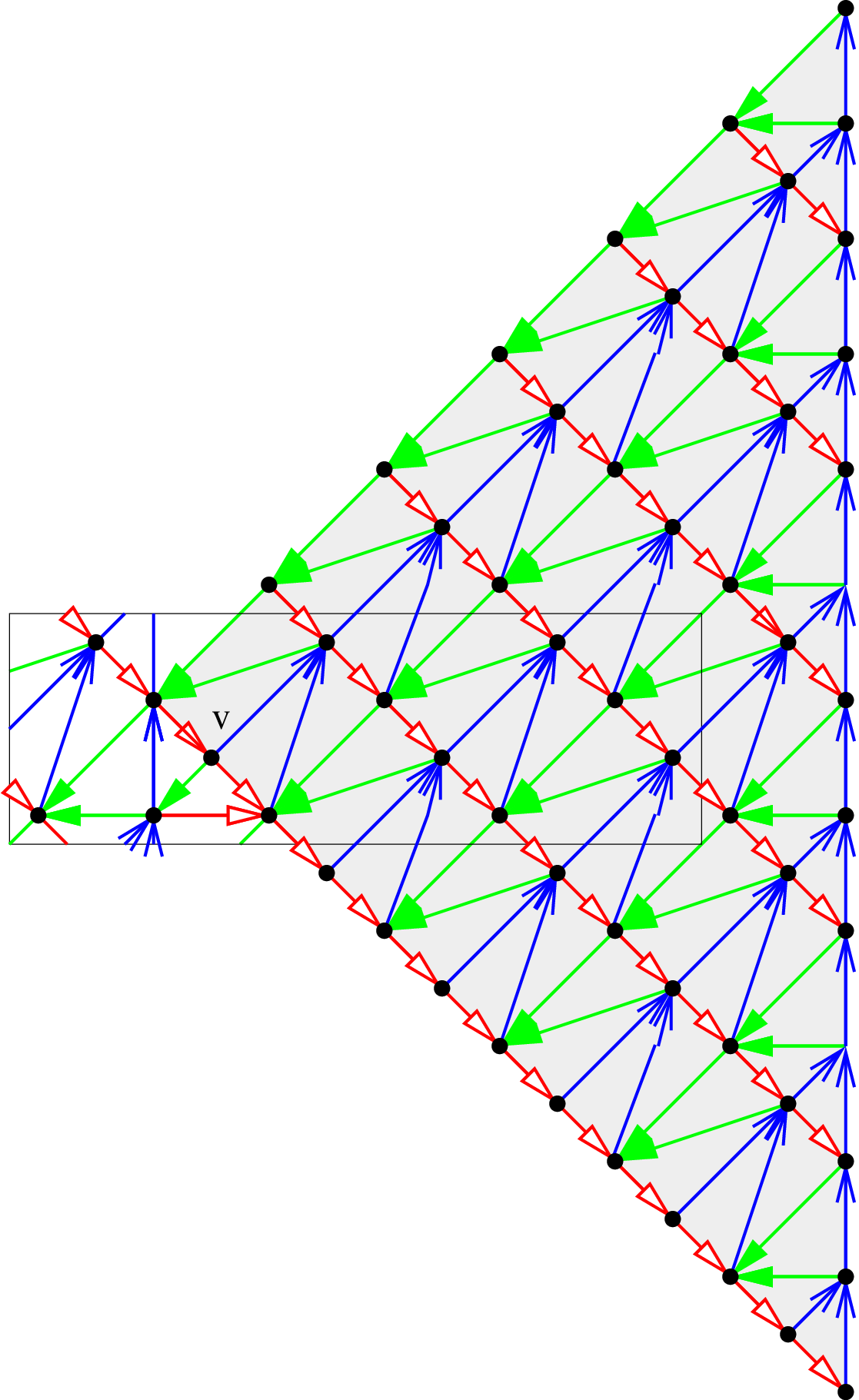}
 \caption{Example of a toroidal triangulation where the number of faces
   in the region delimited by the three monochromatic lines $L_{i}(v)$
   contains $\Omega(n\times f)$ faces.}
 \label{fig:exampletriangle}
 \end{figure} 

 For planar graphs the region vector method gives vertices that all lie on
 the same plane. This property is very helpful in proving that the
 position of the points on $P$ gives straight-line representations.  In the
 torus, things are more complicated as our generalization of the
 region vector method does not give coplanar points.  But
 Lemma~\ref{lem:sum} and~\ref{lem:boundedtriangle} show that all the
 points lie in the region situated between the two planes of equation
 $x+y+z=0$ and $x+y+z=t$, with $t=(5\min(\gamma_i)+
 \max(\gamma_i))f$. Note that $t$ is bounded by $6nf$ by
 Lemma~\ref{th:gamma} and this is independent from $N$. Thus from
 ``far away'' it looks like the points are coplanar and by taking $N$
 sufficiently large, non coplanar points are ``far enough'' from each
 other to enable the region vector method to give straight-line
 representations.  

 Let $N=t+n$.

\begin{lemma}
  \label{lem:shortregionedge}
  Let $u,v$ be two vertices such that $e_{i-1}(v)=uv$,
  $L_i=L_i(u)=L_i(v)$, and such that both $u$, $v$ are in the region
  $R(L_i,L'_i)$ for $L'_i$ a $i$-line consecutive to $L_i$. Then
  $v_{i+1}-u_{i+1}<|R(L_i,L'_i)|$ and $e_{i-1}(v)$ is going
  counterclockwise around the closed disk bounded by
  $\{e_{i-1}(v)\}\cup P_i(u)\cup P_i(v)$.
\end{lemma}

\begin{proof}
  Let $y$ be the first
  vertex of $P_i(v)$ that is also in $P_i(u)$.  Let $Q_u$
  (resp. $Q_v$) the part of $P_i(u)$ (resp. $P_i(v)$) between $u$
  (resp. $v$) and $y$.

  Let $D$ be the closed disk bounded by the cycle $C=(Q_v)^{-1}\cup
  \{e_{i-1}(v)\}\cup Q_u$. If $C$ is going clockwise around $D$, then
  $P_{i+1}(v)$ is leaving $v$ in $D$ and thus has to intersect $Q_u$
  or $Q_v$. In both cases, there is a cycle in $G_{i+1}\cup
  (G_i)^{-1}\cup (G_{i-1})^{-1}$, a contradiction to
  Lemma~\ref{lem:nocontractiblecycleuniversal}.  So $C$ is going
  clockwise around $D$.

  As $L_i(u)=L_i(v)$ and $L_{i-1}(u)=L_{i-1}(v)$, we have
  $v_{i+1}-u_{i+1}=d_{i+1}(v,u)$ and this is equal to the number of
  faces in $D$.  We have $D\subsetneq R(L_i,L'_i)$.  Suppose $D$
  contains two copies of a given face. Then, these two copies are on
  different sides of a $1$-line. By property (T1), it is not possible
  to have a $1$-line entering $D$.  So $D$ contains at most one copy
  of each face of $R(L_i,L'_i)$.
\end{proof}

\begin{lemma}
\label{lem:vectoriel}
For any face $F$ of $G^\infty$, incident to vertices $u,v,w$ (given in
counterclockwise order around $F$), the cross product
$\overrightarrow{vw}\wedge \overrightarrow{vu}$ has strictly positive
coordinates.
\end{lemma}

\begin{proof}
  Consider the angle labeling corresponding to the Schnyder wood.  By
  Lemma~\ref{lem:angle}, the angles at $F$ are labeled in
  counterclockwise order $0,1,2$.  As $\overrightarrow{uv}\wedge
  \overrightarrow{uw}=\overrightarrow{vw}\wedge
  \overrightarrow{vu}=\overrightarrow{wu}\wedge \overrightarrow{wv}$,
  we may assume that $u,v,w$ are such that $u$ is in the angle
  labeled $0$, vertex $v$ in the angle labeled $1$ and vertex $w$ in
  the angle labeled $2$.  The face $F$ is either a cycle completely
  directed into one direction or it has two edges oriented in one
  direction and one edge oriented in the other.
Let
$$\overrightarrow{X}=\overrightarrow{vw}\wedge
\overrightarrow{vu}= \begin{pmatrix}
  (w_1-v_1)(u_2-v_2)-(w_2-v_2)(u_1-v_1)\\
  -(w_0-v_0)(u_2-v_2)+(w_2-v_2)(u_0-v_0)\\
  (w_0-v_0)(u_1-v_1)-(w_1-v_1)(u_0-v_0)
\end{pmatrix}$$

By symmetry, we consider the following two cases:

 \begin{figure}[!h]
 \center
\begin{tabular}{ccc}
 \includegraphics[scale=.5]{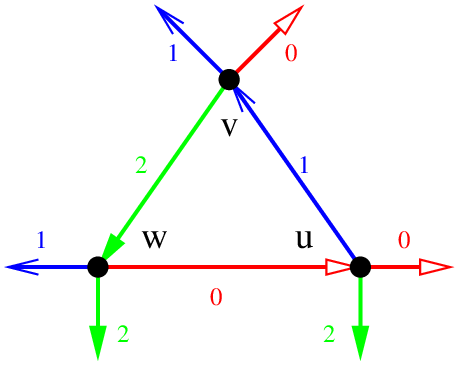}& \hspace{2em} &
 \includegraphics[scale=.5]{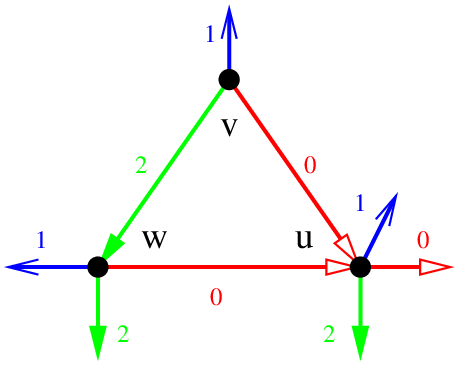}\\
(a)& &(b) \\
\end{tabular}
\caption{(a) case 1 and (b) case 2 of the proof of Lemma~\ref{lem:vectoriel}}
 \label{fig:vectoriel}
 \end{figure} 

 \noindent $\bullet$ \emph{Case 1: the edges of the face $F$ are in
   counterclockwise order $e_1(u)$, $e_2(v)$, $e_0(w)$ (see
   Figure~\ref{fig:vectoriel}.(a)).}

 We have $v\in P_1(u)$, so $v\in R_0(u)\cap R_2(u)$ and $u\in
 R_1^{\circ}(v)$ (as there is no edges oriented in two direction). By
 Lemma~\ref{lem:regionss}, we have $R_0(v)\subseteq R_0(u)$ and
 $R_2(v)\subseteq R_2(u)$ and $R_1(u)\subsetneq R_1(v)$. In fact the
 first two inclusions are strict as $u\notin R_0(v)\cup R_2(v)$. So by
 Lemma~\ref{lem:regionorder}, we have $v_0<u_0$, $v_2<u_2$,
 $u_1<v_1$. We can prove similar inequalities for the other pairs of
 vertices and we obtain $w_0<v_0<u_0$, $u_1<w_1<v_1$,
 $v_2<u_2<w_2$. By just studying the signs of the different terms
 occurring in the value of the coordinates of $\overrightarrow{X}$, it
 is clear that $\overrightarrow{X}$ as strictly positive
 coordinates. (For the first coordinates, it is easier if written in
 the following form $X_0=(u_1-w_1)(v_2-w_2)-(u_2-w_2)(v_1-w_1)$.)

  \noindent $\bullet$ \emph{Case 2: the edges of the face $F$ are in
    counterclockwise order $e_0(v)$, $e_2(v)$, $e_0(w)$.(see
    Figure~\ref{fig:vectoriel}.(b)).}

  As in the previous case, one can easily obtain the following
  inequalities: $w_0<v_0<u_0$, $u_1<w_1<v_1$, $u_2<v_2<w_2$ (the only
  difference with case $1$ is between $u_2$ and $v_2$).  Exactly like
  in the previous case, it is clear to see that $X_0$ and $X_2$ are
  strictly positive.  But there is no way to reformulate $X_1$ to have
  a similar proof.  Let $A=w_2-v_2$, $B=u_0-v_0$, $C=v_0-w_0$
  and $D=v_2-u_2$, so $X_1=AB-CD$ and $A,B,C,D$ are all strictly
  positive.

  Vertices $u,v,w$ are in the region $R(L_1,L'_1)$ for $L'_1$ a
  $1$-line consecutive to $L_1$. We consider two cases depending on
  equality or not between $L_1(u)$ and $L_1(v)$.

\noindent\emph{$\star$ Subcase 2.1: $L_1(u)=L_1(v)$.}

We have $X_1=A(B-D)+D(A-C)$.  

We have $B-D=(u_0+u_2)-(v_0+v_2)=(v_1-u_1)+(\sum u_i - \sum
v_i)$. Since $u\in P_0(v)$, we have $L_0(u)=L_0(v)$.  Suppose that
$L_2(u)=L_2(v)$, then by Lemma~\ref{lem:sum}, we have $\sum u_i = \sum
v_i$, and thus $B-D=v_1-u_1>0$. Suppose now that $L_2(u)\neq
L_2(v)$. By Lemmas~\ref{lem:sum} and~\ref{lem:boundedtriangle}, $\sum
u_i - \sum v_i>-t$.  By Lemma~\ref{lem:regionorder}, $v_1-u_1>
(N-n)|R(L_2(u),L_2(v))|\geq N-n$. So $B-D>N-n-t\geq 0$.

We have $A-C=(w_0+w_2)-(v_0+v_2)=(v_1-w_1)+(\sum w_i - \sum v_i)>\sum
w_i - \sum v_i$.  Suppose that $L_1(v)=L_1(w)$, then by
Lemma~\ref{lem:sum}, we have $\sum v_i = \sum w_i$ and thus
$A-C=v_1-w_1>0$. Then $X_1>0$.  Suppose now that $L_1(v)\neq L_1(w)$.
By Lemma~\ref{lem:shortregionedge}, $D=v_2-u_2<|R(L_1,L'_1)|$.  By
Lemma~\ref{lem:regionorder}, $A=w_2-v_2>(N-n)|R(L_1,L'_1)|$.  By
Lemma~\ref{lem:sum} and~\ref{lem:boundedtriangle}, $\sum w_i - \sum
v_i> -t$, so $A-C>-t$.  Then $X_1>(N-n-t)|R(L_1,L'_1)|>0$.

\noindent\emph{$\star$ Subcase 2.2: $L_1(u)\neq L_1(v)$.}

We have $X_1=B(A-C)+C(B-D)$.  

Suppose that $L_1(w)\neq L_1(v)$. Then $L_1(w)= L_1(u)$. By
Lemma~\ref{lem:shortregionedge} $e_0(w)$ is going counterclockwise
around the closed disk $D$ bounded by $\{e_0(w)\}\cup P_1(w)\cup
P_1(u)$. Then $v$ is inside $D$ and $P_1(v)$ has to intersect
$P_1(w)\cup P_1(u)$, so $L_1(v)=L_1(u)$, contradiction our
assumption. So $L_1(v)= L_1(w)$.

By Lemma~\ref{lem:regionorder}, $B=u_0-v_0>(N-n)|R(L_1,L'_1)|$.  We
have $A-C=(w_0+w_2)-(v_0+v_2)=(v_1-w_1)+(\sum w_i - \sum v_i)$. By
Lemma~\ref{lem:sum}, we have $\sum v_i = \sum w_i$ and thus
$A-C=v_1-w_1>0$.  By (the symmetric of)
Lemma~\ref{lem:shortregionedge}, $C=v_0-w_0<|R(L_1,L'_1)|$.  By
Lemma~\ref{lem:sum} ~\ref{lem:boundedtriangle},
$B-D=(u_0+u_2)-(v_0+v_2)=(v_1-u_1)+(\sum u_i - \sum v_i)>-t$.  So
$X_1>(N-n-t)|R(L_1,L'_1)|>0$.
\end{proof}

Let $G$ be an essentially 3-connected toroidal map.  Consider a
periodic mapping of $G^{\infty}$ embedded graph $H$ (finite or
infinite) and a face $F$ of $H$. Denote $(f_1,f_2,\ldots,f_t)$ the
counterclockwise facial walk around $F$. Given a mapping of the
vertices of $H$ in $\mathbb{R}^2$, we say that $F$ is \emph{correctly
  oriented} if for any triplet $1\le i_1 < i_2 < i_3 \le t$, the
points $f_{i_1}$, $f_{i_2}$, and $f_{i_3}$ form a counterclockwise
triangle.  Note that a correctly oriented face is drawn as a convex
polygon.

\begin{lemma}
\label{th:faceorientation}
  Let $G$ be an essentially 3-connected toroidal map given with a
  periodic mapping of $G^{\infty}$ such that every face of
  $G^{\infty}$ is correctly oriented. This mapping gives a
  straight-line representation of $G^{\infty}$.
\end{lemma}

\begin{proof}
  We proceed by induction on the number of vertices $n$ of $G$. Note
  that the theorem holds for $n=1$, so we assume that $n>1$.  Given
  any vertex $v$ of $G$, let $(u_0,u_1,\ldots,u_{d-1})$ be the
  sequence of its neighbors in counterclockwise order (subscript
  understood modulo $d$). Every face being correctly oriented, for
  every $i\in [0,d-1]$ the oriented angle (oriented counterclockwise)
  $(\overrightarrow{vu_i},\overrightarrow{vu_{i+1}}) < \pi$.  Let the
  winding number $k_v$ of $v$ be the integer such that $2k_v\pi =
  \sum_{i\in[0,d-1]} (\overrightarrow{vu_i},\overrightarrow{vu_{i+1}})
  $. It is clear that $k_v \ge 1$. Let us prove that $k_v=1$ for every
  vertex $v$.
\begin{claim}
For any vertex $v$, its winding number $k_v=1$.
\end{claim}
\begin{proofclaim}
  In a flat torus representation of $G$, we can sum up all the angles
  by grouping them around the vertices or around the faces.
  $$\sum_{v\in V(G)} \sum_{u_i\in N(v)}
(\overrightarrow{vu_i},\overrightarrow{vu_{i+1}})
=
\sum_{F\in F(G)} \sum_{f_i\in F} (\overrightarrow{f_if_{i-1}},\overrightarrow{f_if_{i+1}})$$
The face being correctly oriented, they form convex polygons. Thus the
angles of a face $F$ sum at $(|F|-2)\pi$.
$$\sum_{v\in V(G)} 2k_v\pi = \sum_{F\in F(G)} (|F|-2)\pi$$
$$\sum_{v\in V(G)} k_v =\frac{1}{2}\sum_{F\in F(G)} |F| -f$$
$$\sum_{v\in V(G)} k_v = m - f$$
So by Euler's formula $\sum_{v\in V(G)} k_v =n$, and thus $k_v=1$ for
every vertex $v$.
\end{proofclaim}

Let $v$ be a vertex of $G$ that minimizes the number of loops whose
ends are on $v$. Thus either $v$ has no incident loop, or every vertex
is incident to at least one loop.

Assume that $v$ has no incident loop.  Let $v'$ be any copy of $v$ in
$G^\infty$ and denote its neighbors $(u_0,u_1,\ldots,u_{d-1})$ in
counterclockwise order. As $k_v =1$, the points
$u_0,u_1,\ldots,u_{d-1}$ form a polygon $P$ containing the point $v'$
and the segments $[v',u_i]$ for any $i\in [0,d-1]$.  It is well known
that any polygon, admits a partition into triangles by adding some of
the chords. Let us call $O$ the  outerplanar graph with outer
boundary $(u_0,u_1,\ldots, u_{d-1})$, obtained by this
``triangulation'' of $P$. Let us now consider the toroidal map $G' =
(G\setminus \{v\})\cup O$ and its periodic embedding obtained from the
mapping of $G^\infty$ by removing the copies of $v$.  It is easy to
see that in this embedding every face of $G'$ is correctly oriented
(including the inner faces of $O$, or the faces of $G$ that have been
shortened by an edge $u_iu_{i+1}$). Thus by induction hypothesis, the
mapping gives a straight-line representation of $G'^\infty$.  It is
also a straight-line representation of $G^\infty$ minus the copies of
$v$ where the interior of each copy of the polygons $P$ are pairwise
disjoint and do not intersect any vertex or edge. Thus one can add the
copies of $v$ on their initial positions and add the edges with their
neighbors without intersecting any edge. The obtained drawing is thus
a straight-line representation of $G^\infty$.

Assume now that every vertex is incident to at least one loop.  Since
these loops are non-contractible and do not cross each other, they
form homothetic cycles.  Thus $G$ is as depicted in
Figure~\ref{fig:polygonG}, where the dotted segments stand for edges
that may be in $G$ but not necessarily.  Since the mapping is periodic
the edges corresponding to loops of $G$ form several parallel lines,
cutting the plane into infinite strips.  Since for any $1\leq i\leq
n$, $k_{v_i}=1$, a line of copies of $v_i$ divides the plane, in such
a way that their neighbors which are copies of $v_{i-1}$ and their
neighbors which are copies of $v_{i+1}$ are in distinct
half-planes. Thus adjacent copies of $v_i$ and $v_{i+1}$ are on two
lines bounding a strip. Then one can see that the edges between copies
of $v_i$ and $v_{i+1}$ are contained in this strip without
intersecting each other.  Thus the obtained mapping of $G^\infty$ is a
straight-line representation.
\end{proof}

\begin{figure}[!h]
\begin{center}
\input{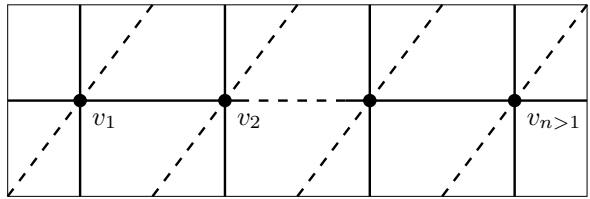}
\caption{The graph $G$ if every vertex is incident to a loop.}
\label{fig:polygonG}
\end{center}
\end{figure}

A plane is \emph{positive} if it has equation $\alpha x +\beta y +
\gamma z=0$ with $\alpha,\beta,\gamma \geq 0$.

\begin{theorem}
\label{th:straightline}
If $G$ is a toroidal triangulation given with a Schnyder wood, and
$V^\infty$ the set of region vectors of vertices of $G^\infty$. Then
the projection of $V^\infty$ on a positive plane gives a straight-line
representation of $G^\infty$.
\end{theorem}

\begin{proof}
  Let $\alpha,\beta,\gamma \geq 0$ and consider the projection of
  $V^\infty$ on the plane $P$ of equation $\alpha x +\beta y + \gamma
  z=0$. A normal vector of the plane is given by the vector
  $\overrightarrow{n}=(\alpha,\beta,\gamma)$. Consider a face $F$ of
  $G^\infty$. Suppose that $F$ is incident to vertices $u,v,w$ (given
  in counterclockwise order around $F$).  By
  Lemma~\ref{lem:vectoriel}, $(\overrightarrow{uv}\wedge
  \overrightarrow{uw}).\overrightarrow{n}$ is positive. Thus the
  projection of the face $F$ on $P$ is correctly oriented. So by
  Lemma~\ref{th:faceorientation}, the projection of $V^\infty$ on
  $P$ gives a straight-line representation of $G^\infty$.
\end{proof}


Theorems~\ref{th:existence} and~\ref{th:straightline} implies
Theorem~\ref{cor:straightline}. Indeed, any toroidal graph $G$ can be
transformed into a toroidal triangulation $G'$ by adding a linear
number of vertices and edges and such that $G'$ is simple if and only
if $G$ is simple (see for example the proof of Lemma~2.3
of~\cite{Moh96}).  Then by Theorem~\ref{th:existence},
$G'$ admits a
Schnyder wood.
By Theorem~\ref{th:straightline}, the projection of the set of region
vectors of vertices of $G'^\infty$ on a positive plane gives a
straight-line representation of $G'^\infty$.  The grid where the
representation is obtained can be the triangular grid, if the
projection is done on the plane of equation $x+y+z=0$, or the square
grid, if the projection is done on one of the plane of equation $x=0$,
$y=0$ or $z=0$.
By Lemma~\ref{th:gamma}, and the choice of $N$, the obtained mapping
is a periodic mapping of $G^\infty$ with respect to non collinear
vectors $Y$ and $Y'$ where the size of these vectors is in $\mathcal
O(\gamma^2 n^2)$ with
$\gamma\leq n$ in general and $\gamma=1$ if the graph is simple and
the Schnyder wood obtained by Theorem~\ref{th:schnydersimple}.
By Lemma~\ref{lem:edgesbounded}, the
length of the edges in this representation are in $\mathcal O(n^3)$ in
general and in $\mathcal O(n^2)$ if the graph is simple.  When the
graph is not simple, there is a non contractible cycle of length $1$
or $2$ and thus the size of one of the two vectors $Y$, $Y'$ is in
$\mathcal O(n^3)$. Thus the grid obtained in
Theorem~\ref{cor:straightline} has size in $\mathcal O(n^3)\times
\mathcal O(n^4)$ in general and $\mathcal O(n^2)\times\mathcal O(n^2)$
if the graph is simple.

The method presented here gives a polynomial algorithm to obtain flat torus
straight-line representation of any toroidal maps in polynomial size
grids. Indeed, all the proofs leads to polynomial
algorithms, even the proof of Theorem~\ref{th:fij}~\cite{Fij} which  uses
results from Robertson and Seymour~\cite{RS86} on disjoint paths
problems.

It would be nice to extend Theorem~\ref{th:straightline}
to obtain convex straight-line representation for essentially
3-connected toroidal maps.

\section{Conclusion}

We have proposed a generalization of Schnyder woods to toroidal maps
with application to graph drawing. Along these lines, several
questions were raised. We recall some briefly:
\begin{itemize}
\item 
Does the set of Schnyder woods of a given toroidal map has a kind of
lattice structure ?
\item Does any simple toroidal triangulation
admits a Schnyder wood where the set of edges of each color induces a
connected subgraph ?
\item Is it possible to use Schnyder woods to embed the universal
  cover of a toroidal map on rigid or coplanar orthogonal surfaces ?
\item Which toroidal maps admits (primal-dual) contact
representation by (homothetic) triangles in a flat torus ?
\item Can geodesic embeddings be
used to obtain convex straight-line representation for essentially
3-connected toroidal maps ?
\end{itemize}

The guideline of Castelli Aleardi et al.~\cite{CFL09} to generalize
Schnyder wood to higher genus was to preserve the tree structure of
planar Schnyder woods and to use this structure for efficient
encoding. For that purpose they introduce several special rules (even
in the case of genus $1$).  Our main guideline while working on this
paper was that the surface of genus $1$, the torus, seems to be the
perfect surface to define Schnyder woods. Euler's formula gives
exactly $m=3n$ for toroidal triangulations. Thus a simple and
symmetric object can be defined by relaxing the tree constraint.  For
genus 0, the plane, there are not enough edges in planar
triangulations to have outdegree three for every vertex. For higher
genus (the double torus, ...) there are too many edges in
triangulations.  An open problem is to find what would be the natural
generalization of our definition of toroidal Schnyder woods to higher
genus.

The results presented here motivated~Castelli~Aleardi
and~Fusy~\cite{CF12} to developed direct methods to obtain straight
line representations for toroidal maps. They manage to generalize
planar canonical ordering to the cylinder to obtain straight-line
representation of simple toroidal triangulations in grids of size
$\mathcal O(n)\times \mathcal O(n^2)$ thus improving the size of our
grid that is $\mathcal O(n^2)\times \mathcal O(n^2)$ in the case of a
simple toroidal map. It should be interesting to investigate further
the links between the two methods as canonical ordering are strongly
related to Schnyder woods.

Planar Schnyder woods appear to have many applications in various
areas like enumeration~\cite{Bon05}, compact coding~\cite{PS06},
representation by geometric objects~\cite{FOR94, GLP11}, graph
spanners~\cite{BGHI10}, graph drawing~\cite{Fel01, Kan96}, etc. In
this paper we use a new definition of Schnyder wood for graph drawing
purpose, it would also be interesting to see if it can be used in
other computer science domains.

\section*{Acknowledgments}
The authors thank Nicolas Bonichon, Luca Castelli Aleardi and Eric
Fusy for fruitful discussions about this work. They also thank a
 student Chloé Desdouits for developing a software to visualize
orthogonal surfaces.


\begin{thebibliography}{00}

\bibitem{Bon05} N.~Bonichon, A bijection between realizers of maximal plane graphs and
pairs of non-crossing dyck paths, \emph{Discrete Mathematics} 298 
(2005) 104-114.

\bibitem{BGHI10} N. Bonichon, C. Gavoille, N. Hanusse, D. Ilcinkas, Connections between Theta-
Graphs, Delaunay Triangulations, and Orthogonal Surfaces. WG10 (2010).

\bibitem{CFL09} L.~Castelli Aleardi, E.~Fusy, T.~Lewiner, Schnyder
  woods for higher genus triangulated surfaces, with applications to
  encoding, \emph{Discrete and Computational Geometry} 42 (2009)
  489-516.

\bibitem{CF12} L.~Castelli Aleardi, E.~Fusy, Canonical ordering for
  triangulations on the cylinder, with applications to periodic
  straight-line drawings, EuroCG'12 (2012).

\bibitem{CEG11} E.~Chambers, D.~Eppstein, M.~Goodrich, M.~L\"offler,
  Drawing graphs in the plane with a prescribed outer face and
  polynomial area, \emph{Lecture Notes in Computer Science} 6502 (2011) 129-140.

\bibitem{DGK11} C.~Duncan, M.~Goodrich, S.~Kobourov, Planar drawings
  of higher-genus graphs, \emph{Journal of Graph Algorithms and Applications}
  15 (2011) 13-32.

\bibitem{DM41} B.~Dushnik, E.W.~Miller, Partially ordered sets,
  American Journal of Mathematics 63 (1941) 600-610.

\bibitem{Fel01} S. Felsner, Convex Drawings of Planar Graphs and the
   Order Dimension of 3-Polytopes, {\it Order} 18 (2001) 19-37.

\bibitem{Fel03} S. Felsner, Geodesic Embeddings and Planar Graphs,
   {\it Order} 20 (2003) 135-150.

\bibitem{Fel04} S. Felsner, Lattice structures from planar graphs,
{\it Electron. J. Combin.} 11 (2004).

\bibitem{FOR94} H. de Fraysseix, P. Ossona de Mendez, P. Rosenstiehl,
  On Triangle Contact Graphs, {\it Combinatorics, Probability and
   Computing} 3 (1994) 233-246.

 \bibitem{FZ08} S. Felsner, F. Zickfeld, Schnyder Woods and Orthogonal
   Surfaces, Discrete Comput Geom 40 (2008) 103-126.

 \bibitem{Fel-book} S. Felsner, \emph{Geometric Graphs and
     Arrangements}, Vieweg, 2004.

 \bibitem{Fij} G. Fijavz, personal communication (2011).

\bibitem{FO01} H. de Fraysseix, P. Ossona de Mendez, On topological
   aspects of orientations, {\it Discrete Mathematics} 229 (2001) 57-72.

\bibitem{GLP11}
D. Gon\c{c}alves, B. L\'ev\^eque, A. Pinlou, Triangle contact representations
and duality, to appear in \emph{Discrete and Computational Geometry}.

\bibitem{GLP11b} D. Gon\c{c}alves, B.  L\'ev\^eque, A. Pinlou, Homothetic
   triangle representations of planar graphs, manuscript, 2011.

\bibitem{Kan96} G.~Kant, Drawing planar graphs using the canonical
ordering, \emph{Algorithmica} 16 (1996) 4-32.

\bibitem{Kra07} J. Kratochv\'il, Bertinoro Workshop on Graph Drawing,
  2007.

\bibitem{Mil02} E. Miller, Planar graphs as minimal resolutions of
   trivariate monomial ideals, {\it Documenta Mathematica} 7 (2002) 43-90.

 \bibitem{Moh96} B. Mohar, Straight-line representations of maps on
   the torus and other flat surfaces, Discrete Mathematics 155 (1996)
   173-181.

\bibitem{Rose89}
P. Rosenstiehl, Embedding in the plane with orientation constraints:
The angle graph, Annals New York Academy of Sciences, 1989.

\bibitem{MR98}
B. Mohar, P. Rosenstiehl,
Tessellation and visibility representations of maps on the torus,
Discrete Comput. Geom. 19 (1998) 249-263.

\bibitem{PS06} D.~Poulalhon, G.~Schaeffer, Optimal coding and sampling
  of triangulations, \emph{Algorithmica} 46 (2006) 505-527.

\bibitem{RS86} N.~Robertson, P.D~Seymour, Graph minors. VI. Disjoint
  paths across a disc,\emph{ Journal of Combinatorial Theory B}, 41
  (1986) 115-138.

\bibitem{Sch89} W. Schnyder, Planar graphs and poset dimension, {\it Order}
  5 (1989) 323-343.

\end{thebibliography}
\end{document}